\RequirePackage{rotating}
\documentclass[acmsmall]{acmart}
\usepackage{fdiaz}
\usepackage{colortbl}

\usepackage{centernot}

\usepackage{subcaption}

\usepackage{amsthm}

\theoremstyle{remark}
\newtheorem{case}{Case}

\setcopyright{none}

\acmDOI{10.475/123_4}

\acmISBN{123-4567-24-567/08/06}

\acmJournal{TORS}
\acmYear{2023}
\copyrightyear{2023}

\acmArticle{4}
\acmPrice{15.00}

\newcommand{\docset}[0]{\fdset{D}}
\newcommand{\docsets}[0]{\docset^+}
\newcommand{\docsetsAll}[0]{\docset^*}

\newcommand{\numdocs}[0]{n}
\newcommand{\topk}[0]{\tilde{\numdocs}}
\newcommand{\topkranking}[0]{\tilde{\ranking}}

\newcommand{\qrels}[0]{y}
\newcommand{\relset}[0]{\fdset{R}}
\newcommand{\prelset}[0]{\relset'}

\newcommand{\numrel}[0]{m}
\newcommand{\numusers}[0]{\numrel^+}
\newcommand{\ranking}[0]{\pi}
\newcommand{\wcranking}[0]{\check{\ranking}}
\newcommand{\dranking}[0]{\ranking^*}

\newcommand{\rankings}[0]{S_{\numdocs}}

\newcommand{\metric}[0]{\mu}
\newcommand{\metricUpper}[0]{\textrm{U}_\mu}
\newcommand{\metricLower}[0]{\textrm{L}_\mu}

\newcommand{\exposure}[0]{e}
\newcommand{\normalization}[0]{z}

\newcommand{\pmetric}[0]{\eta}

\newcommand{\leximin}[0]{\text{leximin}}

\newcommand{\lexirecall}[0]{\text{LR}}

\newcommand{\lexirecallpref}[0]{\succ_{\lexirecall}}
\newcommand{\lexirecalleq}[0]{=_{\lexirecall}}

\newcommand{\numcollapsed}[0]{C}
\newcommand{\prob}[0]{\text{Pr}}
\newcommand{\probthm}[0]{\text{\emph{Pr}}}

\newcommand{\relat}[0]{\text{Rel}}
\newcommand{\RPall}[0]{\mathcal{P}_{\numrel}^{\numdocs}}

\newcommand{\leximinpref}[0]{\succ_{\text{leximin}}}
\newcommand{\leximineq}[0]{=_{\text{leximin}}}

\newcommand{\leximinprefthm}[0]{\succ_{\text{\emph{leximin}}}}
\newcommand{\leximineqthm}[0]{=_{\text{\emph{leximin}}}}

\newcommand{\lexirecalleqthm}[0]{=_{\text{\emph{LR}}}}

\newcommand{\tsepref}[0]{\succ_{\text{WC}}}

\newcommand{\rankPositions}[0]{\Pintegers}
\newcommand{\NNintegers}[0]{\mathbb{Z}^*}
\newcommand{\Pintegers}[0]{\mathbb{Z}^+}
\newcommand{\NNreals}[0]{\mathbb{R}^*}
\newcommand{\Preals}[0]{\mathbb{R}^+}

\newcommand{\rr}[0]{\mathrm{RR}}

\newcommand{\ndcg}[0]{\mathrm{NDCG}}
\newcommand{\ndcgTen}[0]{\ndcg_{10}}

\newcommand{\rbp}[0]{\mathrm{RBP}}
\newcommand{\rbpGamma}[0]{\gamma}

\newcommand{\expectation}[2]{\mathbb{E}_{#1}\left[#2\right]}

\newcommand{\ap}[0]{\text{AP}}
\newcommand{\esl}[0]{\text{SL}}
\newcommand{\recallerror}[0]{\text{RE}}

\newcommand{\optimalRankings}[0]{\rankings^*}

\newcommand{\ndocs}[0]{n}

\newcommand{\rlx}[0]{\pi}
\newcommand{\rly}[0]{\pi'}
\newcommand{\rlset}[0]{S_{\ndocs}}
\newcommand{\rldata}[0]{\tilde{S}_{\ndocs}}

\newcommand{\RP}[0]{p}
\newcommand{\RPx}[0]{\RP}
\newcommand{\RPy}[0]{\RP'}
\newcommand{\invRP}[0]{\overline{\RP}}

\newcommand{\GP}[0]{q}
\newcommand{\GPx}[0]{\GP}
\newcommand{\GPy}[0]{\GP'}

\newcommand{\user}[0]{u}
\newcommand{\wcuser}[0]{\tilde{\user}}
\newcommand{\users}[0]{\mathcal{U}}
\newcommand{\provider}[0]{v}
\newcommand{\providers}[0]{\mathcal{V}}

\newcommand{\UP}[0]{\varrho}
\newcommand{\UPx}[0]{\UP}
\newcommand{\UPy}[0]{\UP'}

\newcommand{\PP}[0]{\phi}
\newcommand{\PPx}[0]{\PP}
\newcommand{\PPy}[0]{\PP'}

\newcommand{\ident}[0]{\text{I}}

\newcommand{\recall}[0]{\text{R}}
\newcommand{\recallK}[0]{\recall_{1000}}

\newcommand{\rprecision}[0]{\text{RP}}

\DeclareMathOperator{\tse}{TSE}

\DeclareMathOperator{\metricworst}{WC}
\DeclareMathOperator{\sort}{sort}

\newcommand{\isep}{\mathrel{{.}\,{.}}\nobreak}
\newcommand{\relsubsets}[0]{\mathcal{W}}
\newcommand{\relsubsetsGT}[1]{\relsubsets_{>#1}}
\newcommand{\relsubsetsGEQ}[1]{\relsubsets_{\ge#1}}
\newcommand{\relsubset}[0]{w}
\newcommand{\relsubsetsNewMin}[0]{\relsubsetsGEQ{k}-\relsubsetsGT{k}}

\newcommand\rankeq{\mathrel{\stackrel{\makebox[0pt]{\mbox{\normalfont\tiny rank}}}{=}}}

\usepackage[object=vectorian]{pgfornament} 
\usepackage{lipsum,tikz}

\begin{document}
\title{Recall, Robustness, and Lexicographic Evaluation}

\author{Fernando Diaz}
\orcid{0000-0003-2345-1288}
\authornote{work done at Google}
\affiliation{
\institution{Carnegie Mellon University}
\city{Pittsbugh}
\state{PA}
\country{USA}
}
\email{diazf@acm.org}

\author{Michael D. Ekstrand}
\orcid{0000-0003-2467-0108}
\affiliation{
\institution{Drexel University}
\city{Philadelphia}
\state{PA}
\country{United States}
}
\email{mdekstrand@drexel.edu}

\author{Bhaskar Mitra}
\orcid{0000-0002-5270-5550}
\affiliation{
\institution{Microsoft}
\city{Montr\'eal}
\state{QC}
\country{Canada}
}
\email{bmitra@microsoft.com}

\begin{abstract}
Although originally developed to evaluate \textit{sets} of items, recall is often used to evaluate \textit{rankings} of items, including those produced by recommender, retrieval, and other machine learning systems.  The application of recall without a formal evaluative motivation has led to criticism of recall as a vague or inappropriate measure.  In light of this debate, we reflect on the measurement of recall in rankings from a formal perspective.  Our analysis is composed of three  tenets: recall, robustness, and lexicographic evaluation.  First, we formally define `recall-orientation' as the sensitivity of a metric to a user interested in finding every relevant item.  Second, we analyze recall-orientation from the perspective of robustness with respect to possible content consumers and providers, connecting recall to recent conversations  about fair ranking.  Finally, we extend this conceptual and theoretical treatment of recall by developing a practical preference-based evaluation method based on lexicographic comparison.  Through extensive empirical analysis across multiple recommendation and retrieval tasks, we establish that our new evaluation method, lexirecall, has convergent validity (i.e., it is correlated with existing recall metrics) and exhibits substantially higher sensitivity in terms of discriminative power and stability in the presence of missing labels.  Our conceptual, theoretical, and empirical analysis substantially deepens our understanding of recall and motivates its adoption through connections to robustness and fairness.
\end{abstract}
\maketitle
\section{Introduction}
\label{sec:introduction}

Researchers use `recall' to evaluate rankings across a variety of retrieval, recommendation \cite{zhao:offline-recsys-eval}, and machine learning tasks \cite{sajjadi:geneval,rudinger-etal-2015-script,flach:pr-curves}. `Recall at $k$ items' ($\recall_{k}$) and R-Precision ($\rprecision$) are popular  metrics used for measuring recall in rankings.
Since the beginning of the ACM Conference on Recommender Systems, on average one third of full papers  measure recall in experiments (Figure \ref{fig:recsys-recall}).
While there is a colloquial interpretation of recall as measuring coverage (as it might be rightfully interpreted in set retrieval), the research community is far from a principled understanding of recall metrics for rankings.  Nevertheless, authors continue to informally refer to evaluation metrics as more or less `recall-oriented' or `precision-oriented' without a formal definition of what this means or quantifying how existing metrics relate to these constructs \cite{kazai:xcg,montazeralghaem:rl-rf,dai:ltr-resources,mohammad:dynamic-shard-cutoff,li:reqrec,diaz:crm}.

\begin{figure}
\centering
\includegraphics[width=\textwidth]{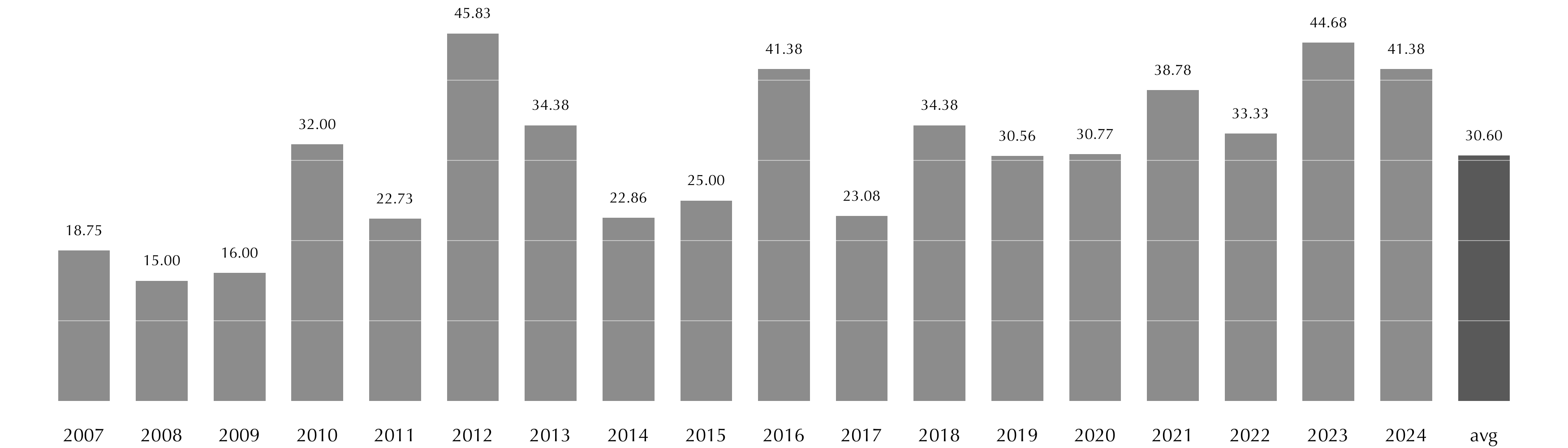}
\caption{Percentage of full papers published at the ACM Conference on Recommender Systems that measure recall in experiments.  Percentages are macro-averaged across years to control for growth in the number of full papers in the conference proceedings.  Details can be found in Appendix \ref{app:recsys-recall}.}\label{fig:recsys-recall}
\end{figure}
Given the prevalence of recall as a metric in recommender system research, understanding recall conceptually, theoretically, and empirically is fundamental to sound evaluation.  Indeed, the lack of a principled understanding of or motivation for recall has caused some to question whether recall is useful as a construct at all.  \citet{herlocker:recsys-eval} suggests that the sparse nature of recommendation data may make recall inappropriate.    \citet{cooper:esl} argues that recall-orientation is inappropriate because user search satisfaction depends on the number of items the user is looking for, which may be fewer than \textit{all} of the relevant items.   \citet{zobel:against-recall} refute several informal justifications for recall: persistence (the depth a user is willing to browse), cardinality (the number of relevant items found), coverage (the number of user intents covered), density (the rank-locality of relevant items), and totality (the retrieval of all relevant items).  So, while many ranking experiments compute recall metrics, precisely what and why we are measuring remains unclear.

In this light, we approach the measurement of recall in rankings from a formal perspective, with an objective of proposing a new interpretation of recall with precise conceptual and theoretical grounding.
Our analysis is composed of three interrelated tenets: recall, robustness, and lexicographic evaluation.
First, we consider recall an essentially contested construct \cite{gallie:essentially-contested-concepts,jacobs:measurement-and-fairness}: it is a high level construct believed to be important but with different, conflicting interpretations, as suggested by \citet{zobel:against-recall}.  By adopting the interpretation of recall as finding the totality of relevant items, we formally define `recall-orientation' as sensitivity to movement of the bottom-ranked relevant item.  Although simple, this definition of recall connects to both early work in position-based evaluation as well as recent work in technology-assisted review.  Moreover, by formally defining recall orientation, we can design a new metric, total search efficiency, that precisely measures recall.  Second, we consider robustness another essentially contested construct, again with different, conflicting interpretations \cite{drenkow:robustness}.  By adopting an interpretation of robustness as the effectiveness for the worst-off user, we can connect it to our notion of recall-orientation.
We demonstrate that recall---and totality specifically---is aligned with  worst-case robustness.  Finally, we extend this conceptual and theoretical treatment of recall by developing a practical preference-based evaluation method based on lexicographic comparison.  We present a conceptual relationship between recall, robustness, and lexicographic evaluation in Figure \ref{fig:overview}.  Through extensive empirical analysis across various retrieval and recommendation tasks, we establish that our new evaluation method, lexicographic recall or \textit{lexirecall}, is correlated with existing recall metrics but exhibits substantially higher discriminative power and stability in the presence of missing labels.  While conceptually and theoretically grounded in notions of robustness and fairness, lexirecall pragmatically  is appropriate when we are concerned with understanding system behavior for users who are focused on finding all relevant items in the same way that experiments adopt reciprocal rank as a high-precision metric.  This paper, by focusing on worst-case performance,  also complements existing work that has focused on average case \cite{diaz:rpp} and best-case \cite{diaz:lexiprecision} preference-based evaluation.   Our conceptual, theoretical, and empirical analysis substantially deepens our understanding of recall as a construct and motivates its adoption through connections to robustness and fairness constructs.

\begin{figure}
\includegraphics*[width=0.65\textwidth]{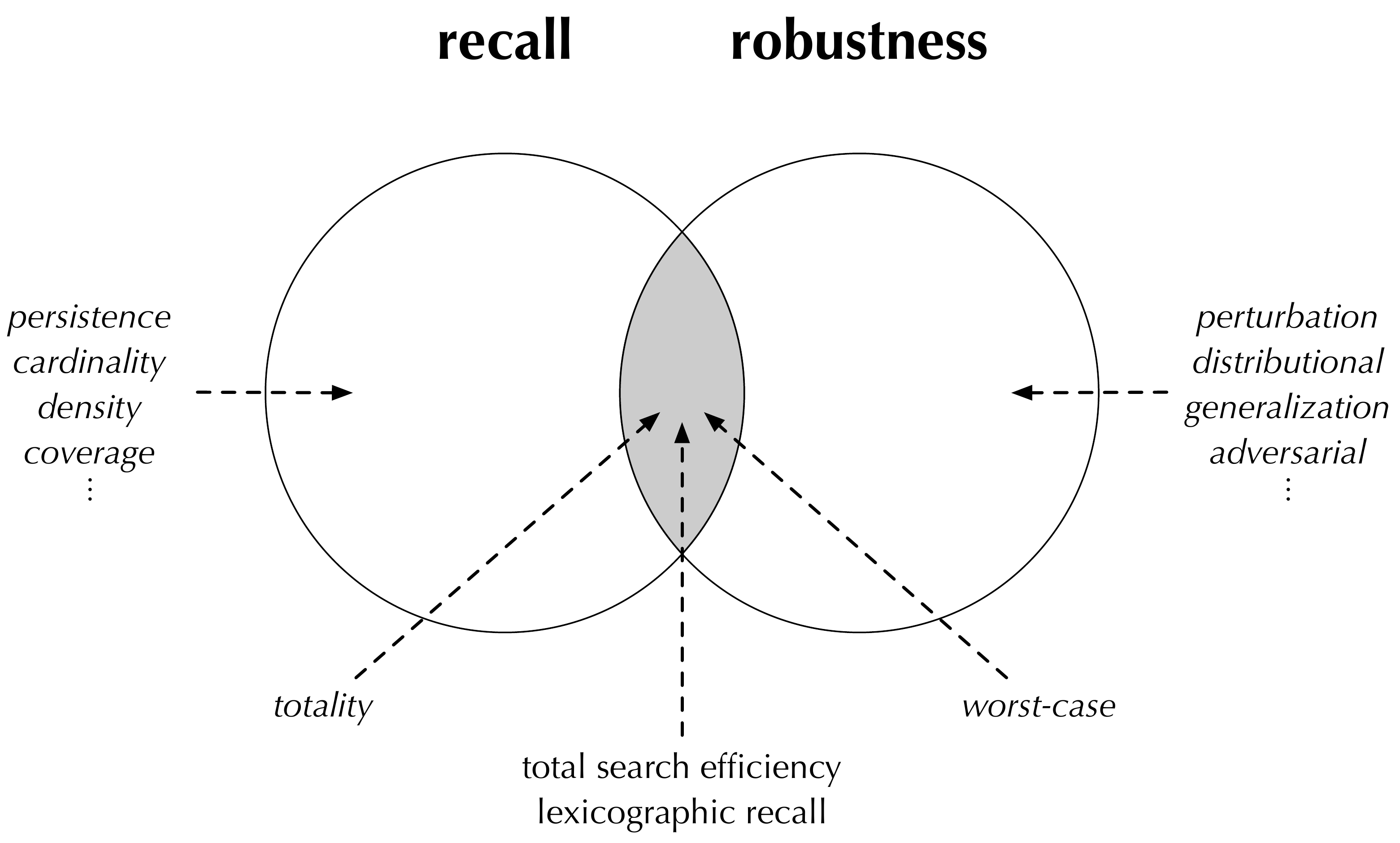}
\caption{Both recall and robustness as a theoretical construct can be conceptualized in multiple ways \cite{zobel:against-recall,drenkow:robustness}.  These two constructs are equivalent when conceptualized as totality (recall) and worst-case analysis (robustness).   Lexicographic recall is located in this intersection. }\label{fig:overview}
\end{figure}
\section{Preliminaries}
\label{sec:preliminaries}
We begin by defining our core concepts and notation in order to provide a clear foundation for our analysis.  While many of these concepts will be familiar to those with a background in the ranking evaluation, we adopt a specific mathematical framework that will be important when proving properties of recall, robustness, and lexicographic evaluation.

We consider ranking systems that are designed to satisfy users\footnote{We adopt \textit{user} as a general term for searcher in the information retrieval context and content consumer in the multi-stakeholder recommendation context.} with with some information need, broadly construed.  This information need may be specific or vague, and may or may not be explicit even in the user's own mind (e.g. ``entertain me for the evening'' can be considered an information need for home-page recommendations on a video streaming service; ``what happened yesterday that I should know about?'' for a daily e-mail of recommended new articles from a news publisher).  Although users in many recommendation scenarios may not seem to have an information need, they do, at any point in time, have a latent order over items in the catalog.  This is true in relatively passive scenarios like radio-like streaming audio where a user, although not explicitly expressing or thinking about their preferences, nevertheless reveals them through their consumption behavior.

While a user's information need is never directly revealed to the ranking system, the user expresses  an observable  \textit{request} to the ranking system.  A request can include information provided by the user either explicitly  (e.g., a query or question for text-based search; an application or site  context for recommendation) or implicitly (e.g., geo-location for text-based search; engagement or session history for recommendation \cite{ludewig:evaluating-sbr}) or both \cite{zamani:search-rec,sontag:personalization}.

A system attempts to satisfy an information need by ranking all of the items in a corpus $\docset$.  A corpus might consist of text documents (e.g., a web crawl) or cultural media (e.g., a music catalog).  And so, if $\numdocs=|\docset|$, a ranking system is a function that, given a request, produces a permutation of the $\numdocs$ items in the collection.  As such, the space of possible system outputs is the set of all permutations of $\numdocs$ items, also known as the symmetric group of degree $\numdocs$ or $\rankings$.

The objective of ranking evaluation is to determine the quality of a permutation $\ranking\in\rankings$ for the user. In the remainder of this section, we will detail precisely how we do this.\footnote{While ranking does not cover all recommendation scenarios, it covers many, including personalized search, home-page recommendations, related product recommendations, and social media timeline ranking, and is frequently considered in research; further, non-ranking interfaces such as music or video streams may well be implemented as a final selection or transformation of an underlying (partial) ranking, and therefore better understanding of ranking evaluation may be applicable even to recommendation surfaces that do not directly expose the ranking to the user.}

\subsection{Relevance}
\label{sec:preliminaries:relevance}
The relevance of an item refers to its value with respect to a user's information need.  In this work, we focus on binary relevance, an approach regularly used in information retrieval and recommender system literature, especially when measuring recall.\footnote{We discuss how our analysis extends in ordinal grades and preferences in Section \ref{sec:robustness:searchers:grades}.}
Let  $\relset\subset\docset$ be the set of items labeled relevant to the request where $\numrel=|\relset|$.

Given a ranking $\ranking$,  then we define $\RPx$ as the $\numrel\times 1$ vector where $\RPx_i$ is the position of the $i$th ranked relevant item in $\ranking$.  We present an example of how to construct $\RPx$ in Figure \ref{fig:notation}.  We will use $\invRP$ to represent the $\numrel\times 1$ vector where $\invRP_i\in\docset$ is the identity of the $i$th ranked relevant item. There are a total of $\ndocs\choose\numrel$ unique $\RPx$ and  each unique $\RPx$ corresponds to a subset of $\numcollapsed=\numrel!(\numdocs-\numrel)!$ unique permutations in $\rankings$.

\begin{figure}
\includegraphics[width=0.65\linewidth]{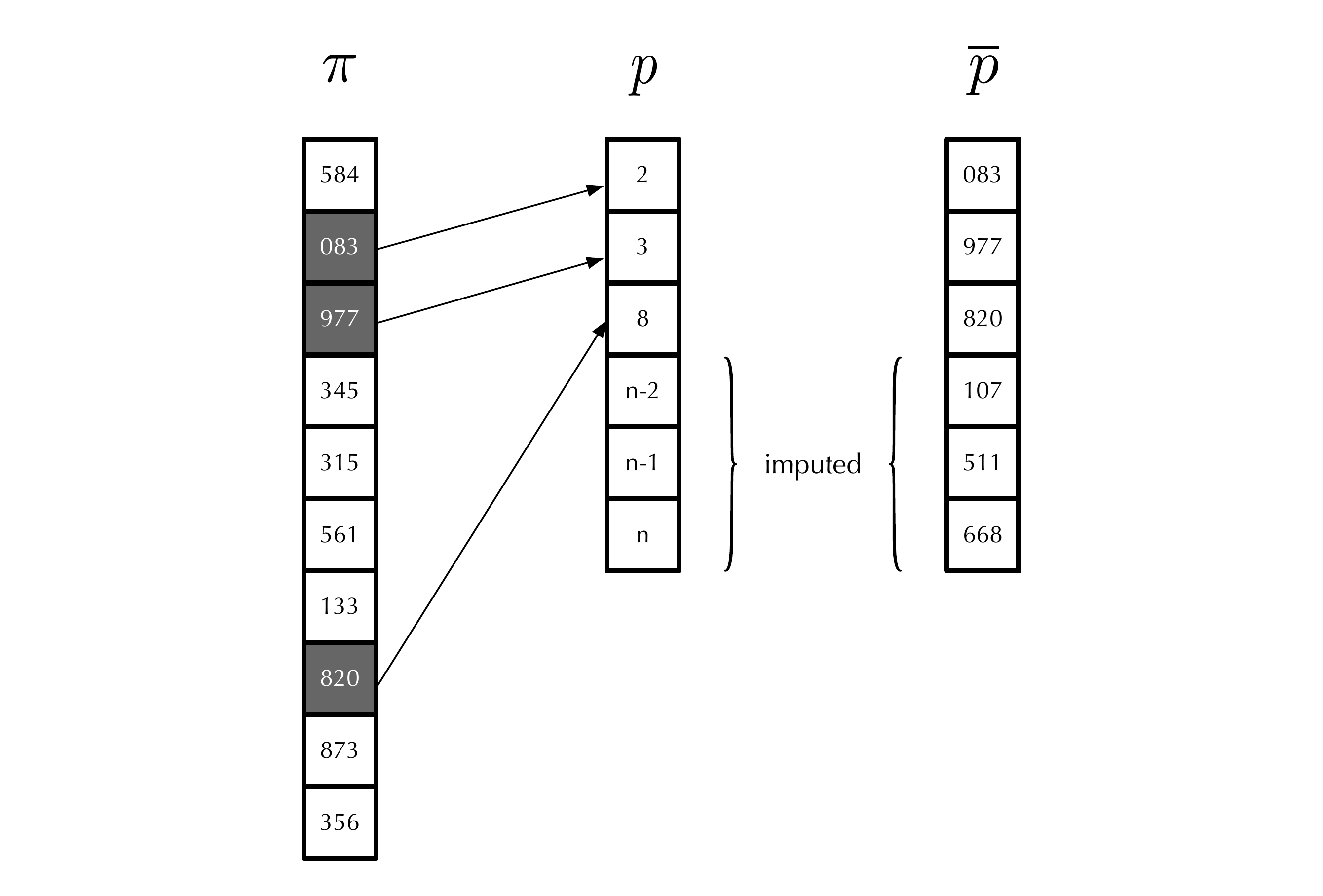}

\caption{Relevance Projection with Imputation.  Given a set of relevant item ids $\relset=\{083,107,511,668,820,977\}$, relevance projection of an incomplete top-10 ranking $\rlx$ and relevant set $\relset$ to a $\numrel\times 1$ vector of positions $\RPx$.  We also show the inverse projection vector $\invRP$ of items at specific recall levels.  } \label{fig:notation}
\end{figure}
\subsection{Permutation Imputation}
\label{sec:preliminaries:imputation}
Many ranking systems only provide a ranking on the top $\topk\ll\numdocs$ items, which may not include all relevant items, leaving elements of $\RPx$ undefined.  In order to use many metrics, especially recall-oriented metrics, we need to impute the positions of the unranked relevant items.  An \textit{optimistic imputation} would place the unranked relevant items immediately after the last retrieved items (i.e. $\topk+1,\topk+2,\ldots$).  Such a protocol would be susceptible to manipulation by returning few or no items.  Alternatively, we consider \textit{pessimistic imputation}, placing the unretrieved relevant items  at the bottom of the total order over the corpus. For example, if a system returns only three of six relevant items in the top $\topk$ at positions 2, 3, and 8, then we would define $\RPx$ as,
\begin{align*}
\RPx&=\underbrace{2,3,8}_{\text{top $\topk$}},\underbrace{\numdocs-2,\numdocs-1,\numdocs}_{\text{bottom $\ndocs-\topk$}}
\end{align*}
Pessimistic imputation is a conservative placement of the unretrieved relevant items and is well-aligned with our interest in robust performance. Moreover, it is consistent with behavior of metrics like rank-biased precision \cite{moffat:rbp}, which implicitly applies an exposure of 0 for unretrieved relevant items (i.e., for large $\numdocs$, $\lim_{i\rightarrow\ndocs}\rbpGamma^i\approx0$); and average precision\footnote{As defined in \tt{trec\_eval}.}, which implicitly applies an exposure of 0 for unretrieved relevant items (i.e., for large $\numdocs$, $\lim_{i\rightarrow\ndocs}\frac{1}{i}\approx0$).  We show an example of projection with imputation in Figure \ref{fig:notation}.
\subsection{Measuring Effectiveness}
When evaluating a ranking system, we consider two sets of important stakeholders: users and providers.  Users approach the system with information needs and requests and  ultimately define what is relevant.  Providers contribute items to the search system which serve to satisfy users.  Each item in the corpus is attributable, explicitly or not, to a content provider.  We include providers in our treatment to draw connections in existing work in the fair ranking literature, which often focuses on providers as stakeholders \cite{ekstrand:fair-ia-survey}.  In this section, we will characterize  a broad family of evaluation metrics for these two sets of users that will allow us, in subsequent sections, to define formal notions of recall and robustness.

\subsubsection{Measuring Effectiveness for Users}
\label{sec:preliminaries:metrics:searchers}

For a fixed information need and ranking $\rlx$, an evaluation metric is a function that scores rankings, $\metric : \rankings \times \docsets \rightarrow \NNreals$ where $\docsets$ is set of all subsets of $\docset$ excluding the empty set.   An evaluation metric, then, is a function whose domain is the joint space of all corpus permutations and possible relevance judgments and whose range is a non-negative scalar value.  We are specifically interested in a class of metrics that can be expressed in terms of a summation over recall levels.
\begin{definition}
\label{def:recall-level}
Given a ranking $\ranking\in\rankings$ and relevant items $\relset\in\docsets$, a \textit{recall-level metric} is an evaluation metric defined as a summation over $\numrel$ recall levels,
\begin{align}
\metric(\ranking,\relset) &= \sum_{i=1}^{\numrel} \exposure(\RPx_i)\normalization(i,\numrel)\label{eq:metric}
\end{align}
where $\exposure:\rankPositions\rightarrow\NNreals$ is a strictly monotonically decreasing \textit{exposure function} proportional to the probability that the user reaches rank position $i$ in their scan of the list; and  $\normalization:\rankPositions\times\rankPositions\rightarrow\NNreals$ is a metric-specific \textit{normalization function} of recall level and size of $\relset$.
\end{definition}
The product $\exposure(\RPx_i)\normalization(i,\numrel)$ is a  decomposition of what  \citeauthor{carterette:user-models-effectiveness} refers to as a `discount function' into an explicit function that models exposure and another that addresses any recall normalization \cite{carterette:user-models-effectiveness}.  
Within the set of recall-level metrics, we are further interested in the sub-class of metrics that satisfy the following criteria for `top-heaviness'.
\begin{definition}
\label{def:top-heavy}
We refer to a recall-level metric as \textit{top-heavy} if, for $j\in[0\isep\numrel)$,

\begin{align*}
\sum_{i=1}^{\numrel} \exposure(\RPx_i)\normalization(i,\numrel) &\geq \sum_{i=j+1}^{\numrel} \exposure(\RPx_i)\normalization(i-j,\numrel-j)
\end{align*}

\end{definition}
Top-heaviness indicates that, in the event that there are $j$ unjudged, relevant items in positions above the remaining $m-j$ relevant items, the metric computed over all $\numrel$ items must be greater than or equal to the metric value computed over only the $\numrel-j$ items.  Because we deal with metrics that include functions of $\numrel$, this is not an obvious property, but one that will be important as we consider incomplete judgments and relationships between possible users in Section \ref{sec:robustness}.\footnote{For more information on the relationship between incomplete judgments and metric stability, see}

Top-heavy recall-level metrics are a precise subclass of discounted metrics, covering a broad class of existing metrics such as average precision ($\ap$), reciprocal rank ($\rr$), normalized discounted cumulative gain ($\ndcg$), and rank-biased precision ($\rbp$), where we define the exposure and normalization  as,
\begin{align*}
\exposure_{\ap}(i)&=\frac{1}{i}&\normalization_{\ap}(i,\numrel) &= \frac{i}{\numrel}\\
\exposure_{\rr}(i)&=\frac{1}{i}&\normalization_{\rr}(i,\numrel) &= \begin{cases} 1 & \text{if $i=1$}\\
0& \text{otherwise}\end{cases}\\
\exposure_{\ndcg}(i)&=\frac{1}{\log_2(i+1)}&\normalization_{\ndcg}(i,\numrel) &= \left(\sum_{k=1}^{\numrel}\frac{1}{\log_2(k+1)}\right)^{-1}\\
\exposure_{\rbp}(i)&=(1-\rbpGamma)\rbpGamma^{i-1}&\normalization_{\rbp}(i,\numrel) &= 1
\end{align*}
Beyond classic ranking metrics, top-heavy recall-level metrics include non-traditional metrics such as those based on linear discounting (e.g., $\exposure_{\text{lin}}(i)=1-\frac{i}{\numdocs}$).  This results in a much broader class of metrics than those normally considered, for example, in the formal analysis of ranking metrics \cite{moffat:seven-properties-of-metrics,amigo:metric-axioms,ferrante:framework-for-utility-metrics}.  As a result, while all top-heavy recall-level metrics satisfy some formal properties of ranking evaluation metrics, large subsets of top-heavy recall-level metrics may satisfy more.  A detailed analysis of formal properties of top-heavy recall-level metrics can be found in Appendix \ref{app:proofs:metrics}.  As mentioned before, we can contrast this with \citeauthor{carterette:user-models-effectiveness}'s decomposition which focuses on the decomposition of metrics into gain and discount components \cite{carterette:user-models-effectiveness}.  In our case, we do not model gain, since we deal with binary relevance.  Our exposure and normalization functions, then, precisely define a subset of \citeauthor{carterette:user-models-effectiveness}'s discount functions that do not fit into his metric taxonomy since they do not consider recall normalization \cite{carterette:user-models-effectiveness}.

We focus on this class of metrics in order to prove properties of robustness in Section \ref{sec:robustness}.

\subsubsection{Measuring Effectiveness for Providers}
\label{sec:preliminaries:metrics:provider}
For content providers, we define the utility they receive from a ranking $\rlx$ as a function of their items' \textit{cumulative positive exposure}, defined as exposure of a provider's relevant content.\footnote{We do not consider provider utility when none of their associated items are relevant to the user's information need.  Although not covering situations where providers benefit from \textit{any} exposure (including of nonrelevant content), it is consistent with similar definitions used in the fair ranking literature \cite{singh:exposure,diaz:expexp}.

While we adopt a cumulative exposure model in this work, alternative notions of provider effectiveness are possible.  For example, normalizing by the number of relevant items contributed $|\prelset|$ would emphasize providers who contribute more content.}  Let $\prelset\subseteq\relset$ be the subset of relevant items belonging to a specific provider.  Since $\exposure$ captures the likelihood that a user inspects a specific rank position, we can compute the cumulative positive exposure as,
\begin{align}
\pmetric_{\exposure}(\ranking,\relset,\prelset) &= \sum_{i=1}^\numrel \exposure(\RPx_i)\ident(\invRP_i\in\prelset)\label{eq:pmetric}
\end{align}
where $\numrel$ and $\RPx$ are based on $\relset$.
Unless necessary, we will drop the subscript $\exposure$ from $\pmetric$ for clarity.

\subsection{Evaluation Method Desiderata}
\label{sec:desiderata}
Because there is no consensus on a single approach to validate a new evaluation method, we assemble desired theoretical and empirical properties of a method drawn from work in information retrieval and recommender system evaluation \cite{sakai:metrics,valcarce:recsys-ranking-metrics-journal}  and measurement theory \cite{jacobs:measurement-and-fairness,xiao:nlg-measurement-theory}.
\begin{itemize}
\item Validity
\begin{itemize}
\item \textbf{Content validity.}
\begin{itemize}
\item Is the evaluation theoretically related to the higher level concept? (Sections  \ref{sec:robustness:searchers}, \ref{sec:robustness:providers})
\item Is the evaluation better correlated with the higher level concept than existing metrics?   (Section  \ref{sec:robustness:analysis})
\end{itemize}
\item \textbf{Convergent validity.}  Is the evaluation method empirically \textit{correlated} with existing methods for measuring the \textit{same} higher level concept? (Section \ref{sec:results:agreement})
\item \textbf{Discriminant validity.}  Is the evaluation method empirically \textit{uncorrelated} with existing methods for measuring \textit{different} higher level concepts? (Section \ref{sec:results:agreement})
\end{itemize}
\item Sensitivity
\begin{itemize}
\item \textbf{Decision sensitivity.} Is the evaluation method better able to distinguish between \textit{rankings} compared to existing methods? (Section  \ref{sec:results:numties})
\item \textbf{System sensitivity.}  Is the evaluation method better able to distinguish between \textit{rankers} compared to existing methods? (Section   \ref{sec:results:sensitivity})
\end{itemize}
\item Reliability
\begin{itemize}
\item \textbf{Stable validity.}  Is the evaluation method stable when labels are missing compared to existing methods? (Section   \ref{sec:results:degradation})
\end{itemize}
\begin{itemize}
\item \textbf{Stable sensitivity.}  Does the evaluation method maintain sensitivity  when labels are missing compared to existing methods? (Section   \ref{sec:results:degradation})
\end{itemize}
\end{itemize}
Throughout this article, when we assess or compare evaluation methods, we will focus on these properties.

We note that the evaluation methods we develop in Sections \ref{sec:tse} and \ref{sec:leximin} measure population-based properties using worst-case analysis.  This means that, for each request, these methods consider a \textit{population of users} that tend to emphasize \textit{under-represented} behaviors and intents.  We contrast this with traditional metrics that model \textit{individual users} and emphasize \textit{well-represented} behaviors and intents.  This means that validation with, for example, behavioral feedback \cite{carterette:cikm2012} or user studies \cite{sanderson:preferences} is not possible.  In lieu of empirical validation, we emphasize both conceptual and theoretical properties of our evaluation methods, grounding them in the relevant work in philosophy and economics. This normative design of an evaluation method is consistent with recent work in the recommender system community \cite{ferraro:commonality,vrijenhoek:normative-diversity-news,vrijenhoek:radio,normalize-workshop-2023}.

\section{Recall}
\label{sec:tse}
As mentioned in Section \ref{sec:introduction}, the description of ranked ranking metrics as `recall-oriented' remains poorly defined, leaving the formal analysis of metrics for recall-orientation difficult. From a technical point of view, some work considers recall-orientation to be a binary criteria, dependent on whether a metric includes the recall base (i.e., $\relset$) in order to be computed \cite{kazai:xcg,sakai:precision-note}.  This would include metrics that compute set-based recall at some rank cutoff \cite{tomlinson:patent-evaluation,cormack:total-recall-sigir} as well as metrics like $\ap$ and $\ndcg$.  Using a set-based recall metric is particularly well-suited for recall-orientation in early stages of multi-stage ranking  \cite{macdonald:ltr-howto,mohammad:dynamic-shard-cutoff}.  A binary notion of recall-orientation does not capture that some metrics may be more recall-oriented than others.  This is captured, in part, by references to recall-orientation as related to  the depth in the ranking considered by the user \cite{diaz:crm}. More frequently, authors appeal to metrics like $\ap$ and $\recallK$ as being recall-oriented without clear discussion of what this means \cite{montazeralghaem:rl-rf,dai:ltr-resources,li:reqrec,hashemi:antique}. On the other hand, both \citet{mackie:prf-sparse} and \cite{magdy:pres-incomplete} refer to $\recallK$ as recall-oriented but $\ap$ being precision-oriented.  In light of the lack of consensus on recall-orientation, in Section \ref{sec:motivation}, we propose a new quantitative view of recall-orientation based on how sensitive a metric is for a user interested in finding every relevant item.  This allows us to see recall-orientation along a spectrum and compare the degrees of recall-orientation of different metrics. In Section \ref{sec:tse:tse}, based on this definition, we derive a new recall metric, total search efficiency.

\subsection{Metric Orientation}
\label{sec:motivation}
We are interested in more precisely defining precision and recall as constructs to be measured in ranking evaluation.  Although most evaluation metrics colloquially capture some aspects of both precision and recall, understanding the sensitivity to each remains vague.  We can address this vagueness by approaching precision and recall as two extremes of recall requirements.  At one extreme, precision as a construct reflects the satisfaction of a user who only needs exactly one relevant item, the minimum amount of retrievable content. We might find this in domains like web search.  At the other extreme, recall as a construct reflects the satisfaction of a user who needs \textit{every} relevant item, the maximum amount of retrievable content.  \citet{herlocker:recsys-eval} call this the `find all good items' tasks;  \citet{zobel:against-recall} refers to this as the \textit{totality} interpretation of recall, found in many technology-assisted review domains.
Indeed, this perspective is supported by tasks like recommender systems for scholarly literature reviews; evaluation programs like the TREC Total Recall Track \cite{roegiest:total-recall-2015} and patent search \cite{current-challenges-patent}; and by metrics like  `position of the last relevant ' \cite{zou:tar}.

We begin by defining the \textit{precision valence} of a ranking of $\ndocs$ items as how efficiently a user can find the \textit{first} relevant item.  For a fixed request, assume that we have $\numrel$ relevant items.    The ideal precision valence occurs when the first relevant item is at rank position 1.  The worst precision valence occurs when the first relevant item is at rank position $\ndocs-\numrel+1$, just above the remaining $\numrel-1$ relevant items.  Similarly, we refer to the \textit{recall valence} of a ranking as how efficiently a user can find \textit{all} of the relevant items.  The ideal recall valence occurs when the last relevant item is at position $\numrel$ (i.e., below the other $\numrel-1$ relevant items) and the worst precision valence when it is at position $\ndocs$.

In order to define the \textit{precision orientation} of a metric, we measure the difference in the best-case precision valence and worst-case precision valence for a given metric.  Although there is only one arrangement of positions of relevant items where the top-ranked item is at position $\ndocs-\numrel+1$, there are $\ndocs-1\choose\numrel-1$ arrangements of positions of relevant items where the top-ranked item is at position 1.  In order to control for the contribution of higher recall levels, we can consider, for the best-case precision valence, the ranking with a relevant item at the first position and the remaining $\numrel-1$ relevant items at the bottom of the ranking.  Precision orientation, then, measures sensitivity for a user interested in one relevant item.   We depict this graphically in Figure \ref{fig:metric-orientation:precision}.
Similarly, in order to define the \textit{recall orientation} of a metric, we measure the difference in the best-case recall valence and worst-case recall valence for a given metric.  Although there is only one arrangement of positions of relevant items where the bottom-ranked item is at position $\numrel$, there are $\ndocs-1\choose\numrel-1$ arrangements of positions of relevant items where the bottom-ranked item is at position $\ndocs$.  In order to control for the contribution of lower recall levels, we can consider, for the worst-case recall valence, the ranking with a relevant item at position $\ndocs$ and the remaining $\numrel-1$ relevant items at the top of the ranking.  Recall orientation, then, measures sensitivity for a user interested in all relevant items.   We depict this graphically in Figure \ref{fig:metric-orientation:recall}.

\begin{figure}
\begin{subfigure}[b]{0.5\linewidth}
\centering
\includegraphics[width=0.65\linewidth]{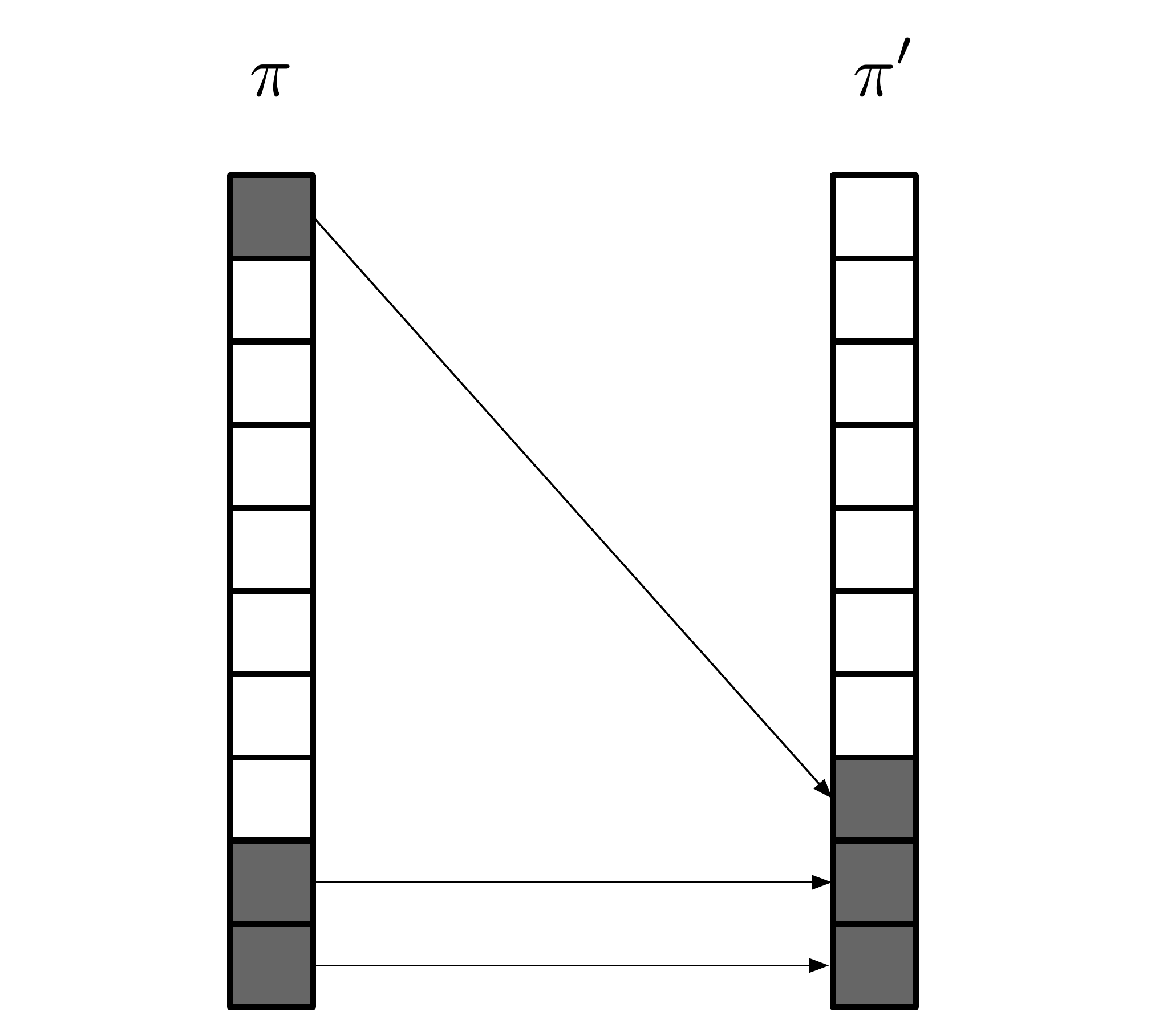}
\caption{Precision. }
\label{fig:metric-orientation:precision}
\end{subfigure}\begin{subfigure}[b]{0.5\linewidth}
\centering
\includegraphics[width=0.65\linewidth]{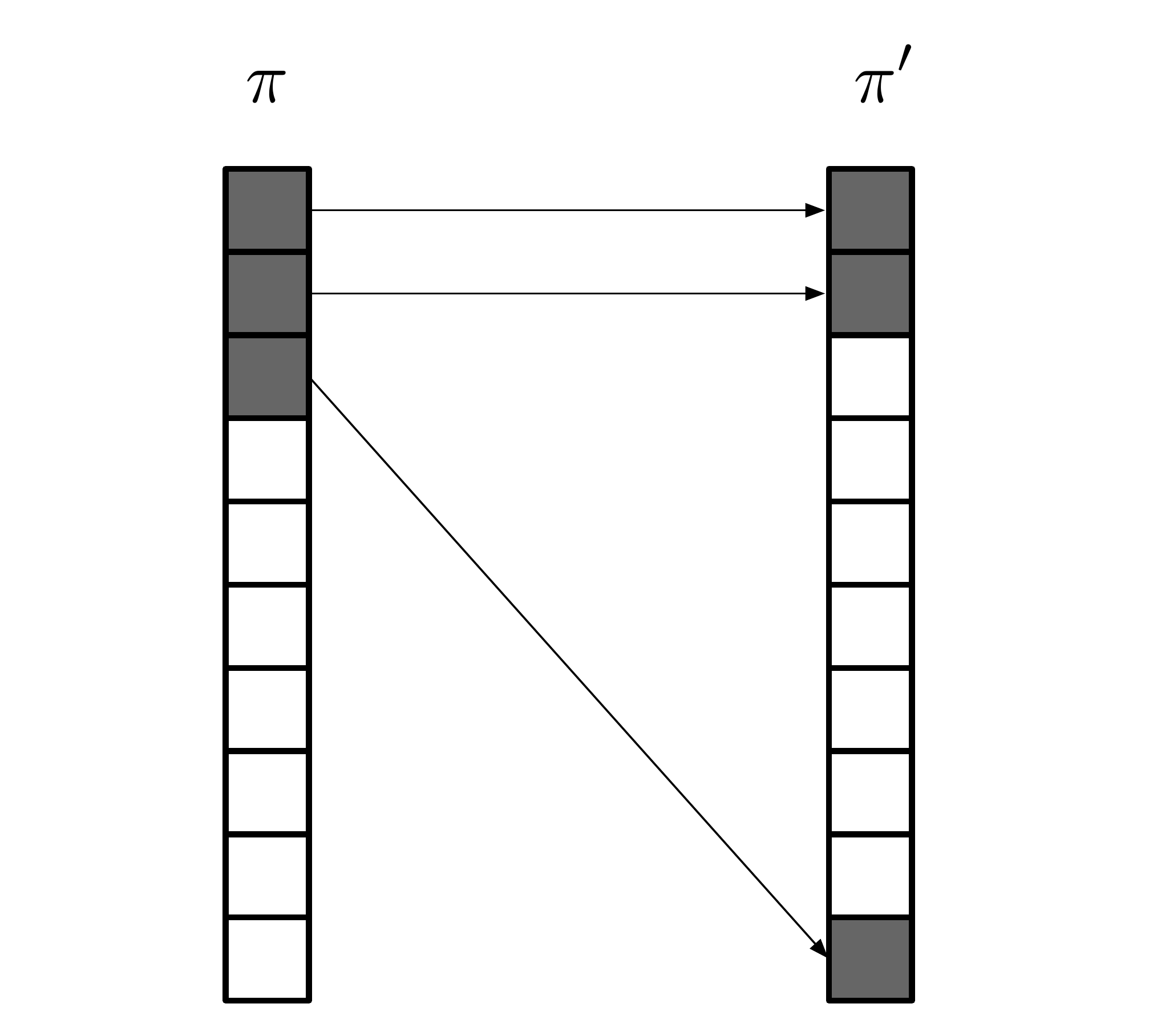}
\caption{Recall.   }
\label{fig:metric-orientation:recall}
\end{subfigure}
\caption{Metric orientation.  Each ranking $\ranking$ of ten items represented with a vector of cells ordered from top to bottom with shaded cells representing relevant items $\relset$. \textit{Precision orientation} (left) measures the degradation in a metric when the highest ranked relevant item is moved to the bottom of the ranking while holding all other positions fixed. \textit{Recall orientation} (right) measures the degradation in a metric when the lowest ranked relevant item is moved to the bottom of the ranking while holding all other positions fixed.  We measure the precision and recall orientation of a metric $\metric$ by the difference between $\metric(\ranking,\relset)-\metric(\ranking',\relset)$.}\label{fig:metric-orientation}
\end{figure}

In order to understand the intuition behind this definition of metric orientation, we can think about the recall requirements of precision-oriented or recall-oriented users.  The prototypical  precision-oriented user is satisfied by a single relevant item.  Precision-orientation quantifies how sensitive a metric is at measuring the best-case and worst-case for this precision-oriented user.  The prototypical  recall-oriented user is only satisfied when they find all of the  relevant items.  Recall-orientation quantifies how sensitive a metric is at measuring the best-case and worst-case for this recall-oriented user.  These prototypical users intentionally represent extremes in order to characterize existing metrics and control for any contribution from other recall requirements (e.g., those greater than one for precision and less than $\numrel$ for recall).

Figure \ref{fig:metric-orientation:results} shows the recall and precision orientation for several well-known evaluation metrics for a ranking of $\numdocs=10^{5}$ items.  Because both precision and recall orientation are functions of the number of  relevant items, we plot values for $\numrel\in[1\isep 15]$.

\begin{figure}
\begin{subfigure}[b]{0.5\linewidth}
\centering
\includegraphics[width=0.9\linewidth]{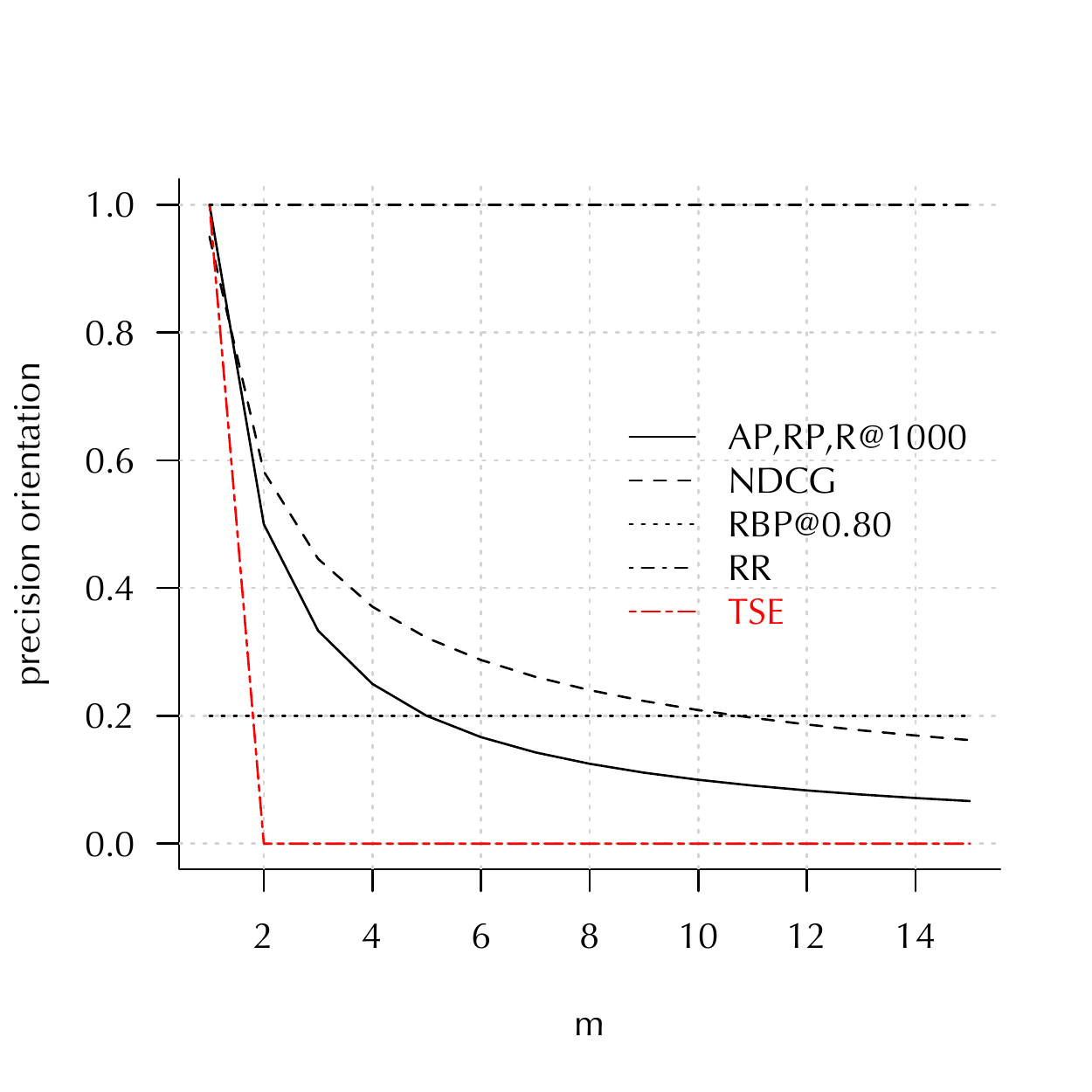}
\caption{Precision. }
\label{fig:metric-orientation:results:precision}
\end{subfigure}\begin{subfigure}[b]{0.5\linewidth}
\centering
\includegraphics[width=0.9\linewidth]{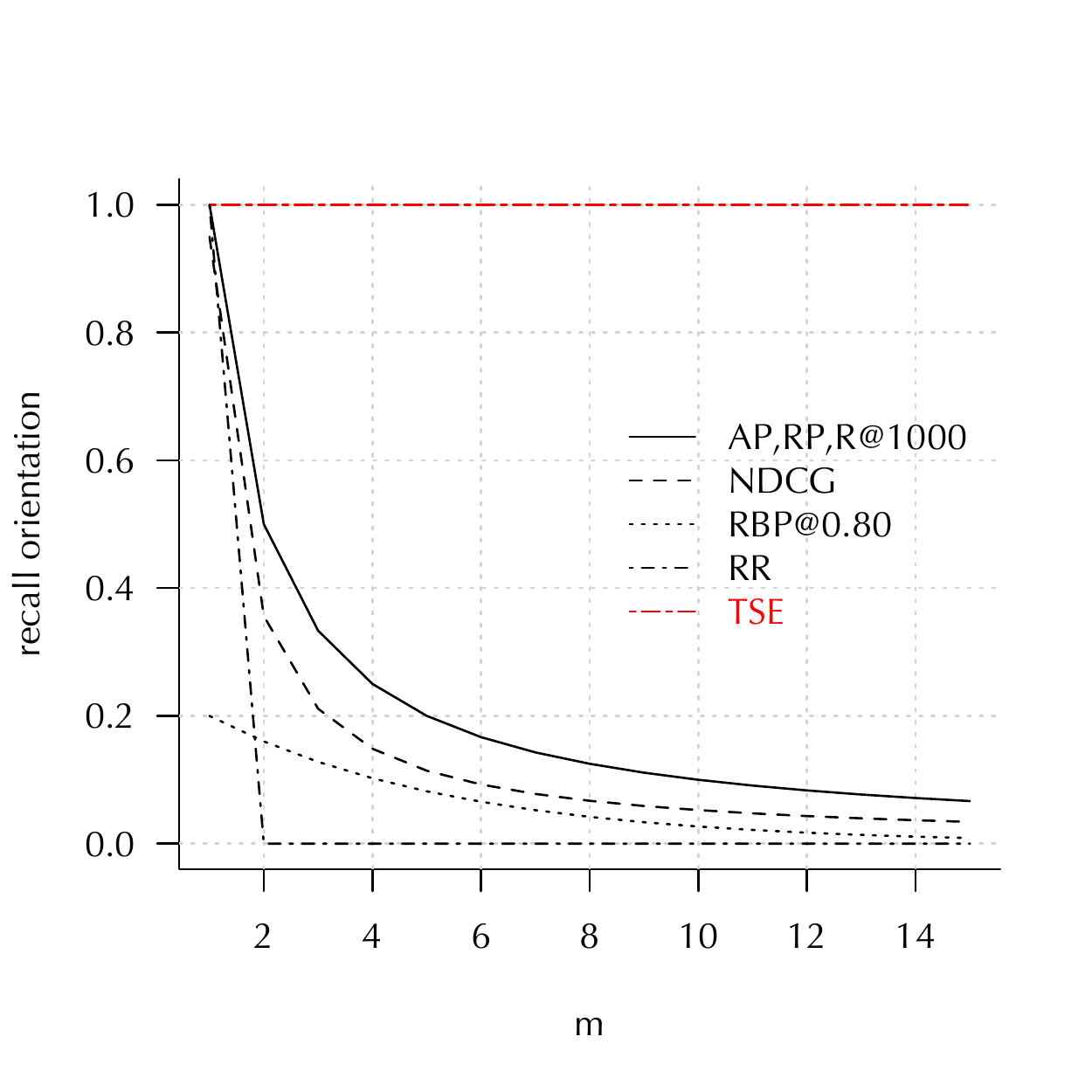}
\caption{Recall.   }
\label{fig:metric-orientation:results:recall}
\end{subfigure}
\caption{Metric orientation of ranking metrics of $10^5$ items with  $m\in[1\isep 15]$  relevant items.  The vertical axis reflects, for $\numrel\in[1,15]$, the change in metric value when (\subref{fig:metric-orientation:results:precision}) moving top-ranked item from position 1 to position $\numdocs-\numrel+1$ or (\subref{fig:metric-orientation:results:recall}) moving bottom-ranked item from position $\numrel$ to position $\numdocs$.
The values for TSE (Equation \ref{eq:tse}) are scaled by the lower and upper bound given a fixed $\numrel$ and therefore apply to any exposure model.
See Figure \ref{fig:metric-orientation} for details.
\textit{This figure best rendered in color.}}\label{fig:metric-orientation:results}.
\end{figure}

In terms of precision orientation,  the ordering of metrics follows conventional wisdom.  $\rr$ is often used for known-item or other high-precision tasks where the user is satisfied by the first relevant item.  Across all values of $\numrel$, $\rr$ dominates other metrics.  $\ndcg$, often used for evaluating both top-$N$ recommendations and web search results, is the next more precision-oriented metric.  $\ap$, `recall at 1000' ($\recallK$), and R-precision ($\rprecision$) all have the same precision-orientation and are dominated by $\ndcg$.
Rank-biased precision  ($\rbp$) dominates $\ndcg$ and $\ap$ once we reach a modest number of relevant items.  Both $\rr$ and $\rbp$ are not sensitive to the number of relevant items because neither is a function of $\numrel$.

In terms of recall orientation, the ordering of metrics follows conventional wisdom in the information retrieval and recommender systems community.  $\recall_{1000}$, $\ap$, and $\rprecision$ all dominate other metrics for all values of $\numrel$.  This family of metrics  is followed by metrics often used for precision tasks, $\ndcg$ and $\rbp$. $\rr$, the least recall-oriented metric, only considers the top-ranked relevant item and shows \textit{no} recall orientation unless there is only one relevant item.  We should note that, except for $\rr$, the recall orientation decreases with the number of relevant items because all of these metrics  aggregate  an increasing number of positions as $\numrel$ increases.  When evaluating over a set of requests with varying values of $\numrel$, mis-calibrated recall valences may result in requests with lower values of $\numrel$ dominating any averaging.

In this analysis, $\rr$ is a precision-oriented `basis' metric insofar as it is only dependent on the position of the highest-ranked relevant item.  We are interested in designing a symmetric recall-oriented `basis' metric that only depends on the position of the lowest-ranked relevant item.  We depict this desired metric as the red line in Figure \ref{fig:metric-orientation:results}. Traditional recall-oriented metrics do not satisfy this since  they depend on the position of the higher-ranked relevant items, especially as $\numrel$ grows.  In this paper, we identify the missing metric that captures recall-orientation while being well-calibrated across values of $\numrel$.

\subsection{Total Search Efficiency}
\label{sec:tse:tse}
Although metrics such as $\ap$ are often referred to as `recall-oriented', in this section, we focus on metrics that explicitly define recall as a construct.
Such recall metrics for ranking systems come in two flavors.\footnote{We exclude set-based  metrics used in set-based retrieval or some technology-assisted review evaluation \cite{cormack:tar,lewis:tar-certification}, since they require systems to provide a cutoff in addition to a ranking.}  The first flavor of recall metrics measures the fraction of relevant items found after a user terminates their scan of the ranked list.  Metrics like $\recallK$ and $\rprecision$ use a model of search depth to simulate how deep a user will scan.  We can define exposure and normalization functions for $\recall_k$ and $\rprecision$,
\begin{align*}
\exposure_{\recall_k}(i)&=\ident(i\leq k)&\normalization_{\recall_k}(j,\numrel) &= \frac{1}{\numrel}\\
\exposure_{\rprecision}(i)&=\ident(i\leq \numrel)&\normalization_{\rprecision}(j,\numrel) &= \frac{1}{\numrel}
\end{align*}
Note that, although we decompose these metrics using the notation of recall-level metrics, neither of these exposure functions strictly monotonically decrease in rank, so  they are not recall-level metrics.

The second flavor of recall metrics measures the effort to find all $\numrel$ relevant items.  \citet{cooper:esl} refers to this as a the \textit{Type 3 search length} and is measured by,
\begin{align*}
\esl_3(\rlx,\relset)&=\RPx_{\numrel}-\numrel\\
&\rankeq\RPx_{\numrel}
\end{align*}
Similarly, \citet{zou:tar} use  `position of the last relevant item'    to evaluate high-recall tasks. By contrast, \citet{rocchio:recall-error} proposed \textit{recall error}, a metric based on the average rank of  relevant items,
\begin{align*}
\recallerror(\rlx,\relset)&=\frac{1}{\numrel}\sum_{i=1}^{\numrel}\RPx_i-\frac{\numrel+1}{2}\\
&\rankeq\sum_{i=1}^{\numrel}\RPx_i
\end{align*}
Recall error can be sensitive to outliers at very low ranks, which occur frequently in even moderately-sized corpora \cite{magdy:pres}.

Inspired by \citeauthor{cooper:esl}'s $\esl_3$, we define a new recall-oriented top-heavy recall-level metric by looking at exposure of relevant items at highest recall level (i.e.  $i=\numrel$ in Equation \ref{eq:metric}).  We can define a recall-oriented metric based on any top-heavy recall-level metric by replacing its normalization function with the following,
\begin{align*}
\normalization_{\esl_3}(i,\numrel)=\begin{cases} 1 & \text{if $i=\numrel$}\\
0& \text{otherwise}
\end{cases}
\end{align*}
We refer to this as the efficiency of finding \textit{all} relevant items or the \textit{total search efficiency},  defined as,
\begin{align}
\tse_{\exposure}(\ranking,\relset) &= \sum_{i=1}^{\numrel} \exposure(\RPx_i)\normalization_{\esl_3}(i,\numrel)\label{eq:tse}\\
&= \exposure(\RPx_{\numrel})\nonumber
\end{align}
where the specific exposure function depends on base top-heavy recall-level metric (e.g., $\ap$, $\ndcg$).
TSE, then, is a family of metrics parameterized by a specific exposure function with properties defined in Section \ref{sec:preliminaries:metrics:searchers}.  Although computing $\tse_{\exposure}$ depends on the exposure model $\exposure$, unless necessary, we will drop the subscript for clarity.
We demonstrate the precision and recall orientation of $\tse$ in Figure \ref{fig:metric-orientation:results}.  Since TSE with an $\ap$ base behaves identically to $\rr$, except from the perspective of recall-orientation, we consider it $\rr$'s recall-oriented counterpart.  In the next section, we will connect this notion of recall to concepts of robustness and fairness.

\section{Robustness}
\label{sec:robustness}
In the context of a single ranking, we are interested in measuring its robustness in terms of its effectiveness for different possible users who might have issued the same request.\footnote{Robustness in the information retrieval community has traditionally emphasized slightly different notions from our ranking-based perspective.  For example, the TREC Robust track emphasized robustness of \textit{searcher effectiveness across information needs}, focusing evaluation on difficult queries \cite{ROBUST2004}.  Similarly, risk-based robust evaluation seeks to ensure that \textit{performance improvements   across information needs} are robust with respect to a baseline \cite{wang:robust-ranking,trecw2013:overview}.  Meanwhile, \citet{goren:ranking-robustness} proposed robustness as the \textit{stability of rankings across adversarial document manipulations}.  In the context of recommender systems, robustness has analogously focused on robustness of \textit{utility across users} \cite{xu:robust-recsys,wen:dro-recsys} and  \textit{stability of rankings across adversarial content providers} \cite{mobasher:robust-recsys}.
}
In retrieval, this might occur when multiple users issue the same text-based query.
In recommendation, this might occur when multiple users share the same engagement or session history, for example in cold-start or low-data situations.
In both retrieval and recommendation, even a single user may have multiple intents depending on the context (e.g., an item may be relevant to a user when they are studying versus when they are relaxing).
We present additional examples in Table \ref{tab:request-ambiguity} that demonstrate situations where data sparsity in the request can conceal two users with different intents.
The concept of robustness across users or user-contexts is related to work in search engine auditing that empirically studies how effectiveness varies across different searchers issuing the same request \cite{mehrotra:demographics-of-search}.
In our work, we define robustness as the effectiveness of a ranking for the worst-off user (or user-context) who might have issued a request.

\begin{table}
\caption{Ambiguity in requests. In both retrieval and recommendation tasks, multiple users or the same user at different times may consider different items relevance.  }\label{tab:request-ambiguity}
{\footnotesize
\begin{tabular}{l|>{\raggedright\arraybackslash}p{1.4in}p{1in}|ll}
&\multicolumn{2}{c|}{request}&\multicolumn{2}{c}{relevance}\\
&&&&\\
task & implicit & explicit & user 1 & user 2\\
\hline
web retrieval & - & query=`salsa' & recipe & dance \\
map retrieval & location=`NYC' & query=`restaurant' & vegetarian & Korean \\
movie recommendation & watched=[`Jaws', `Piranha'] & - & nautical horror &  1970s horror\\
music recommendation & saved=[`Black Sabbath', `Bach'] & query=`study music' & metal &  classical\\
\end{tabular}
}
\end{table}

Underlying our notion of robustness is a population-based perspective on ranking evaluation.  Classic effectiveness measures can be interpreted as expected values over different user populations defined by different browsing behavior \cite{robertson:ap-user-model,sakai:user-models,carterette:user-models-effectiveness}.  For example, \citet{robertson:ap-user-model} demonstrated that $\ap$ can be interpreted as the expected precision over a population of users with different recall requirements.  More generally, \citet{carterette:user-models-effectiveness} demonstrated this for a large class of metrics.  We can contrast this with online evaluation production environments where systems observe individual user behavior and do not need to resort to statistical models to capture different user behavior.  So, just as recent work in the fairness literature disaggregates evaluation metrics to understand  how  performance varies across groups \cite{mehrotra:demographics-of-search,ekstrand:cool-kids,ekstrand:fair-ia-survey,neophytou:revisiting-bias}, we can disaggregate traditional evaluation metrics to understand how performance varies across implicit subpopulations of users or providers.

From an ethical perspective, when considered the expected value over a population of users, traditional metrics make assumptions aligned with average utilitarianism, where the expected utility over some population is used to make decisions \cite{sidgwick:methodsofethics}.  While this reduces, in production environments, to averaging a performance metric across all logged requests, in offline evaluation, this is captured by the distribution underlying the metric, as suggested by  \citet{robertson:ap-user-model} and \citet{carterette:user-models-effectiveness}.  This means that if there are certain types of user behavior that are overrepresented in the data (online evaluation) or the user model (offline evaluation), they will dominate the expectation.  Users whose behaviors or needs have low probability in the data or the user model will be overwhelmed and effectively be obscured from measurement.

For robustness, instead of measuring the effectiveness of a system by adopting average utilitarianism and computing the expected performance over users, we can summarize the distribution of performance over users using alternative traditions based on distributive justice.  This follows recent literature in value-sensitive, normative design of evaluation metrics \cite{ferraro:commonality,vrijenhoek:normative-diversity-news,vrijenhoek:radio,normalize-workshop-2023}.  Specifically, inspired by related work in fair classification \cite{heidari:moral-fair-ml,liang:fairness-without-demographics,shah:rawlsian-fairness-dl-classifiers,memarrast:robust-ltr}, we can adopt Rawls' difference principle which evaluates a decision based on its value to the worst-off individual \cite{rawls:justice-as-fairness-restatement}.  In the context of a single ranking, this means the worst-off user or provider. As such, our worst-case analysis is aligned with Rawlsian versions of \begin{inlinelist}
\item equality of information access  (for users) and
\item fairness of the distribution of exposure (for providers)
\end{inlinelist}.  Even from a utilitarian perspective, systematic under-performance  can cost ranking system providers as a result of user attrition  \cite{liang:fairness-without-demographics,zhang:retention,mladenov:ltv} or negative impacts to a system's brand \cite{srinivasan:algorithms-brands}.

Although motivated by similar societal goals (e.g., equity, justice),   existing methods of measuring fairness in ranking are normatively very different from worst-case robustness.  First, the majority of fair ranking measures emphasize equal exposure  amongst providers \cite{ekstrand:fair-ia-survey} and is based on  strict egalitarianism, a different ethical foundation than Rawls' difference principle \cite{rawls:maximin-reasons}.  Pragmatically, in order to satisfy this within a single ranking, authors restrict analysis to stochastic ranking algorithms \cite{singh:exposure,diaz:expexp} or amortized evaluation \cite{asia:equity-of-attention,trec-fair-ranking-2019}.  Second, fair ranking analyses that focus on users tend be restricted to disaggregated evaluation, without reaggregating \cite{mehrotra:demographics-of-search,ekstrand:cool-kids}.  This is different from our focus on disaggregating and then summarizing the distribution of effectiveness with the worst-off user.  Most importantly, while most fairness work looks at \textit{either} users or providers, in our analysis, we  demonstrate that both worst-case user \textit{and} provider robustness are simultaneously captured by recall, as measured by $\tse$.

\subsection{User Robustness}
\label{sec:robustness:searchers}
Given a ranking $\rlx$, we would like to measure the worst-case effectiveness over a population of possible users.

\subsubsection{Possible Information Needs}
In this section, we describe how relevance, as traditionally used, reflects the set of \textit{possible user information needs}.  To see why, consider the notion of relevance described in Section \ref{sec:preliminaries:relevance}.  This is  often referred to as \textit{topical relevance}, the match between an item and the general topic of the  user \cite{boyce:beyond-topicality,cooper:relevance}.
By contrast, \citet{harter:psychological-relevance} uses the expression \textit{psychological relevance} to refer to  the extent to which, in the course of an information access session, an item changes the user's cognitive state with respect to their information need.  
Otherwise relevant items may stop being useful as a user's anomalous state of knowledge changes \cite{belkin:ASK1}.  Moreover, as \citet{harter:variation} notes, because users approach a system from a variety of backgrounds (i.e., states of knowledge),  the same request might find quite different utility from two topically relevant items in the corpus.  Indeed, researchers have  observed variation in utility in controlled experiments \cite{voorhees:variations,voorhees:trec-5} as well as production environments \cite{dou:large-scale-personalization,teevan:potential-for-personalization}.

In the context of information retrieval,  the notion of topical relevance can be interpreted as the potential utility of a document to the user.  A topically relevant document may be psychologically nonrelevant for a number of reasons.
First, a user may already be familiar with a judged relevant item in the ranking, which, in some cases, will, for that user, make the item nonrelevant. For example, in the context of a literature review, a previously-read relevant article may not be useful \cite{bookstein:relevance,boyce:beyond-topicality}.
Second, even if editorial relevance labels are accurate (i.e. all users would consider labeled items as relevant), the \textit{utility} of items may be isolated to a subset of $\relset$.  For example, in the context of decision support, including legal discovery and systematic review, topical relevance is the first step in finding critical information \cite{cooper:research-synthesis}.  In some cases, there will be a single useful `smoking gun' document amongst the larger set of relevant content.  In other cases, a single subset of relevant documents will allow one to `connect the dots.'
In a patent context, \citet{trippe:patent-recall} describe situations where there is risk to missing items that may turn out to be critical to assessing the validity of a patent.
So, while topical relevance is important, it only reflects the \textit{possible usefulness} to the user \cite{sedona:e-discovery-quality}.

Similarly, in the context of recommendation, the notion of topical relevance can be interpreted as the potential utility---but not  the psychological relevance or desirability---of an item to the user.  So, while a particular item may be of interest to a user in general, at any one point in time, it may be undesired for any number of reasons.  First, the context of a user can affect whether a particular item is desired.  For example, in music recommendation, a user may be interested in bluegrass music in many contexts but, while studying, may prefer ambient music.  Second, as with information retrieval, the desirability  of a previously-consumed relevant item may degrade due to order effects.  For example, in music recommendation, consider a  user who listened to a specific relevant song $X$. There are two possible effects.  On the one hand, \textit{satiation}  \cite{leqi:satiation} means that the user may not want to listen to $X$ immediately again, making it no longer relevant to that user.  Or, \textit{obsessive listening} \cite{ward:familiarity-music,conrad:extreme-relistening} means that  the user may want to listen to $X$ over again, making other otherwise relevant songs no longer relevant to that user.

So, from the perspective of relevance, $\relset$ should be considered the union of psychologically relevant items over all possible states of knowledge a user might have when approaching the system.  Indeed, multiple authors describe topical relevance as necessary but not sufficient for psychological relevance \cite{boyce:beyond-topicality,ruthven:trec-relevance,cooper:relevance}.   These papers suggest that rather than seeing a ranking system as acting to directly provide psychologically relevant items, it provides topically relevant items that are candidates to be scanned for psychologically relevant items by the user \cite{boyce:beyond-topicality}.
In light of this discussion, instead of considering any item in $\relset$ as definitely satisfying the user's information need, we consider it as only having a nonzero probability of satisfying the user's information need.

While the concept of diversity \cite{castells:recsys-diversity,santos:diversity-survey}, common to both recommendation and retrieval tasks, intends to support multiple information needs, they do not capture granular, item- or document-level needs.  For example, measuring book recommendation diversity along the language dimension may detect poor performance for a group of users interested in Spanish language books, it will not be effective at detecting poor performance for Spanish language books unfamiliar to a particular user.

Given that binary relevance reflects the \textit{possibility of psychological relevance}, we are interested in considering all users such that the union of their relevance criteria is $\relset$.  We can enumerate all such users over items as  $\users=\relset^+$, the power set of $\relset$ excluding the empty set.  This means that, for a given request, we have $\numusers=|\users|=2^{\numrel}-1$ possible users interested in at least one relevant item (Figure \ref{fig:users}).  This conservative definition of $\users$ captures all possible satisfiable users.

\begin{figure}
\includegraphics[width=0.5\linewidth]{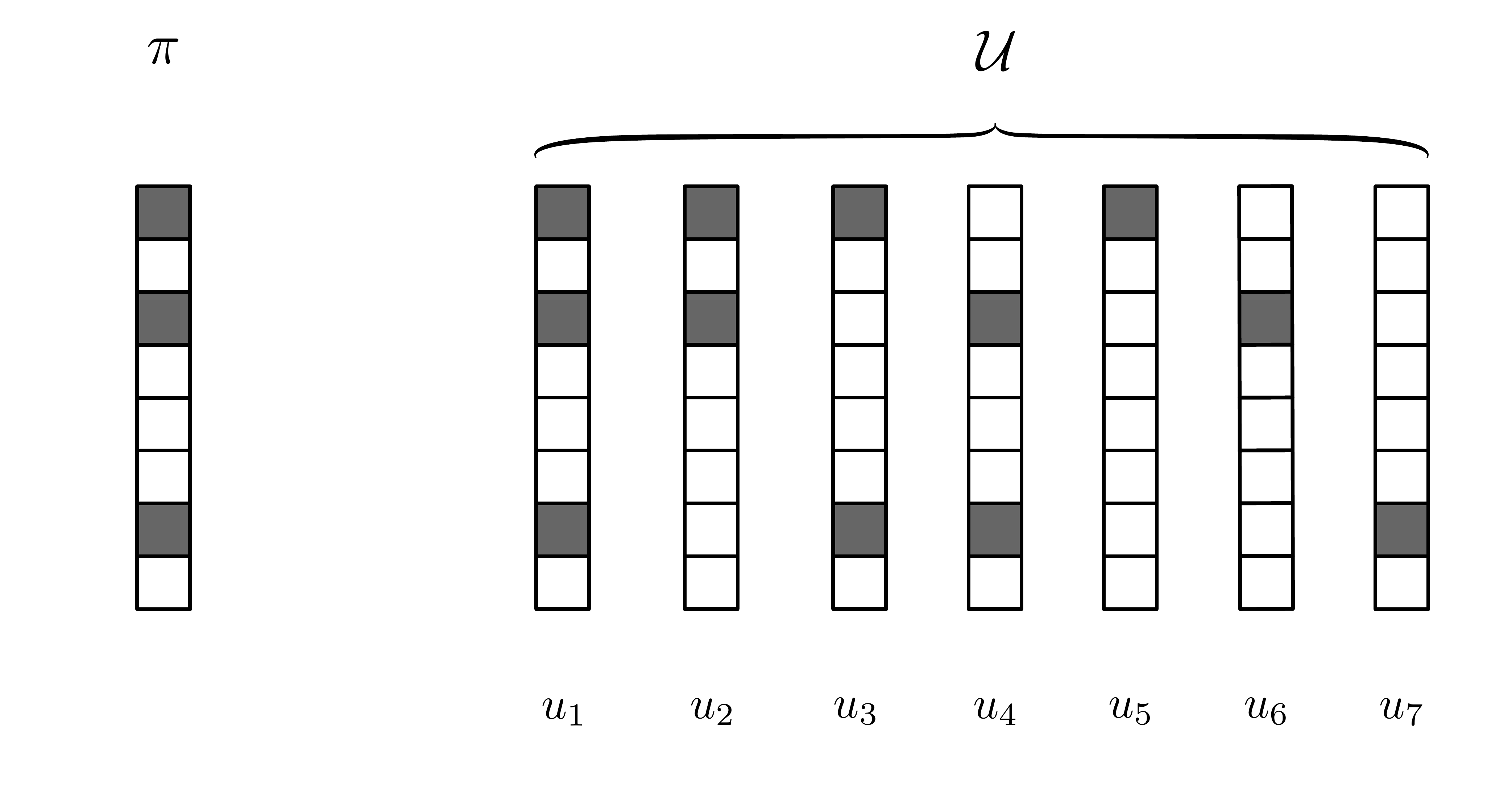}
\caption{Population of possible users $\users$ based on all combinations of relevant items from $\relset$ for a system ranking $\rlx$.  
}\label{fig:users}
\end{figure}

\subsubsection{Robustness Across Possible Information Needs}
\label{sec:robustness:searcher:robustness}
Given a set of possible information needs $\users$ based on $\relset$, we define the robustness of a ranking $\ranking$ as the effectiveness of the ranking for the worst-off user,
\begin{align}
\metricworst_{\metric}(\ranking,\relset)&=\min_{\user\in\users} \metric(\ranking,\user)
\end{align}
This is close to the notion of  robustness proposed by \citet{memarrast:robust-ltr}, who consider a worst-case user  for whom relevant items have a marginal distribution of features that matches the distribution in the full training set.  In comparison, our analysis  considers the \textit{full} set of worst-case situations, including those that do not match the training data.

One problem with this definition of robustness is that, because  $\numusers$  is exponential in $\numrel$, computing the minimum is impractical even for modest $\numrel$.  Fortunately, using the properties of top-heavy recall-level metrics, we can prove that,
\begin{align}
\metricworst_{\metric}(\ranking,\relset)&=\tse(\ranking,\relset)
\end{align}
In other words, the worst-off user is the one associated with $\RPx_{\numrel}$, the lowest-ranked relevant item.  We present a proof in Appendix \ref{app:proofs:wc}.  This result implies that recall-orientation captures the utility of a ranking for the worst-off user.

\subsection{Provider Robustness}
\label{sec:robustness:providers}
Providers are individuals who contribute content to the ranking system's catalog $\docset$.  Given a ranking $\ranking$, we would like to measure the worst-case effectiveness over a population of possible providers.

\subsubsection{Possible Provider Preferences}
Just as with information needs, each ranking consists of exposure of multiple possible providers.
Consider the domains like job applicant ranking systems or dating platforms, where each item in the catalog is associated with an individual person.  We assume that each relevant provider $i\in\relset$ is interested in its cumulative positive exposure in the ranking (Section \ref{sec:preliminaries:metrics:provider}), $\pmetric(\ranking,\relset,\{i\})$.
In the more general case, providers can possibly  be associated with multiple relevant items in $\relset$.  This might occur if a creator contributes multiple items to the catalog (e.g. multiple songs, videos, documents); or, a provider may aggregate content from multiple individual creators (e.g. publishers, labels).  Even if we have metadata attributing groups of items to specific providers or creators, their \textit{preferences} for exposure of those items may be unobserved.  Provider preferences can themselves be complex, covering a broad set of  commercial, artistic, and societal values \cite{hadida:performance-in-the-creative-industries,eikhof:bohemian-entrepreneurs}.

Given the uncertainty and ambiguity over providers and their preferences, as with information needs, we can consider the full set of latent providers and their preferences,  $\providers=\relset^+$ to reflect the set of \textit{possible} provider preferences.

\subsubsection{Robustness Across Possible Provider Preferences}
Just as with users,  we are interested in the utility  of the worst-off provider.  In the simple case where each provider is associated with a single item in $\relset$,  because exposure monotonically decreases with rank position, we know that the worst-off provider will be the one at the lowest rank; this is exactly $\tse(\ranking,\relset)$.  This is similar to earlier provider fairness definition \cite{zhu:rawlsian-cold-start}.  More generally, if, like information needs, we consider $\providers=\relset^+$, we define the worst-off provider  as,
\begin{align}
\metricworst_{\pmetric}(\ranking,\relset)&=\min_{\provider\in\providers} \pmetric(\ranking,\relset,\provider)
\end{align}
Given this definition, we can show that $\tse$ is equal to the utility of the worst-case provider (proof in Appendix \ref{app:proofs:wc}).

Together, the theoretical results in Sections \ref{sec:robustness:searchers} and \ref{sec:robustness:providers} provide  a new interpretation of recall from the perspective of potential subpopulations of content consumers and providers participating on the platform.  Worst-case analysis specifically connects naturally to \citeauthor{zobel:against-recall}'s notion of totality \cite{zobel:against-recall}.

\begin{table}
\caption{Agreement with $\metricworst_{\metric}(\ranking,\relset) < \metricworst_{\metric}(\ranking',\relset)$.  Probability of agreement over 10,000 simulated queries and pairs of random rankings for various corpus sizes.  For each query, we selected $\numrel\sim U(5,50)$ relevant items.  We include the fraction of rankings tied under $\metricworst_{\metric}(\ranking,\relset) = \metricworst_{\metric}(\ranking',\relset)$.      The values for TSE apply to any exposure model.      }\label{tab:wcagreement}
\begin{tabular}{cc|cccccc}
\hline
$\numdocs$ & tied & $\tse$	&   $\recall@1000$  &   $\rprecision$    &	$\ap$   &	$\ndcg$ & random\\
\hline
$10^3$	&	0.012	&	1.000	&	0.000	&	0.285	&	0.541	&	0.535	&0.492\\
$10^4$	&	0.001	&	1.000	&	0.420	&	0.077	&	0.552	&	0.549	&0.497\\
$10^5$	&	0.000	&	1.000	&	0.179	&	0.008	&	0.554	&	0.555	&0.498\\
$10^6$	&	0.000	&	1.000	&	0.026	&	0.001	&	0.547	&	0.554	&0.499\\
\hline
\end{tabular}
\end{table}

\begin{figure}
\centering
\includegraphics[width=0.4\linewidth]{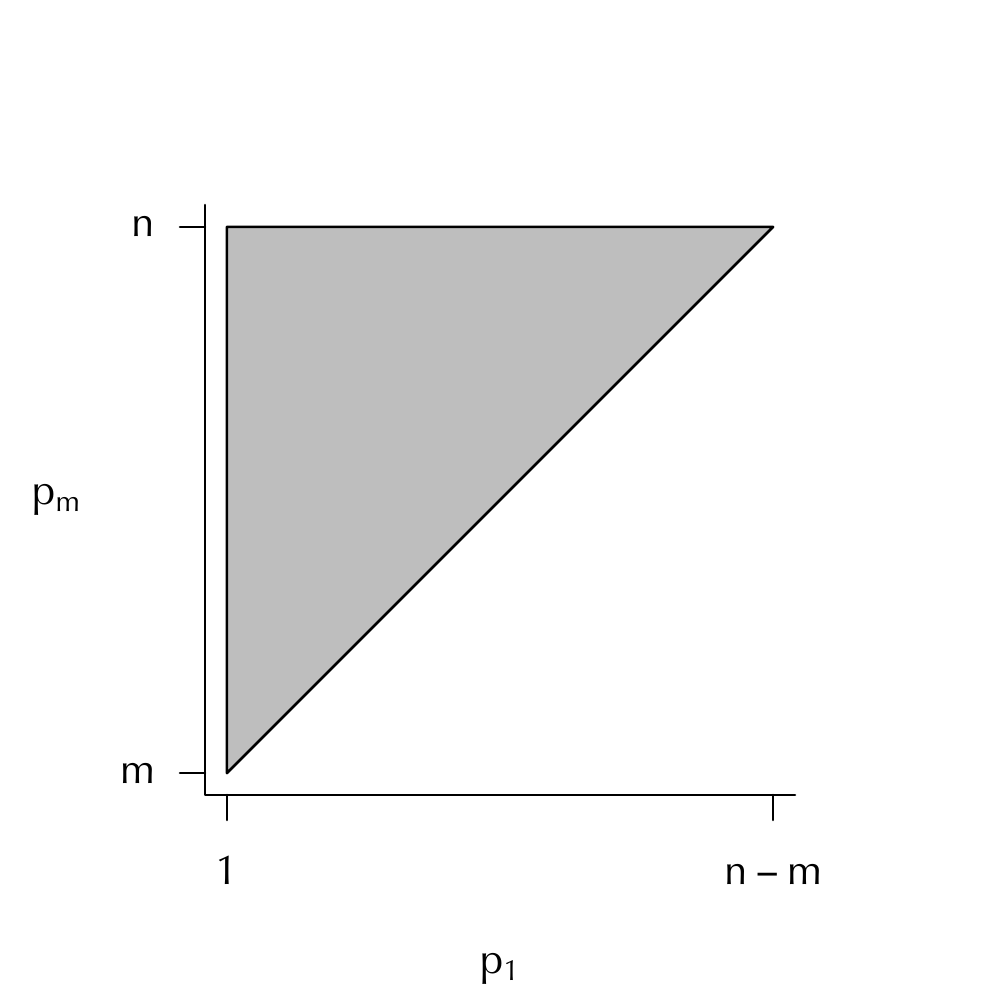}
\caption{Dependence between $\RPx_1$ and $\RPx_{\numrel}$.  The gray region indicates the possible values of each.}\label{fig:robustness:analysis:rpxi-dependence}
\end{figure}

\subsection{Robustness and Existing Metrics}
\label{sec:robustness:analysis}
We have demonstrated that whenever $\metricworst_{\metric}(\ranking,\relset) < \metricworst_{\metric}(\ranking',\relset)$, then $\tse(\ranking,\relset)< \tse(\ranking',\relset)$.  In order to understand the relationship between other metrics and worst-case performance, we simulate 10,000 requests by sampling 10,000 pairs of rankings $\ranking$ and $\ranking'$.  We can then compute the evaluation metric for each ranking and compare $\metric(\ranking,\relset)< \metric(\ranking',\relset)$ with $\metricworst_{\metric}(\ranking,\relset) < \metricworst_{\metric}(\ranking',\relset)$.  We conduct this simulation for $\numdocs\in\{10^3, 10^4, 10^5, 10^6\}$ and present results in  Table \ref{tab:wcagreement}.  We observe that, while $\tse$ has perfect sign agreement with $\metricworst_{\metric}$, other recall-oriented metrics have worse agreement than random, largely because they \textit{only} look at a prefix of $\RPx$.  The sign agreement of $\ap$ and $\ndcg$ is slightly better than random for two reasons.  First,  their aggregation (Equation \ref{eq:metric}) includes $\RPx_{\numrel}$ and will subtly affect the value, despite making a small contribution to the total sum.  Second,  $\RPx_{i}$ values depend on each other because $\rlx$ is a permutation.  In Figure \ref{fig:robustness:analysis:rpxi-dependence}, we compare the possible joint values of $\RPx_1$ and $\RPx_{\numrel}$.  This means that we should expect there to be some dependence between purely precision-oriented metrics (e.g. $\rr$) and purely recall-oriented metrics (e.g. $\tse$).  That said, $\ap$ and $\ndcg$ agree less than $\tse$ because their aggregation includes the positions of relevant items above $\RPx_{\numrel}$.  In whole, this result suggests that, even compared to traditional recall metrics, $\tse$ is better able to measure the robustness of a ranking.
\section{Lexicographic Evaluation}
\label{sec:leximin}
In Section \ref{sec:tse}, we defined recall-orientation from the perspective of users interested in finding the totality of relevant items and then proposed a new metric, $\tse$, based on this interpretation.  In Section \ref{sec:robustness}, we demonstrated how $\tse$ measures the worst-case utility for multiple definitions of users and providers, connecting it to notions of robustness and fairness, through Rawls' difference principle.  In this section, we will further develop the fairness perspective by combining recent work in preference-based evaluation with classic work in social choice theory, improving the nuance in worst-case analysis and allowing it to be useful as an evaluation tool.  We will begin by discussing the practical limitations of $\tse$ for evaluation (Section \ref{sec:leximin:sensitivity}) before developing a preference-based evaluation method derived from social choice theory that generalizes $\tse$ and improves its practical use (\ref{sec:leximin:leximin}).

\subsection{Low Sensitivity of Total Search Efficiency}
\label{sec:leximin:sensitivity}
Although measuring worst-case performance and, as a result, Rawlsian fairness, $\tse$ may not satisfy our desiderata for an evaluation method (Section \ref{sec:desiderata}).  To understand why,
consider two rankings  $\rlx$ and $\rly$ for the same request.  When comparing a pair of systems, we are interested in defining a  preference relation $\rlx\succ\rly$.   The worst-case preference $\rlx \tsepref \rly$, is defined as,
\begin{align}
\rlx \tsepref \rly &\leftrightarrow \min_{\user\in\users}\metric(\rlx,\user) > \min_{\user\in\users}\metric(\rly,\user)
\end{align}
We know from Section \ref{sec:robustness:searcher:robustness} that this can be efficiently computed as $\tse(\rlx) > \tse(\rly)$.  Unfortunately, in situations where the worst-off user is tied (i.e. $\min_{\user\in\users}\metric(\rlx,\user) = \min_{\user\in\users}\metric(\rly,\user)$), we cannot derive a preference between $\rlx$ and $\rly$.  Because we assume that unretrieved items occur at the bottom of the ranking and because most runs do not return all of the relevant items, an evaluation based on $\tse$ will observe many ties between $\rlx$ and $\rly$, limiting its effectiveness at distinguishing runs and use for system development \cite{canamares:numties}.  In Figure \ref{fig:numties-over-cutoff}, we simulated random pairs of rankings of 250,000 items and computed the number of metric ties for  a variety of retrieval cutoffs.  We observe that $\recallK$ and $\rprecision$ both have a large number of ties across all retrieval depths.  Despite having few ties for very deep retrievals, $\tse$ quickly observes many ties.  This is due to our conservative permutation imputation method (Section \ref{sec:preliminaries:imputation}).  We can compare all of these measures to $\ap$, which exhibits high sensitivity across most cutoffs. In this section, we will improve the sensitivity of $\tse$ to be comparable to $\ap$ using methods from social choice theory.

\begin{figure}
\centering
\includegraphics[width=0.5\linewidth]{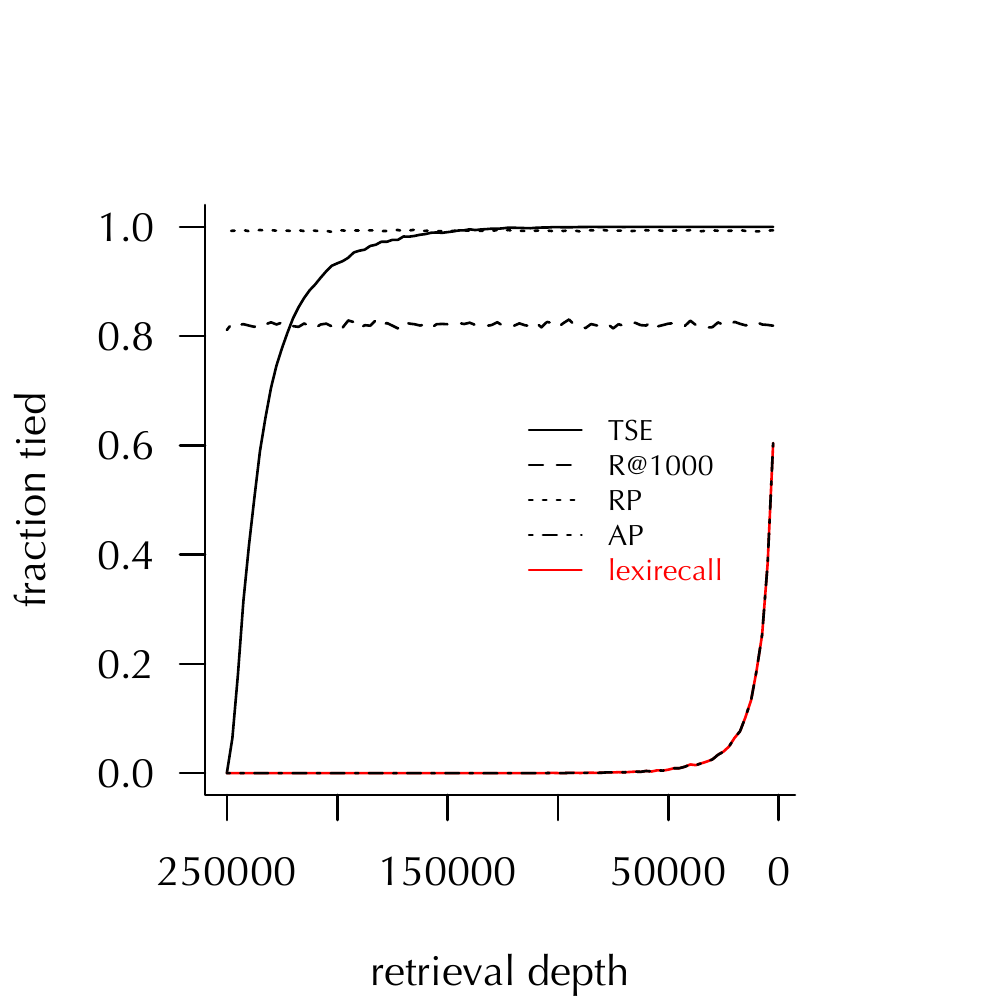}
\caption{Fraction  tied rankings as retrieval depth $\topk$ decreases for $\numdocs=250,000$ and 25 requests.  For each query, we selected $\numrel\sim U(5,50)$ relevant items.  This figure best rendered in color.}\label{fig:numties-over-cutoff}
\end{figure}

\subsection{Lexicographic Recall}
\label{sec:leximin:leximin}
We can address the lack of sensitivity of $\tse$ by turning to recent work on \textit{preference-based evaluation} \cite{diaz:rpp,diaz:lexiprecision,clarke:preference-tutorial-2023}.   As mentioned earlier, in many evaluation scenarios, our objective is to compute $\rlx\succ\rly$.  In \textit{metric-based evaluation}, we compute this preference by first computing the value of an evaluation metric for each ranking.  That is,
\begin{align}
\metric(\rlx)>\metric(\rly)\implies \rlx\succ\rly
\end{align}
Preference-based evaluation \cite{diaz:rpp,diaz:lexiprecision} is a quantitative  evaluation method that directly computes the preference between two rankings $\rlx$ and $\rly$ without first computing an evaluation metric.
\citet{diaz:rpp} and \citet{diaz:lexiprecision} show that preference-based evaluation can achieve much higher statistical sensitivity compared to standard metric-based evaluation.

\begin{figure}
\centering
\includegraphics[scale=0.3]{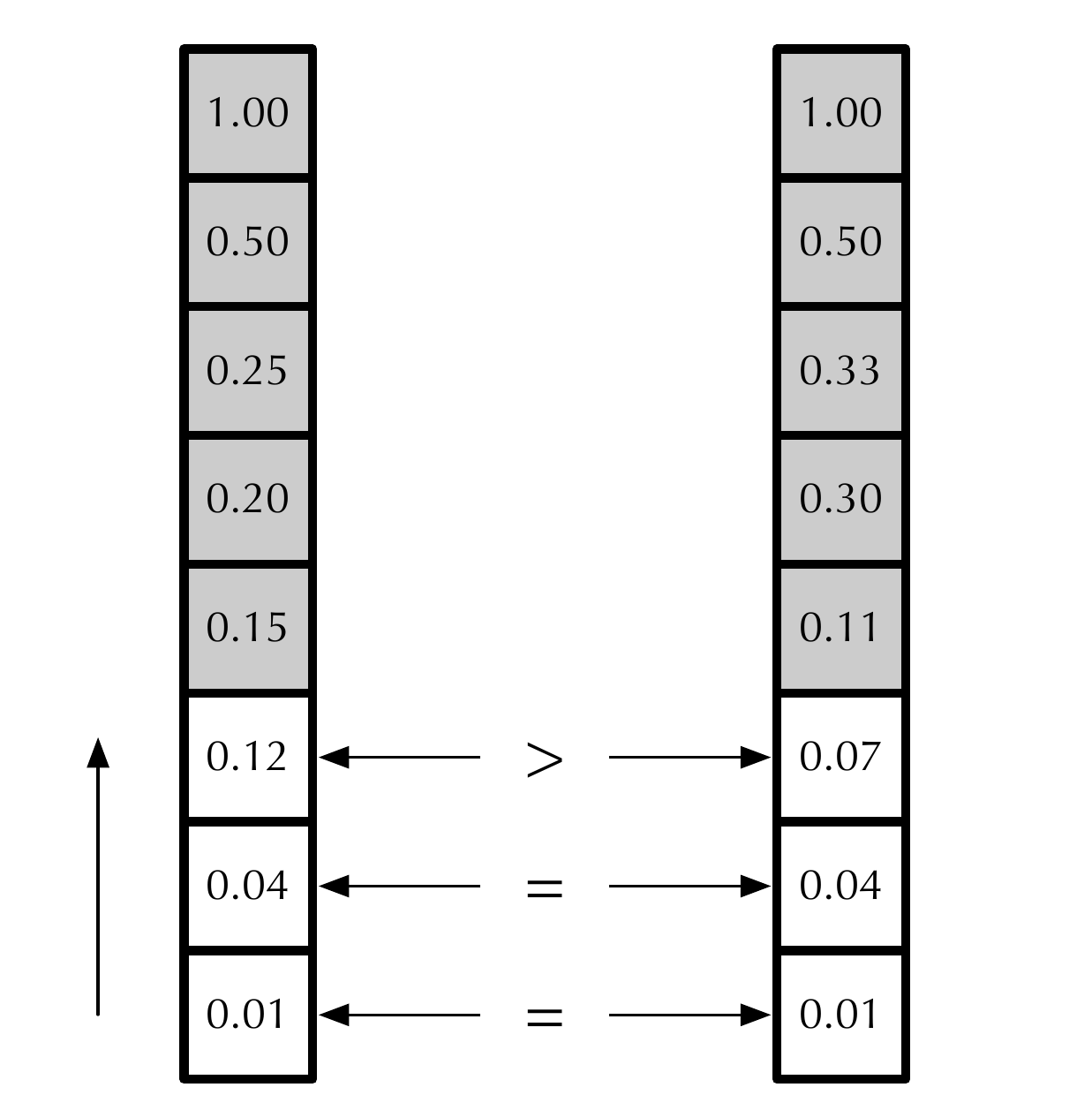}
\caption{Lexicographic relation between two sorted vectors with the same dimension.  In leximin ordering, we compare pairs of elements from the bottom up and define the preference between vectors based on the first difference in values.  In this example, the last element in both vectors is $0.01$ and the second to last element is $0.04$ in both.  The third to last elements differ and we say that we prefer the ranking on the left because $0.12>0.07$.}\label{fig:leximin}
\end{figure}
We can convert $\tse$ into a much more sensitive preference-based evaluation by returning to our discussion of fairness and robustness.  In the context of social choice theory, the number of ties in Rawlsian fairness can be addressed by adopting a recursive procedure known as \textit{leximin}, originally proposed by \citet{sen:collective-choice-and-social-welfare}.
Consider the problem of distributing a resource to $\numrel$ individuals, in our case users or providers.  Further consider two different allocations $x$ and $y$ represented by two $\numrel\times 1$ vectors where $x_i$ is the amount of resource allocated to the $i$th highest ranked individual, similarly for $y$.  In other words, $x$ and $y$ are the resource allocations in decreasing order.  Given these two allocations, we begin by inspecting the allocation to the lowest-ranked items, as we did with Rawlsian fairness.  If $x_\numrel>y_\numrel$, then  the bottom-ranked item of $x$ is better off and $x\succ y$; if $x_\numrel<y_\numrel$, then the bottom-ranked item of $y$ is better off and $x\prec y$; otherwise, the bottom-ranked items are equally well-off and we inspect the position of the next lowest item, $\numrel-1$.  If $x_{\numrel-1}>y_{\numrel-1}$, then we say $x\succ y$; if $x_{\numrel-1}>y_{\numrel-1}$, then we say $x\prec y$; otherwise, we inspect the position of the next lowest item $\numrel-2$.  We continue this procedure until we return a preference or, if we exhaust all $\numrel$ positions, we say that we are indifferent between the two rankings.  Formally,
\begin{align}
x \leximinpref y \leftrightarrow x_i > y_i
\end{align}
where $i=\max\{j\in[1\isep\numrel]: x_i\neq y_j\}$.  We show an example of this process in Figure \ref{fig:leximin}.  This way of comparing rankings generates a total lexicographic ordering over vectors of the same dimensionality and is often used in the fairness literature to address ties when adopting Rawls' difference principle.\footnote{For further discussion of the connection between Rawls' difference principle and leximin, see \cite{sen:rawlsian-axiomatics,sen:weights-and-measures,kolm:justice-and-equity,daspremont:leximax,moulin:fair-division}.}

We can use leximin to define the \textit{lexicographic recall} or lexirecall preference between  $\rlx$ and $\rly$.  For a fixed request and ranking $\rlx$, let $\UPx$ be the $\numusers\times 1$ vector of metric values for $\users$ sorted in decreasing order.  In other words, $\UPx_i=\metric(\ranking,\user^i)$, where $\user^i$ is the user with the $i$th-highest metric value.  We define $\UPy$ equivalently for $\rly$.  Lexicographic recall is defined as,
\begin{align}
\rlx\lexirecallpref\rly\leftrightarrow\UPx \leximinpref \UPy
\end{align}
Using lexirecall, we can define a ordering over unique rankings, which addresses the ties observed in  $\tse$.

Although operating over $\users$ grounds our evaluation in possible user information needs, scoring and ranking $\numusers$ subsets of $\relset$ can be computationally intractable.  So, just as we demonstrated that we only need to inspect the position of the last relevant item to compute $\tsepref$, we can demonstrate that we only need to compare the rank positions of the relevant items to compute $\rlx\lexirecallpref\rly$  (proof in Appendix \ref{app:proofs:lexirecall}) and, therefore,
\begin{align}
\rlx\lexirecallpref\rly\leftrightarrow\RPx_i < \RPy_i
\end{align}
where $i=\max\{j\in[1\isep\numrel]: \RPx_i\neq \RPy_j\}$.  Moreover, although defined as a user-oriented metric, we can also demonstrate that this results in provider leximin as well (proof in Appendix \ref{app:proofs:lexirecall}),

We can better understand lexirecall by returning to our discussion of robustness in Section \ref{sec:robustness}.  While $\tse$ provided one way to distinguish robustness of two rankings, it is very insensitive and unlikely to be of practical use.  In order to address this, we adopted leximin, a well-studied method for addressing insensitivity in applying Rawls' difference principle.  That said, lexirecall is still a measure of robustness. A lexirecall preference is simply claiming that one ranking is more fair or more robust than another.  Over a population of requests, then, we can compute the probability that one system's rankings are fairer or more robust than another.

\subsection{Sensitivity of Lexicographic Recall}
\label{sec:leximin:numties}
We can demonstrate the higher sensitivity of lexirecall  through simulation and analysis of the total space of permutations. In Figure \ref{fig:numties-over-cutoff}, we demonstrated that, for a set of random paired rankings,  the number of ties was high for traditional metrics and grew quickly for $\tse$ as the cutoff $\topk$ decreased.  Figure \ref{fig:numties-over-cutoff} also includes lexirecall, which exhibits substantially fewer ties than traditional metrics and $\tse$.
Independent of simulation, we are also interested in the probability of a tie over for randomly sampled pairs of complete rankings $\rlx,\rly\in\rankings$ (i.e. $\topk=\numdocs$).  We can derive these probabilities (see Appendix \ref{app:proofs:numties}) as functions of $\numrel$, $\numdocs$, and any parameters of the metric (e.g. $k$),
\begin{align*}
\prob(\rlx=_{\tse}\rly) &={\numdocs\choose \numrel}^{-2}\sum_{i=\numrel}^\numdocs{i-1\choose\numrel-1}^2\\
\prob(\rlx=_{\recall_k}\rly) &={\numdocs\choose \numrel}^{-2}\sum_{i=0}^\numrel{{k\choose i}^2{\numdocs-k \choose \numrel-i}^2}\\
\prob(\rlx=_{\rprecision}\rly) &={\numdocs\choose \numrel}^{-2}\sum_{i=0}^\numrel{{\numrel\choose i}^2{\numdocs-\numrel \choose \numrel-i}^2}\\
\prob(\rlx\lexirecalleq\rly)&=\frac{\numrel!(\numdocs-\numrel)!}{\numdocs!}
\end{align*}
To better understand the relationship between these probabilities, we display probabilities of ties for several retrieval depths in Table \ref{tab:numtiesanalytic}.  Although Figure \ref{fig:numties-over-cutoff} demonstrated that $\tse$ exhibited poor sensitivity when $\topk<\numdocs$, it is much more sensitive for complete rankings, in part because random complete rankings are less likely to share $\RPx_{\numrel}$ than imputed rankings.  Both traditional recall metrics exhibit many more ties, especially as the retrieval depth grows.  Amongst methods, lexirecall and $\ap$ demonstrate the few or no ties across corpus sizes.  Table \ref{tab:numtiesanalytic:numrel} presents the same results for varying numbers of relevant items.  The results are consistent  with Table \ref{tab:numtiesanalytic}, where traditional recall metrics exhibit a large number of ties, which decreases slowly with the number of relevant items.  By comparison, $\tse$, lexirecall, and $\ap$ show higher sensitivity with a negligible number of ties.

\begin{table}
\caption{Probability of a metric tie $\prob(\rlx=_{\mu}\rly)$ for randomly sampled permutations and $\numrel=10$.  }\label{tab:numtiesanalytic}
\begin{tabular}{c|ccccc}
\hline
$\numdocs$&	$\tse$&	$\recall_{1000}$&	$\rprecision$&	$\ap$ & $\lexirecall$\\
\hline
$10^3$	&	0.005	&	1.000	&	0.825	&	0.000	&	0.000	\\
$10^4$	&	0.001	&	0.313	&	0.980	&	0.000	&	0.000	\\
$10^5$	&	0.000	&	0.826	&	0.998	&	0.000	&	0.000	\\
$10^6$	&	0.000	&	0.980	&	1.000	&	0.000	&	0.000	\\
\hline
\end{tabular}
\end{table}

\begin{table}
\caption{Probability of a metric tie $\prob(\rlx=_{\mu}\rly)$ for randomly sampled permutations and $\numdocs=10^6$.  }\label{tab:numtiesanalytic:numrel}
\begin{tabular}{c|ccccc}
\hline
$\numrel$&	$\tse$&	$\recall_{1000}$&	$\rprecision$&	$\ap$ & $\lexirecall$\\
\hline
1	&	0.000	&	0.998	&	1.000	&	0.000	&	0.000	\\
5	&	0.000	&	0.990	&	1.000	&	0.000	&	0.000	\\
10	&	0.000	&	0.981	&	1.000	&	0.000	&	0.000	\\
25	&	0.000	&	0.952	&	0.999	&	0.000	&	0.000	\\
50	&	0.000	&	0.907	&	0.995	&	0.000	&	0.000	\\
\hline
\end{tabular}
\end{table}
\section{Empirical Analysis}
In this section, we  empirically assess lexirecall with respect to the associated empirical desiderata from Section \ref{sec:desiderata}:
\begin{inlinelist}
\item correlation with existing metrics,
\item ability to distinguish between rankings,
\item ability to distinguish between systems, and
\item robustness to missing labels.
\end{inlinelist}

\subsection{Methods and Materials}

\subsubsection{Data}
We evaluated  ranking system runs across a variety of conditions (Table \ref{tab:data}).  For each dataset, we have a set of evaluation requests and associated relevance judgments.  In addition, each dataset involved a number of competing systems, each of which produced a ranking for every request.
The movielens, libraryThing, and beerAdvocate datasets were downloaded from a public repository with splits and processing described in prior work \cite{valcarce:recsys-ranking-metrics-journal}.\footnote{\url{https://github.com/dvalcarce/evalMetrics}}   Amazon data sets were prepared using LensKit 0.14.4 \citep{lkpy} to train a variety of recommendation models on 5-core global split datasets from the 2023 Amazon Reviews data \citep{hou2024bridging} and generate top-K recommendation runs for the users in the test splits, using an early version of the LensKit Codex\footnote{\url{https://codex.lenskit.org}; runs will be provided as supplemental resources for this paper when published.}.  For  recommender systems datasets, consistent with \cite{valcarce:recsys-ranking-metrics-journal}, we converted graded judgments to binary labels by considering any rating below 4 as nonrelevant and otherwise relevant for these datasets.
All retrieval datasets were downloaded from NIST.
In order to analyze results for different ranking depths, we categorized  datasets as  deep ($\topk>1000$), standard ($\topk\in(100\isep1000]$), or shallow ($\topk\leq100$).

\begin{table}[t]
\caption{Datasets used in empirical analysis.  Runs submitted to the associated TREC track or recommendation task.  Datasets are labeled according to the depth of runs: `deep' ($\topk>1000$), `standard' ($\topk\in(100\isep1000]$), `shallow' ($\topk\leq100$).}\label{tab:data}
{
\begin{tabular}{lccccc}
\hline
&   requests    &   runs    &   rel/request &   docs/request    &   label\\
\hline
\textbf{recommendation}\\
amzn-cds-vinyl-100	&	6875	&	10	&	2.76	&	100	&	shallow	\\
amzn-cds-vinyl-1000	&	6875	&	10	&	2.76	&	1000	&	standard	\\
amzn-cds-vinyl-2000	&	6875	&	10	&	2.76	&	2000	&	deep	\\
amzn-musical-instruments-100	&	11648	&	10	&	3.08	&	100	&	shallow	\\
amzn-musical-instruments-1000	&	11648	&	10	&	3.08	&	1000	&	standard	\\
amzn-musical-instruments-2000	&	11648	&	10	&	3.08	&	2000	&	deep	\\
amzn-software-100	&	6345	&	10	&	2.28	&	100	&	shallow	\\
amzn-software-1000	&	6345	&	10	&	2.28	&	1000	&	standard	\\
amzn-software-2000	&	6345	&	10	&	2.28	&	2000	&	deep	\\
amzn-video-games-100	&	11554	&	10	&	3.08	&	100	&	shallow	\\
amzn-video-games-1000	&	11554	&	10	&	3.08	&	1000	&	standard	\\
amzn-video-games-2000	&	11554	&	10	&	3.08	&	2000	&	deep	\\
beerAdvocate	&	17564	&	21	&	13.66	&	99.39	&	shallow	\\
libraryThing	&	7227	&	21	&	13.15	&	100.00	&	shallow	\\
movielens	&	6005	&	21	&	18.87	&	100.00	&	shallow	\\
\\
\textbf{retrieval}\\
core (2017)	&	50	&	75	&	180.04	&	8853.11	&	deep	\\
core (2018)	&	50	&	72	&	78.96	&	7102.61	&	deep	\\
deep-docs (2019)	&	43	&	38	&	153.42	&	623.77	&	standard	\\
deep-docs (2020)	&	45	&	64	&	39.27	&	99.55	&	shallow	\\
deep-docs (2021)	&	57	&	66	&	189.63	&	98.83	&	shallow	\\
deep-docs (2022)	&	76	&	42	&	1245.62	&	98.86	&	shallow	\\
deep-docs (2023)	&	82	&	5	&	75.10	&	100.00	&	shallow	\\
deep-pass (2019)	&	43	&	37	&	95.40	&	892.51	&	standard	\\
deep-pass (2020)	&	54	&	59	&	66.78	&	978.01	&	standard	\\
deep-pass (2021)	&	53	&	63	&	191.96	&	99.95	&	shallow	\\
deep-pass (2022)	&	76	&	100	&	628.145	&	97.50	&	shallow	\\
deep-pass (2023)	&	82	&	35	&	49.87	&	99.90	&	shallow	\\
legal (2006)	&	39	&	34	&	110.85	&	4835.07	&	deep	\\
legal (2007)	&	43	&	68	&	101.023	&	22240.30	&	deep	\\
robust (2004)	&	249	&	110	&	69.93	&	913.82	&	standard	\\
web (2009)	&	50	&	48	&	129.98	&	925.31	&	standard	\\
web (2010)	&	48	&	32	&	187.63	&	7013.21	&	deep	\\
web (2011)	&	50	&	61	&	167.56	&	8325.07	&	deep	\\
web (2012)	&	50	&	48	&	187.36	&	6719.53	&	deep	\\
web (2013)	&	50	&	61	&	182.42	&	7174.38	&	deep	\\
web (2014)	&	50	&	30	&	212.58	&	6313.98	&	deep	\\
\hline
\end{tabular}
}
\end{table}

\subsubsection{Evaluation Methods}
We computed lexirecall using pessimistic imputation.
We compare $\lexirecall$ with two traditional recall metrics ($\recallK$ and $\rprecision$) and two metrics that combine recall and precision ($\ap$ and $\ndcg$).  Definitions for metrics can be found in Section \ref{sec:preliminaries}.
An implementation can be found at \url{https://github.com/diazf/pref_eval}.

\subsection{Results}
\label{sec:results}
\subsubsection{Agreement with Existing Metrics}
\label{sec:results:agreement}
To understand the similarity of lexirecall and traditional metrics, we measured its preference agreement with traditional metrics.  Specifically,  given an observed metric difference, $\metric(\rlx)\neq\metric(\rly)$, for a traditional metric in our datasets, we computed how often lexirecall agreed with the ordering of $\rlx$ and $\rly$,
\begin{align*}
\frac{\sum_{\rlx,\rly\in\rldata} \ident\left(\rlx\succ_{\metric}\rly \land \rlx\lexirecallpref\rly\right)}{\sum_{\rlx,\rly\in\rldata}\ident\left(\rlx\succ_{\metric}\rly\right)}
\end{align*}
where $\rldata$ is the set of rankings in our dataset.  We present results in Figure \ref{fig:agreement}.  For reference, we include high precision metrics $\rr$ and $\ndcgTen$, which expectedly have the weakest agreement with lexirecall across all retrieval depths.  Similarly, across all depths, we observed highest sign agreement with $\recallK$, indicating an alignment between lexirecall and traditional notions of recall.  Note that in the standard and shallow conditions, where $\topk\leq1000$, if there is a difference in $\recallK$, then there is a difference in lexirecall due to pessimistic imputation; the converse is not true since lexirecall can distinguish rankings that are tied under $\recallK$.  The agreement with $\ndcg$ increases to match that of $\recallK$ with increased depth, perhaps due to the weaker position discounting in NDCG and higher likelihood of including a value based on the lowest ranked relevant item as depth increases (see Section \ref{sec:robustness:analysis}).  Both $\rprecision$ and $\ap$ show comparable agreement higher relative to $\rr$ and $\ndcgTen$.

\begin{figure}
\centering
\begin{subfigure}[b]{\linewidth}
\centering
\includegraphics[width=0.95\linewidth]{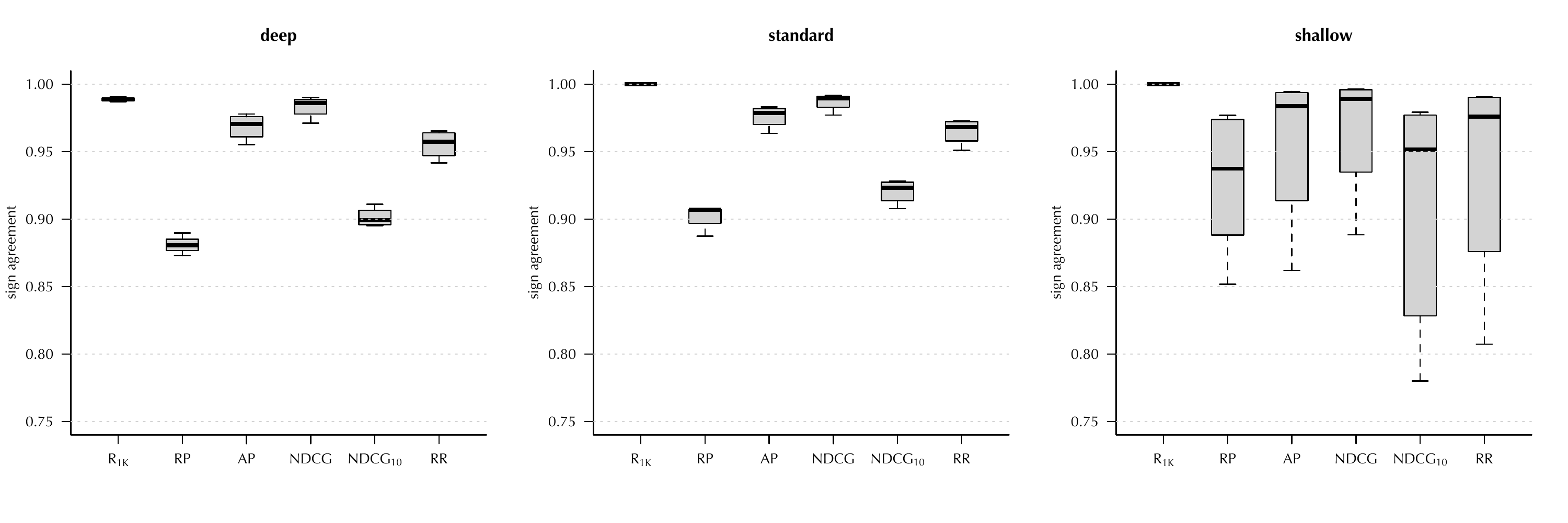}
\caption{Recommendation}
\end{subfigure}
\begin{subfigure}[b]{\linewidth}
\centering
\includegraphics[width=0.95\linewidth]{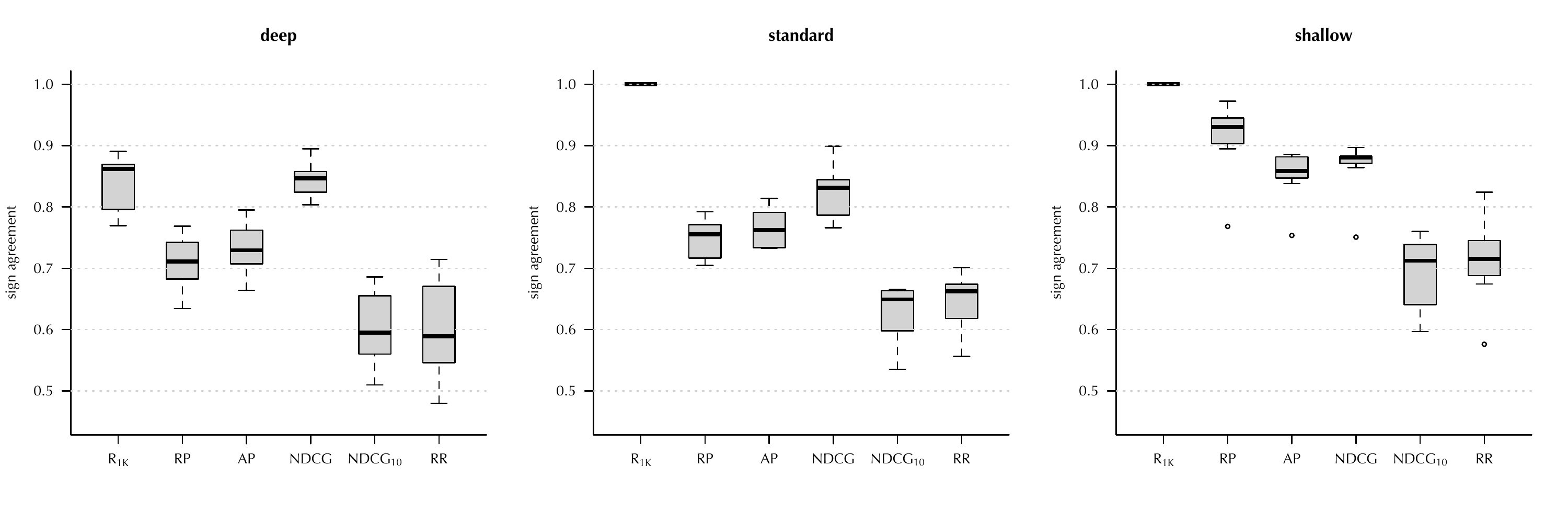}
\caption{Retrieval}
\end{subfigure}
\caption{Agreement between lexirecall and traditional metrics over rankings $\rlx,\rly$ in our datasets.  Fraction of ranking pairs  where the lexirecall preference agrees with the sign of the metric difference.  }\label{fig:agreement}
\end{figure}

\subsubsection{Detection of Difference Between Rankings}
\label{sec:results:numties}
In Section \ref{sec:leximin:numties}, we observed that, for  pairs of rankings $\rlx,\rly\in\rankings$ sampled \textit{uniformly at random}, lexirecall  resulted in fewer ties than $\rprecision$ and $\recallK$.  Figure \ref{fig:numties} presents the empirical fraction of ties when sampling from rankings in our dataset (i.e. $\rldata$).  First, consider  traditional recall metrics $\recallK$ and $\rprecision$.  The empirical fraction of ties is substantially lower than suggested by Figure \ref{fig:numties-over-cutoff} (different $\topk$) and Table \ref{tab:numtiesanalytic} (different $\numdocs$).  This can be explained by the concentration of relevant items in the top positions $\rldata$ compared to $\rlset$, resulting in fewer ties.  That said, the fraction of ties for both of these metrics is substantially higher than observed for lexirecall, something consistent with results in Section \ref{sec:robustness:analysis}.  Although $\ap$ and $\ndcg$ capture both precision and recall valance, we can see that lexirecall is comparable in fraction of ties.

Although the general trends are consistent across recommendation and retrieval contexts, recommendation metrics have substantially higher variance.  This results from the sparser set of labels in recommendation, where the total number of unique vectors, $\ndocs\choose\numrel$, is lower and therefore there is a higher probability of ties.

\begin{figure}
\centering
\begin{subfigure}[b]{\linewidth}
\includegraphics[width=0.95\linewidth]{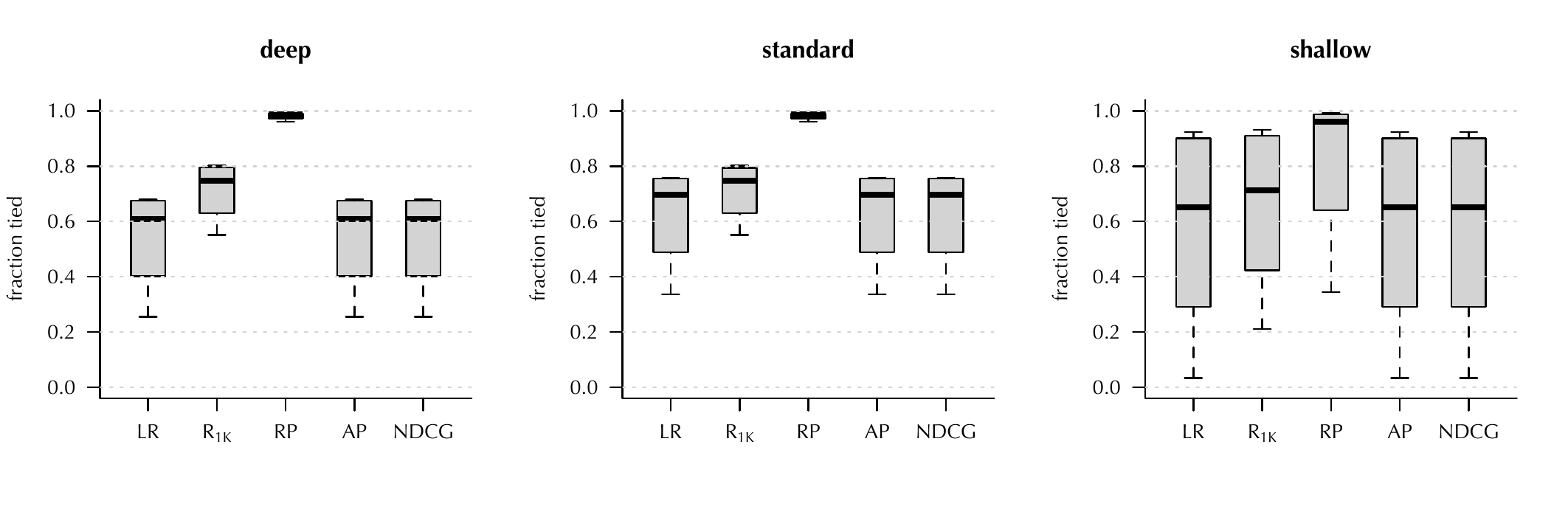}
\caption{Recommendation}\label{fig:numties:uniform:recsys}
\end{subfigure}
\begin{subfigure}[b]{\linewidth}
\centering
\includegraphics[width=0.95\linewidth]{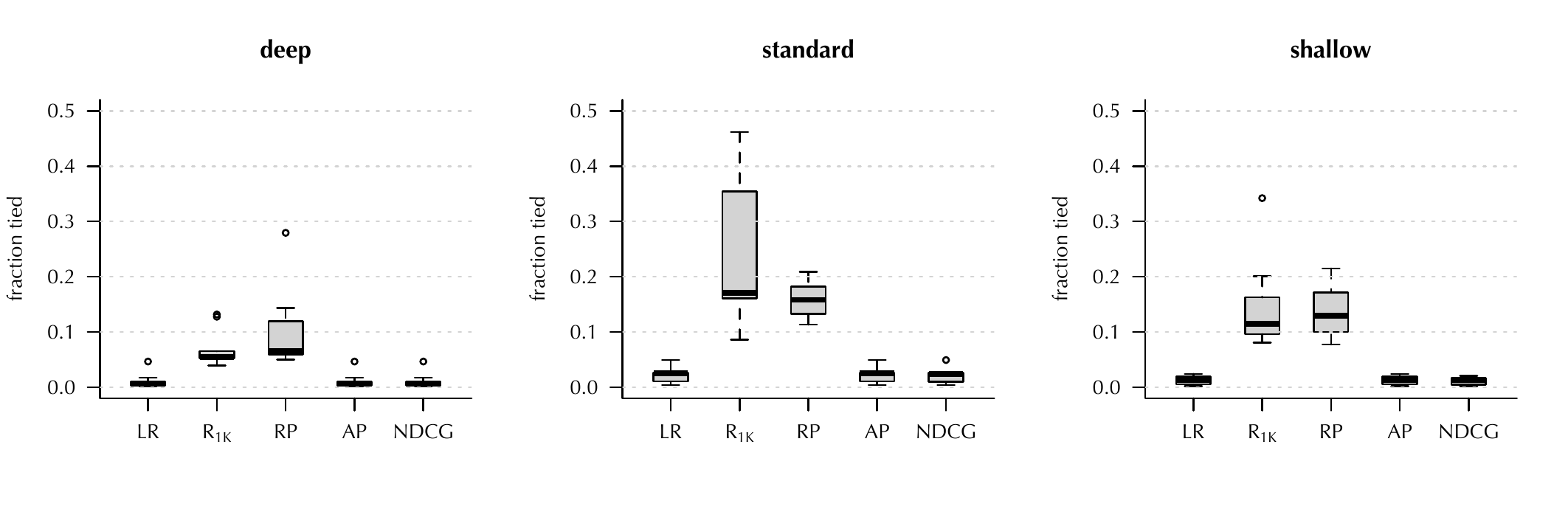}
\caption{Retrieval}\label{fig:numties:uniform:retrieval}
\end{subfigure}
\caption{Fraction of tied comparisons of rankings $\rlx,\rly$ in recommendation and retrieval  datasets.}\label{fig:numties}
\end{figure}
\subsubsection{Statistical Sensitivity}
\label{sec:results:sensitivity}
We are also interested in the ability of lexirecall to detect statistical differences between pairs of \textit{runs} (i.e. sets of rankings generated by a single system for a shared set of queries).  To do so, we adopted Sakai's method of measuring the discriminative power of a metric \cite{sakai:metrics}.  This approach measures the fraction of pairs of systems that the method detects statistical differences with $p<0.05$.  We use two methods to compute $p$-values.  In the first, we compute a standard statistical test and, correcting for multiple comparisons, measure the fraction of $p$-values below $0.05$.  For lexirecall, we adopt a binomial test since we have binary outcomes.  For other metrics, we adopt a Student's $t$-test, as recommended in the literature \cite{smucker:tests}.  We also conducted experiments using incorrect-but-consistent statistical tests with similar outcomes. In order to correct for multiple comparisons for all tests, we use the conservative Holm-Bonferroni method \cite{boytsov:multiple-tests}.  Our second method of computing $p$-values  uses Tukey's honestly significant difference (HSD) test as proposed by   \citet{carterette:multiple-testing}.  This method is considered a more appropriate approach to addressing multiple comparisons compared to our first approach. The goal of this analysis is to understand the statistical sensitivity of lexirecall compared to other recall-oriented metrics, while presenting non-recall metrics for reference.

We present the results of this analysis in Figures \ref{fig:metric-sensitivity:standard} and \ref{fig:metric-sensitivity:hsd}.   When using standard tests (Figure \ref{fig:metric-sensitivity:standard}),   lexirecall is slightly better at detecting significant differences compared to existing recall metrics at deeper retrievals.  We can refine this analysis by inspecting the HSD results (Figure \ref{fig:metric-sensitivity:hsd}).  In this case,  the sensitivity of lexirecall manifests more strongly, clearly more discriminative than existing recall metrics for deep retrievals, although losing this power as retrieval depth decreases.  This is consistent with previous observations for preference-based evaluation \cite{diaz:rpp,diaz:lexiprecision}.

\begin{figure}
\begin{subfigure}[b]{\linewidth}
\centering
\includegraphics[width=0.95\linewidth]{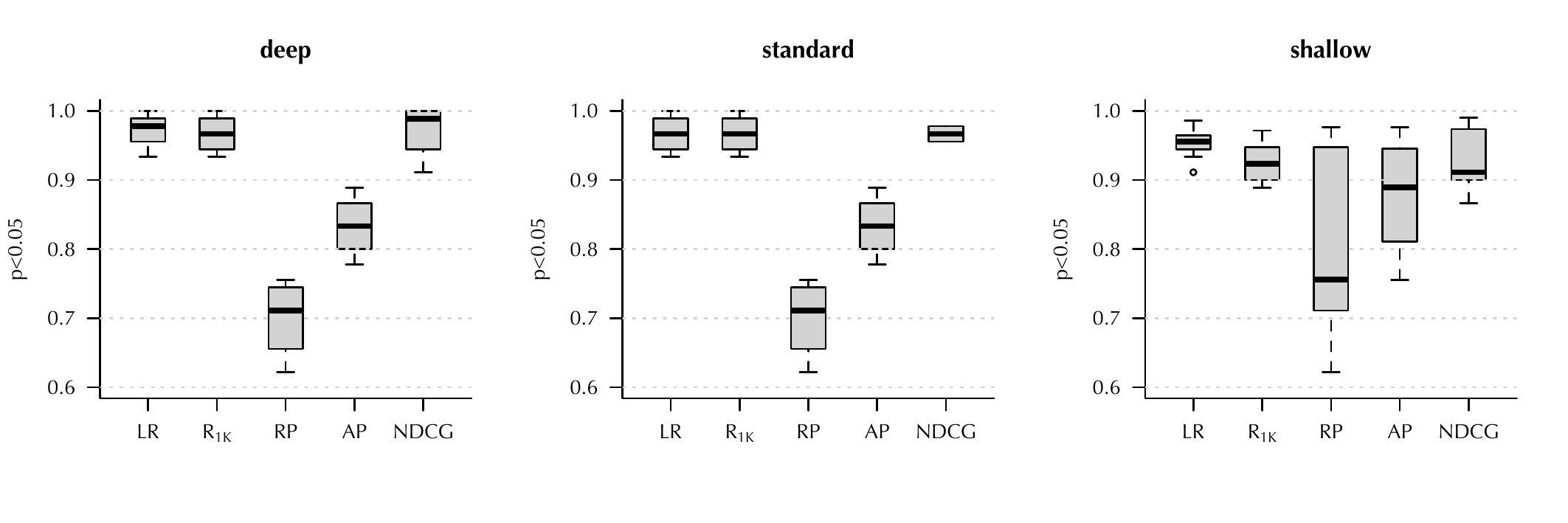}
\caption{Recommendation}\label{fig:metric-sensitivity:standard:recsys}
\end{subfigure}
\begin{subfigure}[b]{\linewidth}
\centering
\includegraphics[width=0.95\linewidth]{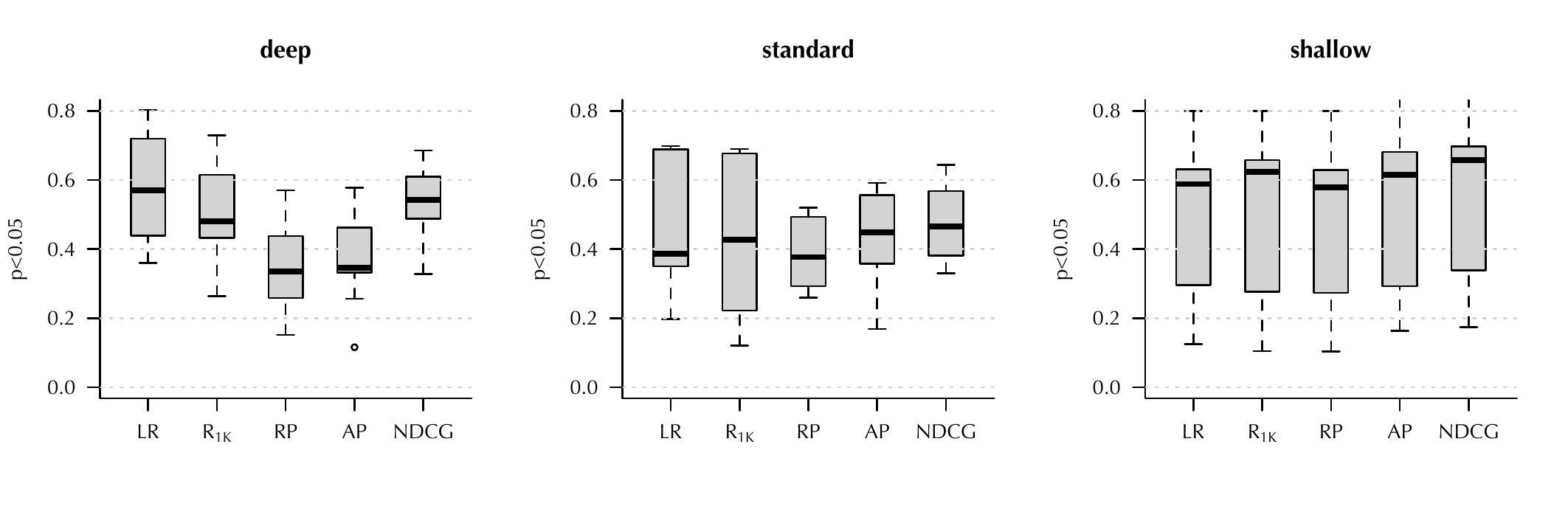}
\caption{Retrieval}\label{fig:metric-sensitivity:standard:retrieval}
\end{subfigure}
\caption{Statistical sensitivity.  Fraction of run pairs where we observe a statistically significant difference (i.e. $p<0.05$) using a paired test.  }\label{fig:metric-sensitivity:standard}
\end{figure}

\begin{figure}
\begin{subfigure}[b]{\linewidth}
\centering
\includegraphics[width=0.95\linewidth]{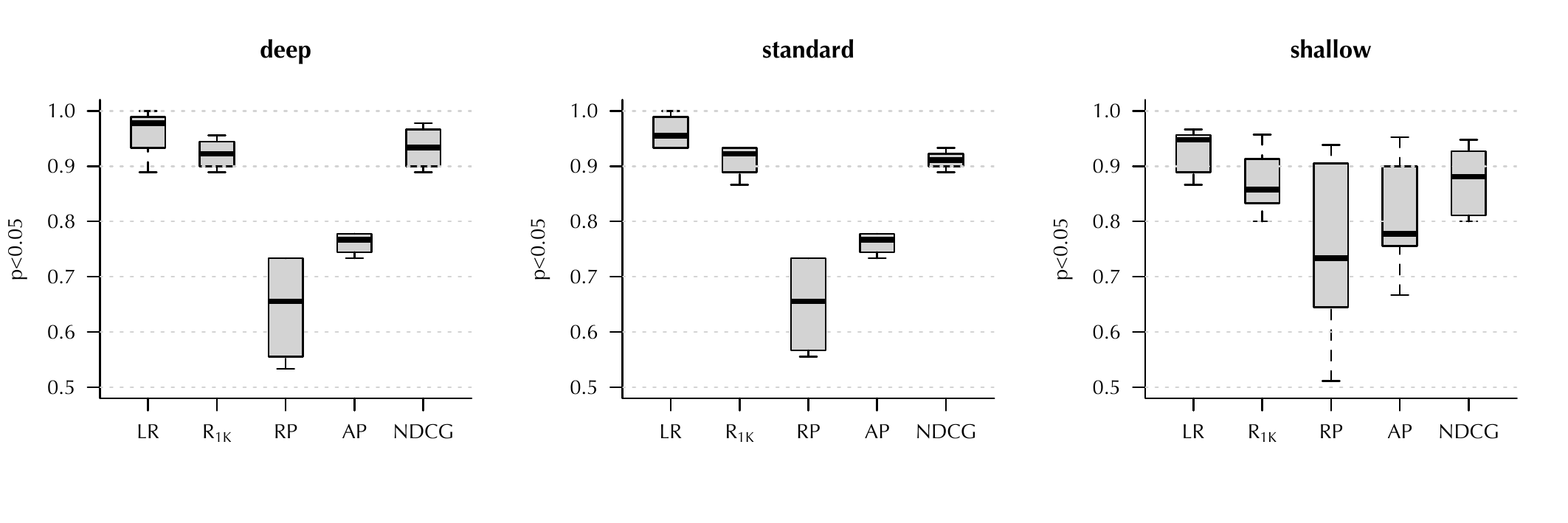}
\caption{Recommendation}\label{fig:metric-sensitivity:hsd:recsys}
\end{subfigure}
\begin{subfigure}[b]{\linewidth}
\centering
\includegraphics[width=0.95\linewidth]{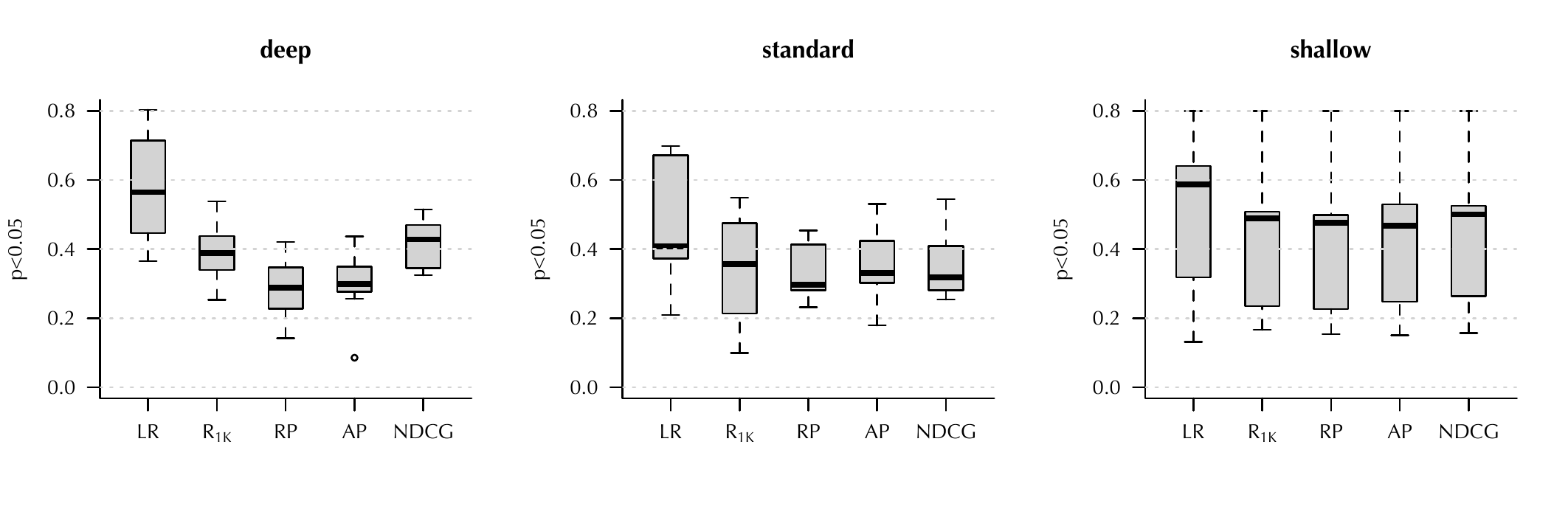}
\caption{Retrieval}\label{fig:metric-sensitivity:hsd:retrieval}
\end{subfigure}
\caption{Statistical sensitivity.  Fraction of run pairs where we observe a statistically significant difference (i.e. $p<0.05$) using Tukey's honestly significant difference test.  }\label{fig:metric-sensitivity:hsd}
\end{figure}
\subsubsection{Label Degradation}
\label{sec:results:degradation}
Effective evaluation methods are stable in the presence of missing relevance labels.  In this analysis, we held the number of queries fixed and randomly removed a fraction of relevant items per query, leaving at least one relevant item per query. We uniformly sample from all relevant items $\relset$ to remove.  We present results for the effect of label degradation on the fraction of ties (we expect more ties with fewer labels) and preference agreement with full data (we expect lower agreement with fewer labels).  As with the previous section, the goal of this analysis is to compare lexirecall  to other recall-oriented metrics, while presenting non-recall metrics for reference.

In terms of fraction of ties (Figure \ref{fig:degration:numties:uniform}), lexirecall degrades comparably to metrics like $\ap$ and $\ndcg$ and substantially more gracefully compared to existing recall metrics $\recallK$ and $\rprecision$.  While the importance of relevance labels for recall-oriented evaluation is important, this result suggests that existing metrics are extremely brittle when labels are missing.  All methods observed more ties at shallower retrieval depths with degradation more pronounced for traditional recall-oriented metrics.  As with our analysis of the number of ties (Section \ref{sec:results:numties}), because of the sparsity of labels in recommendation, we observe much higher variance compared to retrieval, where labels are more complete due to pooling \cite{sparck-jones:pooling,zobel:pooling}.  As such, the recommendation tasks can be seen as operating in the far right-hand portion of retrieval experiments.

\begin{figure}
\centering
\begin{subfigure}[b]{\linewidth}
\centering
\includegraphics[width=0.90\linewidth]{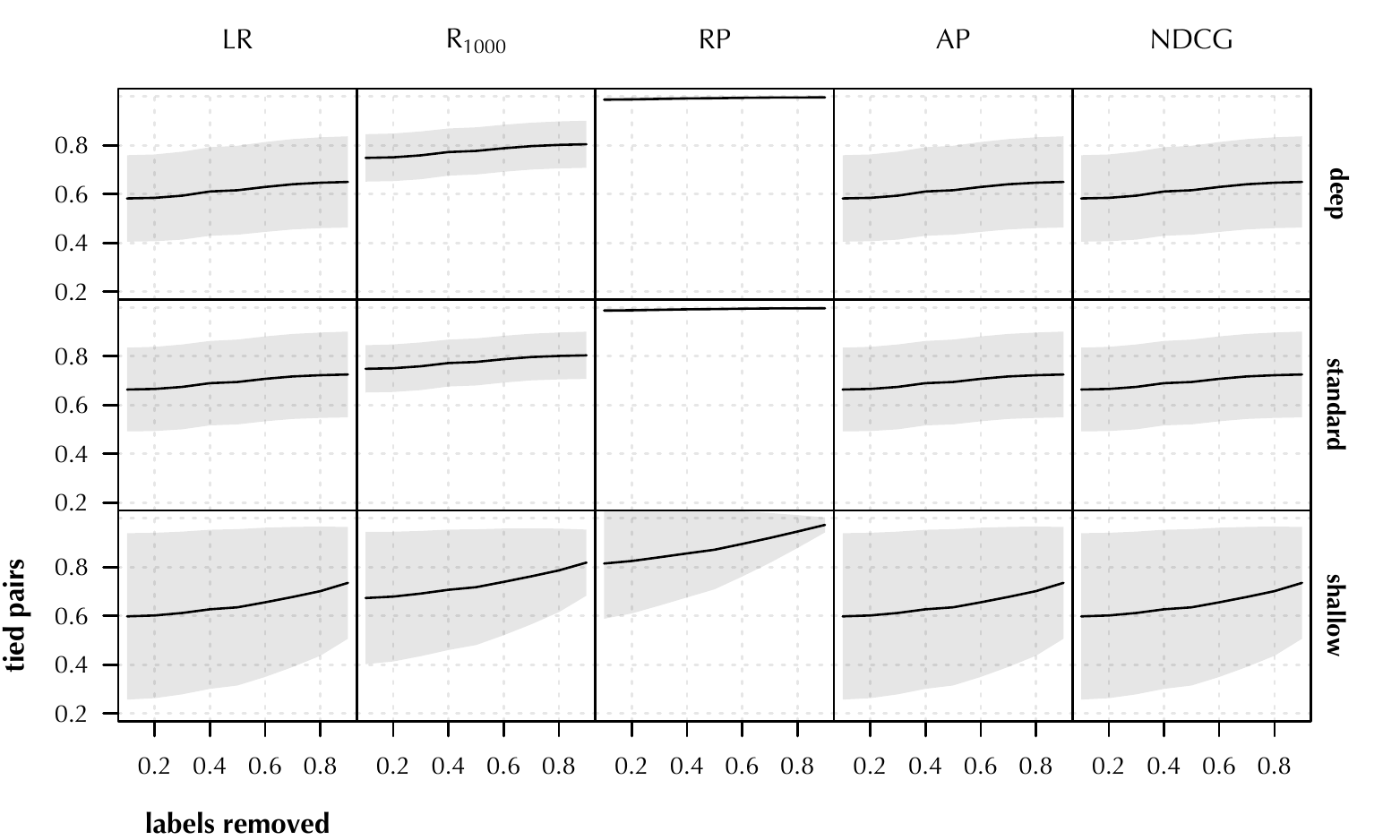}
\caption{Recommendation}\label{fig:degration:numties:uniform:recsys}
\end{subfigure}
\begin{subfigure}[b]{\linewidth}
\centering
\includegraphics[width=0.90\linewidth]{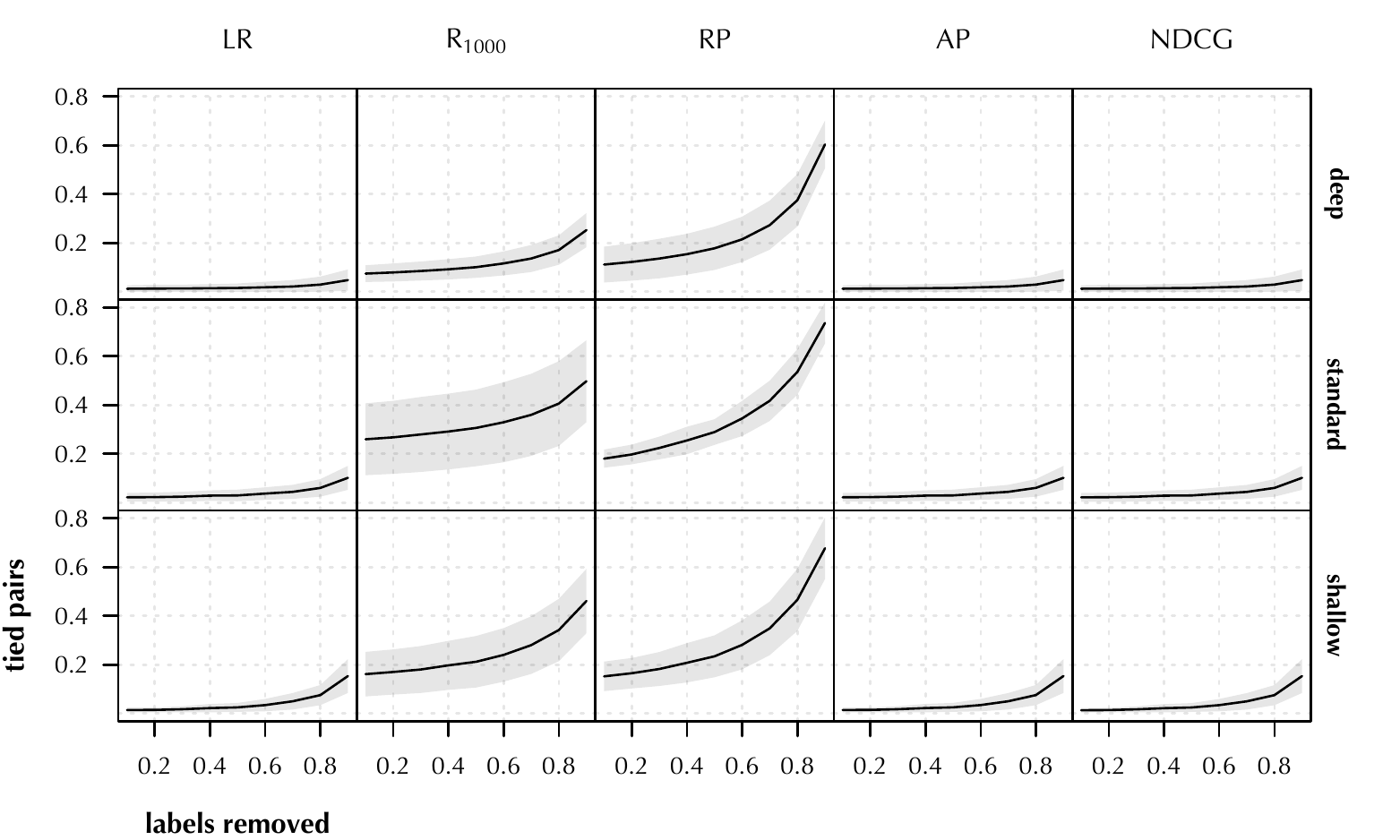}
\caption{Retrieval}\label{fig:degration:numties:uniform:retrieval}
\end{subfigure}
\caption{Number of ties as labels removed.  For each query, we removed a fraction of items from $\relset$ and counted the number of tied rankings.  Relevant items were  sampled uniformly at random.  Average over ten samples. }\label{fig:degration:numties:uniform}
\end{figure}

\begin{figure}
\centering
\begin{subfigure}[b]{\linewidth}
\centering
\includegraphics[width=0.90\linewidth]{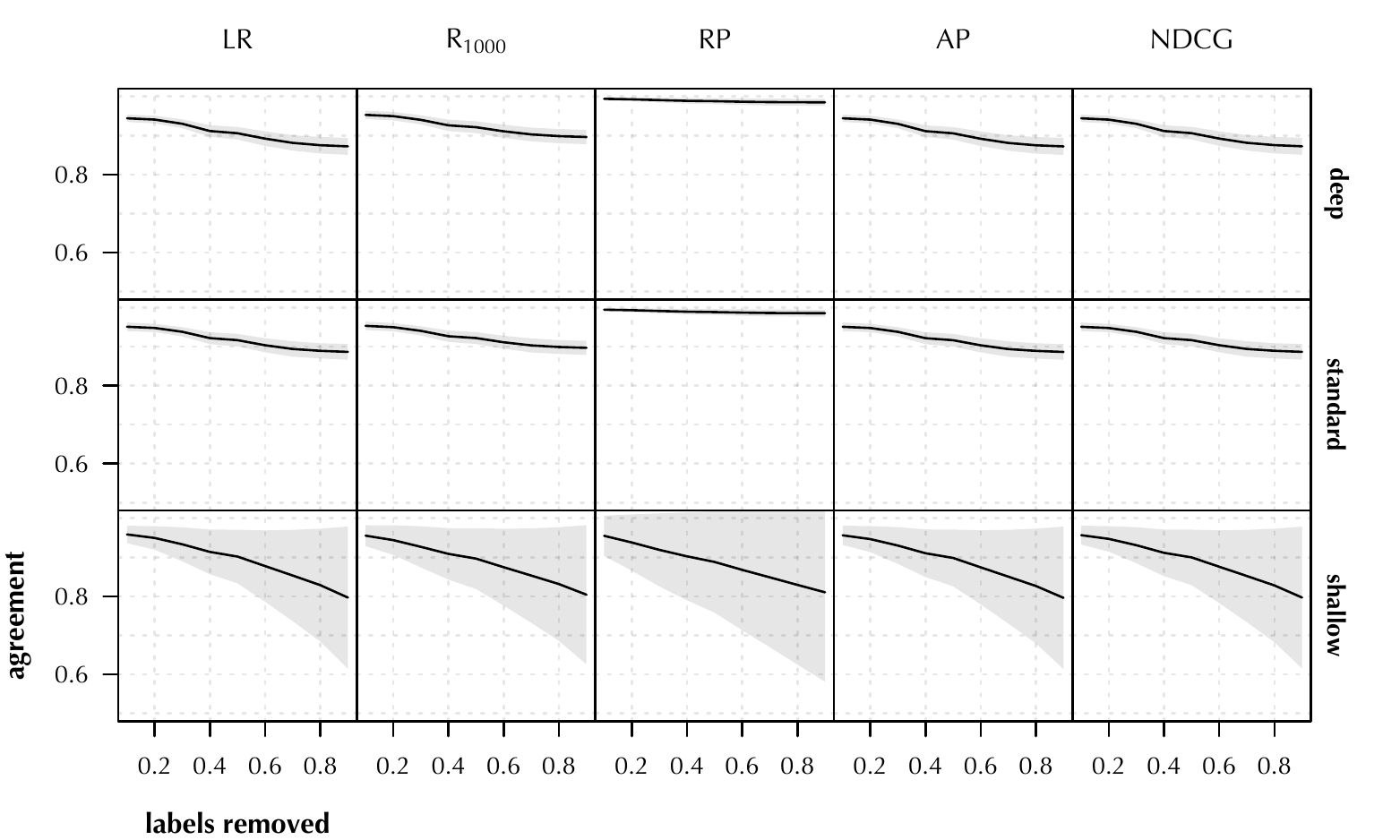}
\caption{Recommendation}\label{fig:degration:agreement:uniform:recsys}
\end{subfigure}
\begin{subfigure}[b]{\linewidth}
\centering
\includegraphics[width=0.90\linewidth]{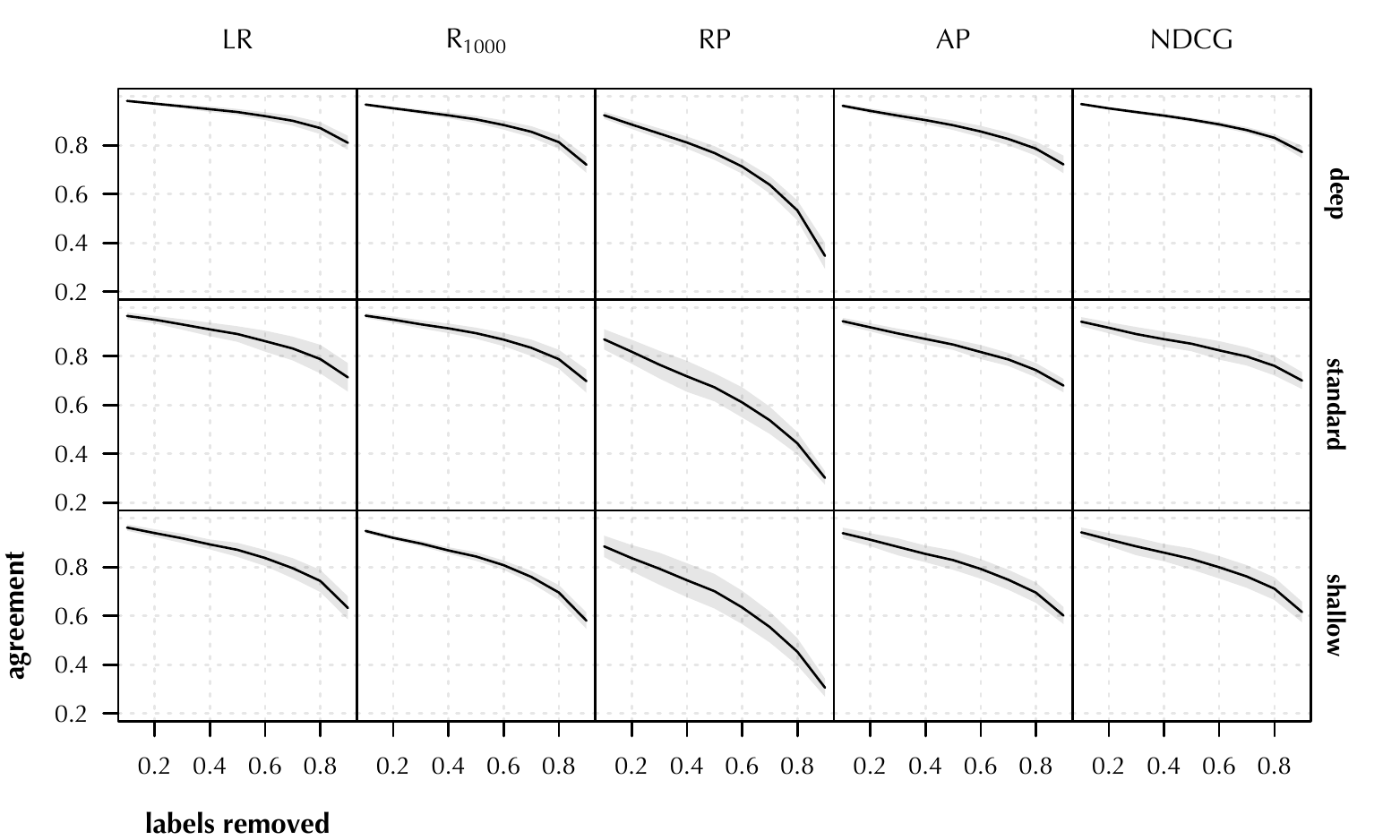}
\caption{Retrieval}\label{fig:degration:agreement:uniform:retrieval}
\end{subfigure}
\caption{Preference agreement between metric based on degraded labels and complete labels. For each query, we removed a fraction of items from $\relset$ and measured the fraction of preferences $\succ_{\metric}$ based on the incomplete relevance judgments that agreed with the preference when using the complete set of relevant items.  Relevant items were  sampled uniformly at random.  Average over ten samples. }\label{fig:degration:agreement:uniform}
\end{figure}

In terms of agreement with preferences based on full data (Figure \ref{fig:degration:agreement:uniform}), lexirecall again degrades comparably to $\ap$ and $\ndcg$.  In recommendation experiments, $\rprecision$ and $\recallK$ behave comparably or better than lexirecall, $\ap$, and $\ndcg$.  However, this may be an artifact of the poor behavior of $\rprecision$ and $\recallK$ under extremely sparse labels, as demonstrated in earlier experiments.  That is, since these metrics observe so many ties and ties increase with label degradation, accurately predicting these preferences under `full data' will be easier with label degradation.  We contrast this with retrieval experiments, where labels are more complete. In this condition, $\rprecision$ is much more sensitive to degradation, dropping in performance quickly.  On the other hand, $\recallK$ behaves similar to lexirecall when preserving most labels, but, for drastically sparse labels, the agreement drops.  We again see slightly worse degradations with shallower retrieval depths across all metrics.  $\rprecision$ in particular demonstrates significantly worse degradation compared to all metrics, while $\recallK$ shows worse degradation when removing relevant items based on ranking frequency.

\section{Discussion}
\label{sec:discussion}
We begin our discussion by returning to the desiderata in Section \ref{sec:desiderata}. We originally sought to define and understand recall from a more formal grounding, allowing us to draw connections to recent literature in fairness and robustness, supporting our first desideratum.  Moreover,  our definition of recall orientation both directly implied the appropriate recall metric and differentiated it from existing metrics, supporting our second desideratum.  Finally, our empirical analysis demonstrated that lexirecall captures many of the properties of existing metrics, while being substantially more sensitive and stable in the presence of missing labels, supporting our remaining desideratum.  Collectively, we find strong support for investigating lexirecall as a method for assessing our robustness perspective of recall.

In light of our conceptual, theoretical, and empirical analysis, we can make a number of other observations about recall, robustness, and lexicographic evaluation.

\subsection{Recall}
\subsubsection{Labeling}
Although lexirecall appears more stable in the presence of missing labels than existing recall-oriented metrics,  the performance of recall and robustness evaluation depends critically on having comprehensive relevance labels.  This situation is exacerbated in recommender system environments where judgments, while often highly personalized and based on psychological relevance, can be extremely incomplete (see \cite{herlocker:recsys-eval} for a discussion).    This suggests two possibilities.  Sparsity in recommendation can arise from
\begin{inlinelist}
\item very niche relevance or
\item severe under-labeling of relevant items.
\end{inlinelist}
In the first situation, our experiments suggest that lexirecall is substantially more sensitive than existing recall metrics.  However, in the second situation is accurate, we can examine metric behavior under retrieval scenarios, where labels are more complete.  Indeed, our experiments would suggest that using existing recall metrics for sparse recommendation tasks does not accurately reflect recall under more complete data.  This echos similar observations in the recommender system community for addressing data sparsity issues \cite{ekstrand:sturgeon,tian:estimating-error,ihemelandu:candidate-set-sampling}.

While time-consuming, we believe that, in order to develop robust and fair systems, new techniques for expanding labeled sets for recall evaluation are necessary.    In retrieval contexts, initiatives like TREC adopt pooling as a strategy to achieve more complete judgments \cite{sparck-jones:pooling}.  Unfortunately, in situations where relevance is derived from behavioral feedback (e.g. \cite{kelly:implicit-feedback-survey,diaz:neurips-2020-tutorial}), comprehensiveness of relevance is often not the focus.  Moreover, the transience of information needs in production environments makes reliable detection of $\relset$ an open problem.  This is compounded by the desire---in both retrieval and recommendation contexts---to estimate relevance for evaluation and optimization transparently to the user. As such, new methods for estimating relevance of a broader set of items without impacting the user are necessary.

\subsubsection{Depth Considerations for Recall-Oriented Evaluation}
All of our recall-oriented evaluations (e.g. lexirecall, $\rprecision$, $\recallK$) suffer  when operating within a shallow retrieval environment.  We recommend that, especially for recall-oriented evaluation, retrieval depths be high, regardless of the specific evaluation method.  Moreover, as labels become sparser, both $\rprecision$ and $\recallK$ show substantially more ties (Figure \ref{fig:degration:numties:uniform}) and poor stability in the presence of missing labels (Figure \ref{fig:degration:agreement:uniform}).  We recommend that, for shallow retrieval with sparse labels, $\rprecision$ should be avoided altogether.

\subsection{Robustness}
\subsubsection{Number of Ties and Metric-Based Evaluation}
The high number of ties from $\rprecision$ and $\recallK$ arises when collapsing all permutations that share the same recall value.  Top-heavy recall-level metrics that have nonzero weight over all relevant items effectively encode the $\ndocs\choose \numrel$ permutations onto the real line.  We should expect more ties and lower statistical sensitivity for metrics that have low cutoffs (e.g. $\topk<100$).  This includes $\rr$ and variants of top-heavy recall-level metrics with rank cutoffs (e.g. $\ndcgTen$).  Even for  top-heavy recall-level metrics, we expect ties if $\exposure(i)\approx0$ for unretrieved items or if the numerical precision limits the ability to represent all $\ndocs\choose \numrel$ positions of relevant items.  Because of their top-heaviness, these ties are more likely to occur for differences at the lower ranks, precisely the positions worst-case performance emphasizes.  As a result, even though some top-heavy metrics may theoretically \textit{include} worst-case performance, they will not \textit{emphasize} it in the metric value.

\subsubsection{Mixed Orientation Metrics}
We saw agreement between lexirecall and $\ndcg$, even though the latter captures precision orientation (Figure \ref{fig:metric-orientation:results}).  In Section \ref{sec:robustness:analysis}, we explained that this may be due to either the inclusion of $\RPx_\numrel$ in top-heavy recall-level metric computation (Equation \ref{eq:metric}) or because of structural dependencies between rank positions of $\RPx_i$ and $\RPx_j$.  Alternatively, since we observed strong empirical agreement between lexirecall and $\ndcg$,  the position of the last ranked relevant item may be predictable because of systematic behavior in the model.  For example, for many scenarios, performance higher in the ranking may be predictive of worst-case performance.  Even if this is case for many systems or domains, we caution against presuming that performance at the top of the ranking is predictive of worst-case performance.  If the worst-case performance is systematic and amongst smaller-sized groups (i.e. those unlikely to appear at the top), then the performance will not be well-predicted by larger, systematically-higher ranked items from dominant groups.  We recommend lexirecall to detect worst-case performance in isolation of other criteria (e.g. precision).

\subsubsection{A Comment on Graded Metrics}
\label{sec:robustness:searchers:grades}
Although we have focused on binary relevance, many ranking scenarios use graded or ordinal relevance.  Consider relevance labels represented as an ordinal scale, where higher grades reflect a higher probability of satisfying the information need according to the rater's subjective opinion, as happens in many retrieval scenarios.
Under such a grading scheme, an item labeled with the minimum grade has a probability of relevance of $0$ (i.e. no user would ever find the item relevant) and an item labeled with the maximum grade has a probability of relevance of $1-\epsilon$ (i.e. \textit{almost} every user would find the item relevant).
We can determine grades that reflect the probability of relevance  through
\begin{inlinelist}
\item labeling instructions (e.g. `an item with a high grade should satisfy many users; an item with a medium grade should satisfy some users; an item with a low grade should satisfy few users'),
\item voting schemes \cite{gordon:disagreement-convolution}, or
\item aggregated behavioral data (e.g., clickthrough rate) \cite{zheng:clicks-relevance}.
\end{inlinelist}
No matter how grades are determined, for a fixed request, a user with a less popular intent may \textit{not} be satisfied by an item relevant to a more popular intent.  This implies that notions of optimality in graded ranking evaluation (i.e. that higher grades should be ranked above lower grades) explicitly values dominant group intents over minority intents.  From the perspective of robustness, this means that, for graded judgments, we should consider $\relset$ to include all items with a non-zero chance of satisfying a user.  This is precisely the approach adopted for $\tse$ and lexirecall.

\subsubsection{Robustness of Optimal Rankings}
\label{sec:robustness:searcher:optimal}
In the case of binary relevance, \citet{robertson:prp}'s Probability Ranking Principle suggests that an optimal ranking will place all relevant items above nonrelevant items. If  $\optimalRankings$ is the set of optimal permutations, then it consists of permutations that rank all of the items in $\relset$ above $\docset-\relset$.  Traditional ranking methods have largely been deterministic insofar as, given a request, they always return a fixed $\ranking\in\rankings$.  The \textit{optimal deterministic ranker}, then, is any ranker that selects a fixed $\ranking\in\optimalRankings$ and,  therefore, for any optimal deterministic ranker, the worst-case performance is $\exposure(\numrel)$.

The situation changes if we consider stochastic rankers \cite{radlinski:active-exploration,bruch:stochastic-ltr,diaz:expexp,singh:exposure,oosterhuis:policy-aware-unbiased-ltr}, systems that, in response to a request, sample a ranking $\ranking$ from some distribution over $\rankings$.  Such systems have been proposed in the context of online learning \cite{radlinski:active-exploration} and fair ranking \cite{diaz:expexp}.  An \textit{optimal stochastic ranker} would, in response to a request, uniformly sample a ranking from $\optimalRankings$.

We can show that the worst-case expected  performance of the optimal stochastic ranker is better than the worst-case expected performance for \textit{any} optimal deterministic ranker,
\begin{align*}
\min_{\user\in\users}\expectation{\ranking\sim\optimalRankings}{\metric(\ranking,\user)}&\geq\min_{\user\in\users}\metric(\dranking,\user), \forall \dranking\in\optimalRankings
\end{align*}
We present a proof in Appendix \ref{app:proofs:opt}.
Figure \ref{fig:stochastic} displays the difference in worst-case performance between optimal deterministic and stochastic rankings for $\numrel\in[1\isep 25]$.  This result provides evidence from a robustness perspective that ranking system design should explore the design space of stochastic rankers.

\begin{figure}
\includegraphics[width=\linewidth]{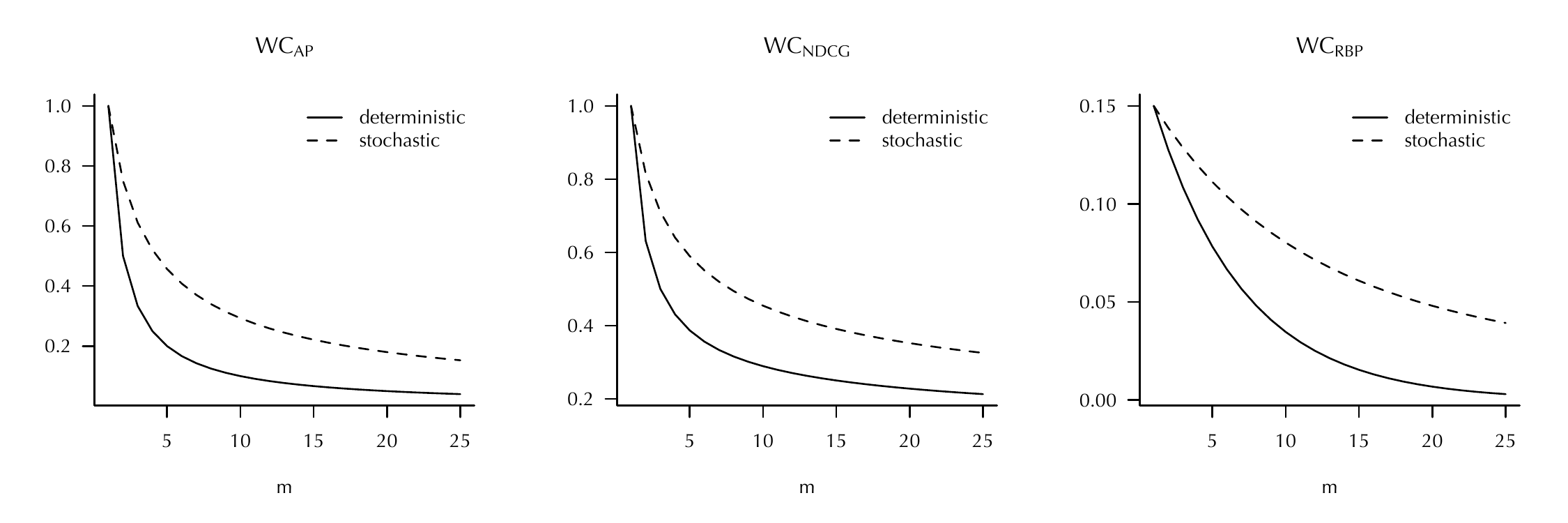}
\caption{Worst-case performance of optimal deterministic and stochastic rankings for $\numrel\in[1\isep 25]$.}\label{fig:stochastic}
\end{figure}

\subsection{Lexicographic Evaluation}

\subsubsection{Recovering an Evaluation Metric}
\label{sec:leximin:metric}
In contrast with preference-based evaluation like lexirecall,  metric-based evaluation can be performed efficiently for each ranking independently, moving the complexity from $O(|\rldata|\log |\rldata|)$ to $O(|\rldata|)$.  Fortunately, existing results in the computation of leximin point to how to design such a metric.

\citet{yager:analytic-leximin} demonstrates one can construct a leximin representation of a ranking such that $\leximin(x)>\leximin(y)\leftrightarrow x\leximinpref y$.  Specifically, if $x$ is a $\numrel\times 1$  allocation vector sorted in decreasing order,
\begin{align}
\leximin(x)&=\sum_{i=1}^\numrel w_ix_i
\end{align}
where $w$ is a \textit{bottom-heavy} weight vector such that $w_1\ll w_2\ll\ldots\ll w_\numrel$.  In our situation, a system `allocates' exposure but, because $\exposure(\RPx_i)$ is a monotonically decreasing function of $\RPx_i$, we only need to compare the rank positions $\RPx$.  As such, we can define our allocation vectors as $x_i=\frac{\numdocs-\RPx_i}{\numdocs}$.  We can use this to define the recall level weight vector $w$ as,
\begin{align*}
w_i&=\begin{cases}
\frac{\Delta^{\numrel-1}}{(1+\Delta)^{\numrel-1}} & i=1\\
\frac{\Delta^{\numrel-i}}{(1+\Delta)^{\numrel+1-i}} & i>1
\end{cases}
\end{align*}
where $\Delta=(\numdocs+\epsilon)^{-1}$ and $\epsilon\in(0,1)$ is a free parameter.  
We can then define \textit{metric lexirecall} as,

\begin{align*}
\lexirecall(\RPx)&=\leximin(x)\\
&\propto -\sum_{i=1}^{\numrel}  w_i\RPx_i
\end{align*}
We can then compare rankings directly using $\leximin(\RPx)$.  Since $\sum_iw_i=1$, this can be interpreted as the bottom-heavy weighted average of the positions of relevant items.  We contrast this with uniform weighting found in the recall error metric \cite{rocchio:recall-error} or top-heavy weighting found in precision-oriented metrics.

While providing interesting theoretical connection to existing top-heavy recall-level metrics, in practice, due to the large values of $\numdocs$ and $\RPx_i$, computing the metric  lexirecall can  suffer from numerical precision issues.  When $\numdocs$ is unknown---for example in dynamic or extremely large corpora---metric lexirecall cannot be calculated at all.  In these situations, computing lexirecall  is feasible due to our imputation procedure and the fact that we only care about relative positions.

\subsubsection{Optimization}
Although our focus has been on evaluation, optimizing for lexicographic criteria may be an alternative method for designing recall-oriented algorithms, for example for technology-assisted review or  candidate generation.  One way to accomplish this is to optimize for metric lexirecall  discussed in the previous section.   Since learning to rank methods often optimize for functions of positions of relevant items  (e.g. \cite{qin:approxndcg}), standard approaches may suffice.  Alternatively, in Section \ref{sec:robustness:searcher:optimal}, we observed that optimal stochastic rankers outperformed optimal deterministic rankers in terms of worst-case performance.  This suggests that stochastic ranking techniques similar to those developed in the context of other fairness notions (e.g. \cite{diaz:expexp}) can be used for recall-oriented tasks.
\section{Conclusion}
\label{sec:conclusion}
Our analysis indicates that recall, when interpreted as retrieving the totality of relevant items, has conceptual and theoretical connections to worst-case robustness of retrieval effectiveness across possible users.
By providing a clear definition of recall-orientation, we have directly captured the recall-orientation of existing metrics and recognized a missing `basis metric' for recall.  Moreover, by developing $\tse$ as the counterpart to $\rr$, we  directly connected recall-orientation to robustness and Rawlsian notions of fairness, providing a normative argument for improving techniques for gathering complete relevance judgments.  Doing so helps ensure the effective computation of recall and, in turn, address potential unfairness.  To effectively deploy $\tse$, we adopted lexicographic evaluation techniques and introduce the lexirecall preference-based evaluation method, which we empirically demonstrated was preferable to existing recall metrics.  We anticipate that variants of lexicographic evaluation can be applied for other constructs.  These three themes of recall, robustness, and lexicographic evaluation, while each individually potentially being interesting areas of theoretical analysis, work collectively to substantially improve our understanding of recall, a concept that may be as old as the field of information retrieval and recommender systems.

\appendix
\bibliographystyle{ACM-Reference-Format}
\bibliography{XX-recall}
\newpage
\begin{table}
\caption{Notation}\label{tab:notation}
\begin{tabular}{lp{0.70\textwidth}}
$\Pintegers$    &   positive integers\\
$\NNintegers$    &   non-negative integers\\
$\Preals$    &   positive reals\\
$\NNreals$    &   non-negative reals\\
$S_k$ &   set of all permutations of $k$ items\\
$\fdset{A}^+$   &   non-empty subsets of $\fdset{A}$\\

$\ident:X\rightarrow\{0,1\}$   &   indicator function (i.e., returns 1 if $X$ is true; 0 otherwise)\\
$x_{-i}$   &   reverse index (i.e., $x_{k-i+1}$ for the $k$-dimensional vector $x$)\\
\\

$\docset$   &   corpus\\

$\relset$   &   relevant set \\

$\numdocs$  &   size of corpus (i.e., $|\docset|$)\\

$\numrel$  &   size of relevant set (i.e., $|\relset|$)\\
$\topk$  &   number of items retrieved\\
\\
$\rankings$   &   set of all permutations of $\docset$ \\
$\rldata$   &   subset of $\rankings$ generated by multiple systems for a fixed request \\

\\
$\ranking$  &   permutation of $\docset$\\
$\RP$   &   sorted positions of relevant items\\
$\invRP$ & item ids of relevant items sorted by position\\
$\numcollapsed$ & number of unique permutations in $\rankings$ for a given $\RP$ (i.e., $\numrel!(\numdocs-\numrel)!$)\\
\\

$\metric:\rankings\times\docsets\rightarrow\NNreals$   & user evaluation metric\\
$\pmetric:\rankings\times\docsets\rightarrow\NNreals$   & provider evaluation metric\\
$\metric(\ranking,\relset)\rankeq\metric'(\ranking,\relset)$ &  $\metric$ and $\metric'$  rank $\rankings$ identically\\
\\
$\exposure : \rankPositions\rightarrow \NNreals$ &    exposure of position\\
$\normalization : \rankPositions\times\rankPositions\rightarrow \NNreals$ &   normalization function\\
\\
$\users$ &   set of possible user information needs for a request\\
$\providers$ &   set of possible providers for a request\\

\end{tabular}
\end{table}

\section{Full Papers Published at the ACM Conference on Recommender Systems that Measure Recall}
\label{app:recsys-recall}
We downloaded and reviewed all full papers published at the ACM Conference on Recommender Systems (2007-2024), omitting reproducibility, late-breaking results, short, demo, doctoral, or keynote papers.  A paper was identified if recall was quantitatively measured as part of experimentation.  Conceptual discussions of recall (e.g., in comparison to experimentation metrics) were not considered to use recall.  We identified only those papers that specifically name recall (e.g., $\recall_{k}$), although some operationalizations of  hit rate may be equivalent to recall \cite{tamm:recsys-metric-computation}.   We present the full papers identified in Table \ref{tab:recsys-recall}.

{ 
\renewcommand{\arraystretch}{1.1}
\begin{table}
\centering
\caption{DOIs of full papers published at the ACM Conference on Recommender Systems (2007-2024) that measure recall in experiments. }\label{tab:recsys-recall}
{\footnotesize
\begin{tabular}{cccp{4in}}
year & total & recall & DOIs\\
\hline
2007 & 16 & 3 & 10.1145/1297231.12972\{48, 50, 51\}\\
2008 & 20 & 3 & 10.1145/1454008.14540\{17, 32, 34\}\\
2009 & 25 & 4 & 10.1145/1639714.16397\{24, 26, 37, 45\}\\
2010 & 25 & 8 & 10.1145/1864708.18647\{21, 28, 31, 32, 41, 44, 45, 47\}\\
2011 & 22 & 5 & 10.1145/2043932.20439\{41, 45, 47, 57, 65\}\\
2012 & 24 & 11 & 10.1145/2365952.23659\{62, 63, 67, 68, 69, 72, 73, 76, 79, 82, 89\}\\
2013 & 32 & 11 & 10.1145/2507157.2507\{160, 170, 171, 172, 182, 184, 186, 209, 210, 211, 215\}\\
2014 & 35 & 8 & 10.1145/2645710.26457\{21, 23, 29, 34, 38, 40, 43, 46\}\\
2015 & 28 & 7 & 10.1145/2792838.28001\{70, 72, 74, 76, 85, 89, 95\}\\ 
2016 & 29 & 12 & 10.1145/2959100.29591\{33, 37, 46, 47, 49, 51, 57, 67, 70, 78, 80, 82\}\\ 
2017 & 26 & 6 & 10.1145/3109859.3109\{877, 879, 887, 896, 900, 903\}\\ 
2018 & 32 & 11 & 10.1145/3240323.3240\{343, 347, 349, 355, 363, 371, 372, 374, 381, 391, 405\}\\ 
2019 & 36 & 11 & 10.1145/3298689.334\{6996, 7002, 7007, 7009, 7012, 7013, 7026, 7036, 7044, 7058, 7065\}\\ 
2020 & 39 & 12 & 10.1145/3383313.341\{1476, 2232, 2235, 2243, 2247, 2248, 2249, 2258, 2259, 2262, 2265, 2268\}\\ 
2021 & 49 & 19 & 10.1145/3460231.3474\{228, 230, 234, 238, 240, 242, 249, 252, 255, 257, 260, 263, 265, 266, 268, 270, 272, 273, 275\}\\ 
2022 & 39 & 13 & 10.1145/3523227.3546\{752, 754, 755, 760, 762, 763, 768, 770, 771, 775, 782, 784, 785\}\\ 
2023 & 47 & 21 & 10.1145/3604915.3608\{766, 771, 773, 781, 783, 784, 785, 786, 803, 804, 806, 809, 810, 811, 812, 815, 863, 868, 871, 878, 882\}\\ 
2024 & 58 & 24 & 10.1145/3640457.3688\{096, 098, 100, 104, 108, 109, 113, 117, 121, 122, 123, 124, 125, 127, 133, 137, 138, 139, 145, 146, 148, 149, 151, 153\}  
\end{tabular}
}
\end{table}
}

\section{Metric Properties}
\label{app:proofs:metrics}

We can connect our proposed metrics to prior work by leveraging several properties defined in the community.  Drawing on fundamental contributions from \citet{moffat:seven-properties-of-metrics}, \citet{ferrante:framework-for-utility-metrics}, and \citet{amigo:metric-axioms}, we produced a synthesized list of properties.  These properties should be considered descriptive---not prescriptive---since
\begin{inlinelist}
\item they can be in tension, and
\item there are valid metrics (e.g., high-precision, diversity, fairness) that do not satisfy all of them.
\end{inlinelist}
Let $\topkranking$ be a top-$\topk$ from which we impute a total ranking using pessimistic imputation (Section \ref{sec:preliminaries:imputation}).

Consistent with our setup, we assume that the evaluator has access to a set of fixed relevance assessments $\relset$ over a fixed corpus $\docset$ and pessimistic imputation of a top-$\topk$ ranking $\topkranking$.  Given a ranking $\ranking$, let $\qrels_i$ refer to the relevance of $\ranking_i$.   Finally, we assume that the number of unretrieved $\numdocs-\topk$ items is larger than the number of relevant items $\numrel$, which is the case in most ranking tasks.

\begin{enumerate}
\item \textbf{Monotonicity in retrieval size \cite{moffat:seven-properties-of-metrics,amigo:metric-axioms}.} This property refers to the behavior of $\metric$ as we append items to $\topkranking$.   Following \citet{moffat:seven-properties-of-metrics}, a metric is \textit{monotonically increasing in retrieval size} if it is non-decreasing as $\topk$ is increased by appending either relevant or nonrelevant items to $\topkranking$.  Following \citet{amigo:metric-axioms}, a metric is \textit{(strictly) monotonically decreasing in nonrelevance} if it (strictly) decreases as $\topk$ is increased by appending nonrelevant items to $\topkranking$; \citet{amigo:metric-axioms} refer to the strict version of this as `Confidence'.  Note that the \citet{moffat:seven-properties-of-metrics} and \citet{amigo:metric-axioms} properties are in tension.  For completeness, we refer to a a metric as \textit{(strictly) monotonically increasing in relevance} if it (strictly) increases as $\topk$ is increased by appending relevant items to $\topkranking$.
\item \textbf{Monotonicity in swapping up \cite{moffat:seven-properties-of-metrics,amigo:metric-axioms,ferrante:framework-for-utility-metrics}.} A metric is \textit{(strictly) monotonically increasing in swapping up} if, when  $i<j$ and $\qrels_i < \qrels_j$,  we observe a (strictly) monotonic increase in $\metric$ when we swap the items at positions $i$ and $j$.  \citet{ferrante:framework-for-utility-metrics} refers to the non-decreasing property as `Swap'.  When $j>\topk$, \citet{moffat:seven-properties-of-metrics} refers to the strictly increasing property as `Convergence' and \citet{ferrante:framework-for-utility-metrics} refers to the non-decreasing property as `Replacement'.  When $j\leq\topk$, \citet{moffat:seven-properties-of-metrics} refers to the strictly increasing property as `Top-weightedness'.  For contiguous swaps (i.e.,  $j=i+1$), \citet{amigo:metric-axioms} refer to the strictly increasing property as `Priority'.
\item \textbf{Concavity in contiguous swap depth \cite{amigo:metric-axioms}.} A metric is \textit{(strictly) concave in contiguous swap depth} if, when $i<j$ and $\qrels_i < \qrels_{i+1}$ and $\qrels_{j} < \qrels_{j+1}$, swapping $i$ and $i+1$ will lead to a (strictly) larger improvement in $\metric$ compared to  swapping $j$ and $j+1$.  \citet{amigo:metric-axioms} refer this to `Deepness'.
\item \textbf{Suffix Invariance \cite{amigo:metric-axioms}.} Given two rankings $\ranking$ and $\ranking'$ that have relevant items in same positions in the top-$\topk$ prefix, a metric is \textit{suffix invariant} if, no matter the positions of the remaining relevant items, the metric values will be the same.  \citet{amigo:metric-axioms} refer this to `Deepness Threshold'.
\item \textbf{Prefix Invariance \cite{amigo:metric-axioms}.} Given two rankings $\ranking$ and $\ranking'$ that have relevant items in same positions in the bottom-$k$ suffix, a metric is \textit{prefix invariant} if, no matter the positions of the remaining relevant items, the metric values will be the same.  \citet{amigo:metric-axioms} refer this to `Closeness Threshold'.
\item \textbf{Boundedness \cite{moffat:seven-properties-of-metrics}.}
A metric is \textit{bounded from above} if there exists an $\metricUpper\in\Re$ such that,
\begin{align*}
\forall\relset\in\docsets, \forall\ranking\in\rankings &: \metric(\ranking,\relset)\leq \metricUpper
\end{align*}
Similarly, a metric is \textit{bounded from below} if there exists an $\metricLower\in\Re$ such that,
\begin{align*}
\forall\relset\in\docsets, \forall\ranking\in\rankings &: \metric(\ranking,\relset)\geq \metricLower
\end{align*}
In general, a metric is \textit{bounded} if it is bounded from above and below.
\item \textbf{Localization \cite{moffat:seven-properties-of-metrics}.}  A metric is \textit{localized} if can be  computed only with the information in the top-$\topk$; in other words, the metric value is suffix invariant \textit{and} independent of the number of unretrieved relevant items.
\item \textbf{Completeness \cite{moffat:seven-properties-of-metrics}.} A metric is \textit{complete} if it is defined when $\numrel=0$ (i.e., metrics are of the form $\metric: \rankings\times\docsetsAll\rightarrow\Re$).
\item \textbf{Realizability \cite{moffat:seven-properties-of-metrics}.}
A metric is \textit{realizable  above} if it is bounded from above and
\begin{align*}
\forall\relset\in\docsets, \exists\ranking\in\rankings&: \metric(\ranking,\relset)= \metricUpper
\end{align*}
A metric is \textit{realizable  below} if it is bounded from below and
\begin{align*}
\forall\relset\in\docsets, \exists\ranking\in\rankings&: \metric(\ranking,\relset)= \metricLower
\end{align*}
If $\numrel>0$, \citet{moffat:seven-properties-of-metrics} refers to metrics that are  realizable above as simply `Realizable'.
\end{enumerate}
In some cases, we have adopted a property name different from the original to help with clarity.  In subsequent sections, we will be demonstrate which of these properties are present for top-heavy recall-level metrics, TSE, and lexirecall.

Several properties were not present in any of our evaluation methods.  None of our methods are strictly decreasing in nonrelevance.  \citet{amigo:metric-axioms} note also that, ``[a]s far as we know, current evaluation measures do not consider this aspect.''  None of our methods are prefix or suffix invariant.  This is largely due to the fact that \begin{inlinelist}
\item exposure is strictly monotonically decreasing, and
\item normalization is a function of recall level (and the number of relevant items).
\end{inlinelist}
As a result, any position-based `flatness' in the computation is missing.  None of our methods are guaranteed to be bounded to allow maximal flexibility in our analysis; this also means that none of our methods are guaranteed to be realizable.  None of our methods are localized because we explicitly use $\numrel$ and $\numdocs$ in pessimistic imputation; while lexirecall does not need $\numdocs$, it does still requite $\numrel$.

We summarize the properties for top-heavy recall-level metrics, total search efficiency, and lexirecall in Table \ref{tab:metric-properties}.

\begin{table}
\caption{Metric Properties.  1: if the second derivative of the exposure function is strictly positive.  2: if normalization function is defined for $\numrel=0$. The properties concavity, boundedness, localization, completeness, and realizability are specific to metric-based evaluation and therefore cannot be analyzed for preference-based evaluation like lexirecall.}\label{tab:metric-properties}
\centering
{\small
\begin{tabular}{lcccc}
\hline
& THRL &  THRL&  $\tse$ & $\lexirecall$  \\
&  &  $\normalization(i,\numrel) > 0$&  & \\
\hline
increasing in retrieval size \cite{moffat:seven-properties-of-metrics}& $\checkmark$ & $\checkmark$ & $\checkmark$ & $\checkmark$ \\
\arrayrulecolor{lightgray}\hline
decreasing in nonrelevance & $\checkmark$ & $\checkmark$ & $\checkmark$ & $\checkmark$\\
\arrayrulecolor{lightgray}\hline
strictly decreasing in nonrelevance \cite{amigo:metric-axioms}&  & & & \\
\arrayrulecolor{lightgray}\hline
increasing in relevance & $\checkmark$ & $\checkmark$ & $\checkmark$ & $\checkmark$\\
\arrayrulecolor{lightgray}\hline
strictly increasing in relevance &  &  $\checkmark$ & & $\checkmark$\\
\arrayrulecolor{lightgray}\hline
increasing in swapping up \cite{ferrante:framework-for-utility-metrics} & $\checkmark$ &  $\checkmark$ & $\checkmark$ & $\checkmark$\\
\arrayrulecolor{lightgray}\hline
strictly  increasing in swapping up \cite{moffat:seven-properties-of-metrics,amigo:metric-axioms} &  &  $\checkmark$ & & \\
\arrayrulecolor{lightgray}\hline
concavity in contiguous swap depth & $\checkmark$ & $\checkmark$  & $\checkmark$ & NA \\
\arrayrulecolor{lightgray}\hline
strict concavity in contiguous swap depth \cite{amigo:metric-axioms} &  & $\checkmark^1$ & & NA\\
\arrayrulecolor{lightgray}\hline
suffix invariance \cite{amigo:metric-axioms} &  &  & & \\
\arrayrulecolor{lightgray}\hline
prefix invariance \cite{amigo:metric-axioms} &  &  & & \\
\arrayrulecolor{lightgray}\hline
boundedness \cite{moffat:seven-properties-of-metrics} &  &  & & NA \\
\arrayrulecolor{lightgray}\hline
localization \cite{moffat:seven-properties-of-metrics} &  &  & & NA \\
\arrayrulecolor{lightgray}\hline
completeness \cite{moffat:seven-properties-of-metrics} & $\checkmark^2$ & $\checkmark^2$ & $\checkmark^2$ & NA \\
\arrayrulecolor{lightgray}\hline
realizability \cite{moffat:seven-properties-of-metrics} &  &  & & NA \\
\arrayrulecolor{black}\hline
\end{tabular}}
\end{table}

\subsection{Properties of Top-Heavy Recall-Level Metrics}
\label{app:proofs:metrics:topheavy}
Let $\metric$ be a  top-heavy recall-level metric with exposure function $\exposure$ and normalization function $\normalization$.

\begin{theorem}
\label{thm:metrics:monotonicity:size}
$\metric$ is monotonically increasing in retrieval size.
\end{theorem}
\begin{proof}
Let $\topkranking'$ be $\topkranking$ with another item appended, with pessimistically imputed rankings $\ranking'$ and $\ranking$.  If the new item is nonrelevant, then $\forall i, \RPx_i=\RPy_i$ and, therefore,  $\metric(\ranking,\relset)=\metric(\ranking',\relset)$ and $\metric$ is trivially non-decreasing.  Next, consider the case where the new item is relevant.  This can only happen if $\topkranking$ includes $\tilde{\numrel}<\numrel$ relevant items.  As such, this is equivalent to swapping a relevant item in the imputed  ranking $\ranking$ from position $\numdocs-(\numrel-\tilde{\numrel}-1)$ to position $\topk+1$.  Let $\Delta_{i,j}\metric(\ranking,\relset)$ be the difference in metric value from swapping a relevant item in position $j$ to position $i$.
\begin{align}
\Delta_{\topk+1,\numdocs-(\numrel-\tilde{\numrel}-1)}\metric(\ranking,\relset)&=\metric(\ranking',\relset)-\metric(\ranking,\relset)\nonumber\\
&=\sum_{i=1}^{\numrel} \exposure(\RPy_i)\normalization(i,\numrel)-\sum_{i=1}^{\numrel} \exposure(\RPx_i)\normalization(i,\numrel)\nonumber\\
&=\exposure(\RPy_{\tilde{\numrel}+1})\normalization(\tilde{\numrel}+1,\numrel)-\exposure(\RPx_{\tilde{\numrel}+1})\normalization(\tilde{\numrel}+1,\numrel)\nonumber\\
&=\normalization(\tilde{\numrel}+1,\numrel)(\exposure(\topk+1)-\exposure(\numdocs-(\numrel-\tilde{\numrel}-1)))\label{eq:metrics:monotonicity:size:rhs}
\end{align}
Since $\normalization(i,\numrel) \geq 0$  and  $\exposure(\topk+1)>\exposure(\numdocs-(\numrel-\tilde{\numrel}-1))$, we know that Equation \ref{eq:metrics:monotonicity:size:rhs} will always be non-negative.
\end{proof}

\begin{theorem}
\label{thm:metrics:monotonicity:nonrel}
$\metric$ is monotonically decreasing in nonrelevance.
\end{theorem}
\begin{proof}
See the first case in the proof Theorem \ref{thm:metrics:monotonicity:size}.
\end{proof}

\begin{theorem}
\label{thm:metrics:monotonicity:rel}
If $\normalization(i,\numrel)$ is (strictly) positive, $\metric$ is (strictly) monotonically increasing in relevance.
\end{theorem}
\begin{proof}
See the second case in the proof Theorem \ref{thm:metrics:monotonicity:size}.  Moreover, if $\normalization(i,\numrel)$ is strictly positive, then Equation \ref{eq:metrics:monotonicity:size:rhs} will always be strictly positive.
\end{proof}

\begin{theorem}
\label{thm:metrics:swap}
If $\normalization(i,\numrel)$ is (strictly) positive, $\metric$ is (strictly) monotonically increasing in swapping up.
\end{theorem}
\begin{proof}
Given a ranking $\ranking\in\rankings$, let $j$ be the $j$th relevant item and $k<\RPx_j$ the position of an arbitrary  nonrelevant item ranked above it.  Let $\ell$ be the recall level of the first relevant item below position $k$ (i.e., $\ell=\min\{i:\RPx_i>k\}$).
\begin{align*}
\Delta_{\RPx_j,k}\metric(\ranking,\relset)&=\normalization(\ell,\numrel)(\exposure(k)-\exposure(\RPx_{\ell}))+
\sum_{i=\ell}^{j-1} \normalization(i+1,\numrel)(\exposure(\RPx_{i})-\exposure(\RPx_{i+1}))
\end{align*}

Since $\normalization(i,\numrel) \geq 0$ and  $\exposure(\RPy_i)>\exposure(\RPx_i)$ for $i\in[\ell,j]$, we know that this difference will always be non-negative.  Moreover, if $\normalization(i,\numrel)$ is strictly positive, then the difference will always be strictly positive.
\end{proof}

\begin{theorem}
\label{thm:metrics:concave}
If $\normalization(i,\numrel)$ strictly positive and the second derivative of $\exposure$ is (strictly) positive, then $\metric$ is (strictly) concave in contiguous swap depth.
\end{theorem}
\begin{proof}
Let $\ranking,\ranking'\in\rankings$ be  two rankings  whose relevant position vectors differ in one element $j$ (i.e., $\forall i\neq j, \RPx_i=\RPy_i$) and $\RPx_j<\RPy_j$.  The metric difference for moving the $j$th relevant item up one positions in $\ranking$ is, 
\begin{align*}
\Delta_{\RPx_j,\RPx_j-1}\metric(\ranking,\relset)&=\normalization(j,\numrel)(\exposure(\RPx_j-1)-\exposure(\RPx_j))
\end{align*}
If the second derivative of $\exposure$ is positive, since $\RPx_j<\RPy_j$,
\begin{align*}
(\exposure(\RPx_j-1)-\exposure(\RPx_j)) &\geq (\exposure(\RPy_j-1)-\exposure(\RPy_j))
\end{align*}
Moreover, because $\normalization(i,\numrel) > 0$,
\begin{align*}
\normalization(j,\numrel)(\exposure(\RPx_j-1)-\exposure(\RPx_j)) &\geq \normalization(j,\numrel)(\exposure(\RPy_j-1)-\exposure(\RPy_j))\\
\Delta_{\RPx_j,\RPx_j-1}\metric(\ranking,\relset) &\geq \Delta_{\RPy_j,\RPy_j-1}\metric(\ranking,\relset)
\end{align*}
Where the inequality is strict if the second derivative of $\exposure$ is strictly positive.
\end{proof}

\begin{theorem}
\label{thm:metrics:completeness}
If $\normalization(i,\numrel)$ is defined for $\numrel=0$, then $\metric$ is complete.
\end{theorem}
\begin{proof}
The only factor in Equation \ref{eq:metric} that depends on $\numrel$ is $\normalization$.  If it is defined for $\numrel=0$, then $\metric$ is defined as well.
\end{proof}

We note that, of the remaining properties, although we cannot prove every top-heavy recall-level metric will satisfy them, there are top-heavy recall-level metrics that do.

\subsection{Properties of Total Search Efficiency}
Let $\metric$ be TSE  with an arbitrary exposure function $\exposure$ and normalization function $\normalization$.  Because TSE is a top-heavy recall-level metric, we know that it satisfies all of the properties in Section \ref{app:proofs:metrics:topheavy} \textit{except} those conditional on $\normalization(i,\numrel) > 0$, since $\normalization_{\esl_3}(i,\numrel) = 0$ when $i<\numrel$.
\subsection{Properties of lexirecall}
\begin{theorem}
\label{thm:lexirecall:monotonicity:size}
lexirecall is monotonically increasing in retrieval size.
\end{theorem}
\begin{proof}
Let $\topkranking'$ be $\topkranking$ with another item appended, with pessimistically imputed rankings $\ranking'$ and $\ranking$.  If the new item is nonrelevant, then $\forall i, \RPx_i=\RPy_i$ and, therefore,  $\metric(\ranking,\relset)=\metric(\ranking',\relset)$ and lexirecall is trivially non-decreasing.  Next, consider the case where the new item is relevant.  This can only happen if $\topkranking$ includes $\tilde{\numrel}<\numrel$ relevant items.  As such, this is equivalent to swapping a relevant item in the imputed  ranking $\ranking$ from position $\numdocs-(\numrel-\tilde{\numrel}-1)$ to position $\topk+1$.  Because of pessimistic imputation, the bottom $\numrel-\tilde{\numrel}-1$ relevant items will be tied.  However, $\RPy_{\numrel+1}=\topk+1$ and $\RPx_{\numrel+1}=\numdocs-(\numrel-\tilde{\numrel}-1)$.
If $\numrel-\tilde{\numrel} < \numdocs-\topk$, then lexirecall will be positive.
\end{proof}

\begin{theorem}
\label{thm:lexirecall:monotonicity:nonrel}
lexirecall is monotonically decreasing in nonrelevance.
\end{theorem}
\begin{proof}
Let $\topkranking'$ be $\topkranking$ with a nonrelevant item appended.  Since the new item is nonrelevant,  $\RPx$ will be unchanged and lexirecall will be the same and is trivially non-increasing.
\end{proof}

\begin{theorem}
\label{thm:lexirecall:monotonicity:rel}
lexirecall is strictly monotonically increasing in relevance.
\end{theorem}
\begin{proof}
See the second case in the proof Theorem \ref{thm:lexirecall:monotonicity:size}.
\end{proof}

\begin{theorem}
\label{thm:lexirecall:swap}
$\metric$ is  monotonically increasing in swapping up.
\end{theorem}
\begin{proof}
Given a ranking $\ranking\in\rankings$, let $j$ be the $j$th relevant item and $k<\RPx_j$ the position of an arbitrary  nonrelevant item ranked above it.  Since $\forall i>j, \RPx_i=\RPy_i$, we only need to compare $\RPx_j$ and $\RPy_j$.  If $k > \RPx_{j-1}$, then, because $k<\RPx_j$, $\ranking'\succ\ranking$.  If $k < \RPx_{j-1}$, then $\RPy_j=\RPx_{j-1}$.  Because $\RPx_{j-1}<\RPx_j$, $\ranking'\succ\ranking$.
\end{proof}

The properties concavity, boundedness, localization, completeness, and realizability are specific to metric-based evaluation and therefore cannot be analyzed for preference-based evaluation like lexirecall.
\section{Robustness}
\label{app:proofs}
In order to demonstrate the relationship between the order of relevant items $\RPx$ and the order of either $\users$ or $\providers$, we first introduce a representation of subsets of positions of relevant items.  Let $\relsubsets=[1\isep \numrel]^+$ be the set of all non-empty sorted lists of integers between $1$ and $\numrel$.  Moreover, let $\relsubsetsGT{i}=[i+1\isep \numrel]^+$ and $\relsubsetsGEQ{i}=[i\isep \numrel]^+$.  This is a way to represent each individual $\user\in\users$, for example, in Figure \ref{fig:users}.  To see how, notice that, because both $\users$ and $\providers$ are also power sets of $\numrel$ distinct integers, there is a one to one correspondence with  $\relsubsets$.  Specifically,
\begin{align*}
\forall \relsubset\in\relsubsets, \user &= \{i \in \relsubset | \rlx_{\RPx_i} \}\\
\forall \user\in\users,  \relsubset &= \sort(\{\rlx_{\RPx_{j}}\in\user|j \})
\end{align*}
and similar for $\providers$.

Given a way to represent each $\user\in\users$, we need to sort these users according to their utility.  We can use our metric definitions $\metric(\rlx,\user)$ and $\pmetric(\rlx,\relset,\provider)$ to define a partial ordering over $\relsubsets$.  This naturally can be represented as a graph where edges reflect that the utility of one user is greater than another.  We present an example based on $\numrel=5$ of the transitive reduction of the partially ordered set for both $\users$ and $\providers$ in Figure \ref{fig:rangesets}.  Although we will present formal proofs, these visualizations help understand the utility structure behind these sets of users.

\begin{figure}
\begin{subfigure}[b]{\linewidth}
\centering
\includegraphics[width=\linewidth]{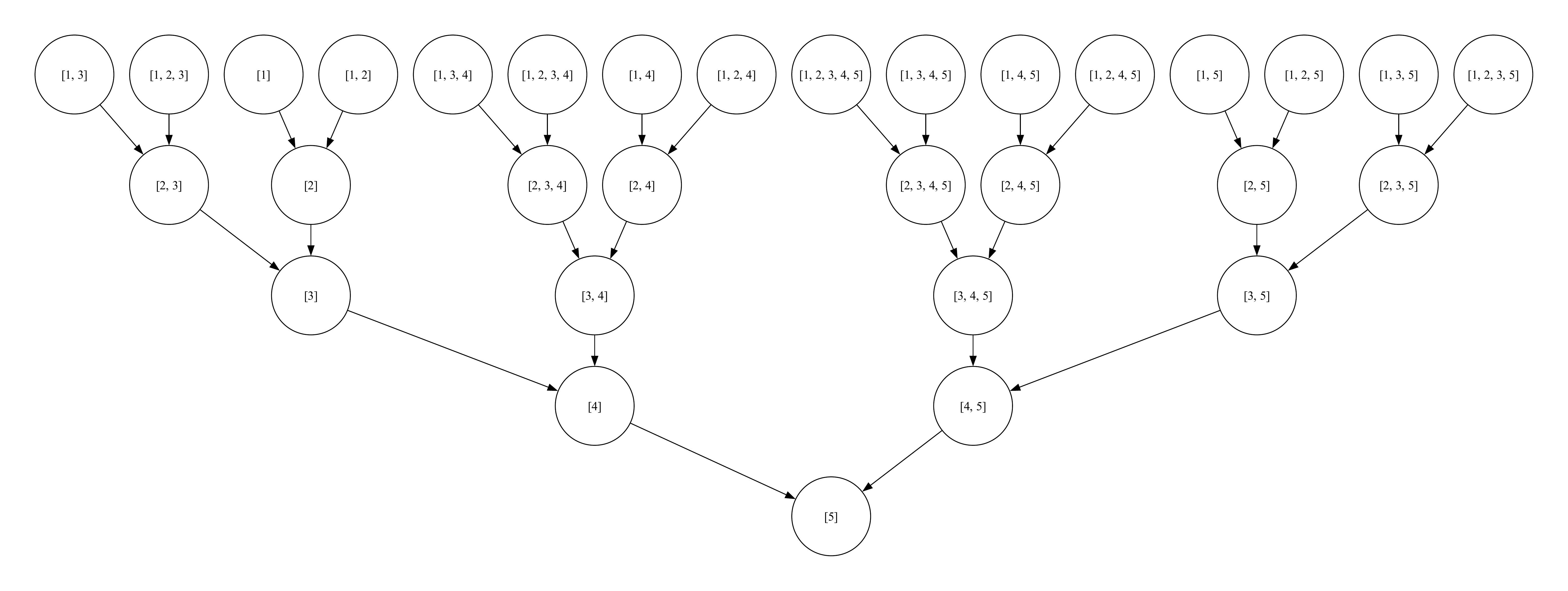}
\caption{Users}\label{fig:rangesets:users}
\end{subfigure}
\begin{subfigure}[b]{\linewidth}
\centering
\includegraphics[width=\linewidth]{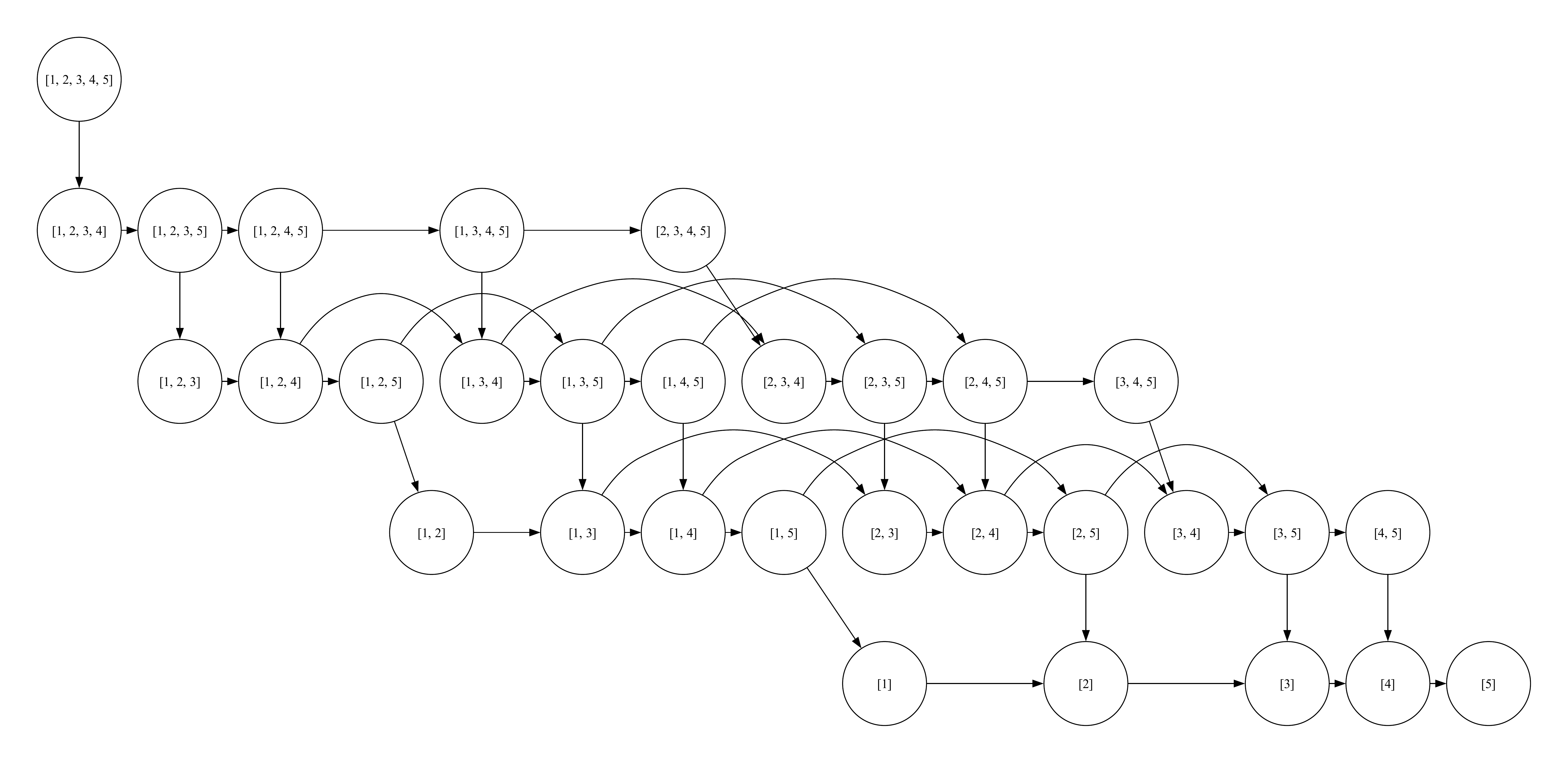}
\caption{Providers}\label{fig:rangesets:providers}
\end{subfigure}
\caption{Transitive reduction of the partially ordered set $\relsubsets$ and $\numrel=5$.  Nodes correspond to elements in $\relsubsets$.  A directed edge from $\relsubset$ to $\relsubset'$ if $\metric(\rlx,b)>\metric(\rlx,b')$ (Figure \ref{fig:rangesets:users}) or  $\pmetric(\rlx,\relset,b)>\pmetric(\rlx,\relset,b')$ (Figure \ref{fig:rangesets:providers}) based on the properties of top-heavy recall-level metrics (Definition \ref{def:top-heavy}).} \label{fig:rangesets}
\end{figure}

\subsection{Worst-case and Total Search Efficiency}
\label{app:proofs:wc}
Notice that, in our example, there is a single, unique minimal element for both $\users$ and $\providers$.  In this section, we will prove that this element will always be associated with the lowest-ranked relevant item, $\tse(\ranking,\relset)$.  Throughout these proofs, we will use abbreviations for top-heavy recall-level metric properties defined in Section \ref{sec:preliminaries}.
\begin{theorem}
\label{thm:worstcase}
If $\metric$ is a top-heavy recall-level metric and $\normalization(1,1) =1$, then
\begin{align*}
\metricworst_{\metric}(\ranking,\relset)&=\tse(\ranking,\relset)
\end{align*}
\end{theorem}
\begin{proof}
We want to show that, for any $\user\in\users$, the performance $\metric(\ranking,\user)$ is greater than or equal to $\tse(\ranking,\relset)$, with equality for the user associated with the lowest-ranked relevant item.
Recall that each $\relsubset\in\relsubsets$ is associated with a $\user\in\users$,
\begin{align*}
\metric(\ranking,\relsubset)&\geq\exposure(\RPx_{\relsubset_{-1}})\normalization(1,1) &\text{top-heaviness}\\
&=\exposure(\RPx_{\relsubset_{-1}})  &\normalization(1,1)=1\\
&\geq\exposure(\RPx_\numrel) &\text{$\RPx_{\relsubset_{-1}} \leq \RPx_\numrel$ and $\exposure(i) > \exposure(i'), \forall i < i'$}\\
&=\tse(\ranking,\relset)
\end{align*}
\end{proof}

\begin{theorem}
\label{thm:worstcase:providers}
If $\pmetric$ is associated with a top-heavy recall-level metric $\metric$, then
\begin{align*}
\metricworst_{\pmetric}(\ranking,\relset)&=\tse(\ranking,\relset)
\end{align*}
\end{theorem}
\begin{proof}
Because all summands of Equation \ref{eq:pmetric} are positive, we know that the minimum $\provider\in\providers$ will correspond to the smallest summand.  Moreover, because of the monotonically decreasing exposure, this is the exposure of the lowest ranked relevant item, which is exactly $\tse(\ranking,\relset)$.
\end{proof}

\subsection{Leximin and Lexicographic Recall}
\label{app:proofs:lexirecall}
Let $\GPx_i=\exposure(\RPx_i)$.  Since $\RPx$ is sorted in increasing order and $\exposure(x)$ is monotonically decreasing in $x$, $\GPx$ is monotonically decreasing. Assuming we measure the performance of a ranking $\ranking$ for a user $\user$ as $\metric(\ranking,\user)$, then we define $\UPx$ to be the $\numusers\times 1$ vector  containing the value of $\metric(\ranking,\user)$ for each $\user\in\users$.   We will assume that $\UPx$, like $\GPx$, is sorted in decreasing order.

\begin{lemma}
\label{lem:min}
If $\GPx\leximinprefthm\GPy$ and $\normalization(1,1) =1$, then the minimum non-tied user is in $\relsubsetsGEQ{k}-\relsubsetsGT{k}$, where $k=\min\{i\in[1\isep\numrel] | \GPx_i\neq\GPy_i\}$.
\end{lemma}
\begin{proof}
Let $k=\min\{i\in[1\isep\numrel] | \GPx_i\neq\GPy_i\}$.  Because each $\relsubset\in\relsubsetsGT{k}$ is comprised of indices $i>k$ and because $\forall i>k, \RPx_i=\RPy_i$, for each associated user $\user$, $\metric(\RPx,\user)=\metric(\RPy,\user)$.  When computing leximin, by the axiom of the Independence of Identical Consequences \cite{barbara:leximin},  we can remove all of the elements from $\UPx$ and $\UPy$ associated with the $\relsubsetsGT{k}$.  This means that the minimum non-tied pair will be in $\relsubsets-\relsubsetsGT{k}$.

Now we need to show that $\relsubsetsNewMin$ is the set of minimal elements of  $\relsubsets-\relsubsetsGT{k}$.      This holds if we can show that for every element $\relsubset\in\relsubsets-\relsubsetsGEQ{k}$, there exists $\relsubset'\in\relsubsetsNewMin$, such that  $\metric(\rlx,\relsubset)\geq\metric(\rlx,\relsubset')$.

If $\relsubset_{-1} \leq k$ (as depicted by the pink nodes Figure \ref{fig:rangesets}), then let $\relsubset'=\{k\}$,
\begin{align*}
\metric(\rlx,\relsubset)
&\geq \exposure(\RPx_{\relsubset_{-1}})\normalization(1,1) &\text{top-heaviness}\\
&= \exposure(\RPx_{\relsubset_{-1}}) &\normalization(1,1) =1\\
&\geq \exposure(\RPx_{k}) &\text{$\RPx_{\relsubset_{-1}} \leq \RPx_{k}$ and $\exposure(i) > \exposure(i'), \forall i < i'$}\\
&=\metric(\rlx,\relsubset')
\end{align*}

If $\relsubset_{-1} > k$ and $k\in \relsubset$ (as depicted by the blue nodes Figure \ref{fig:rangesets}), then we can let $\relsubset'=\{j\in \relsubset: j\geq k\}$.

Let $j=|\relsubset|-|\relsubset'|$,
\begin{align*}
\metric(\rlx,\relsubset)
&\geq\sum_{i=j+1}^{|\relsubset|}\exposure(\RPx_{\relsubset_i})\normalization(i-j,|\relsubset|-j) &\text{top-heaviness}\\
&=\sum_{i=1}^{|\relsubset'|}\exposure(\RPx_{\relsubset'_i})\normalization(i,|\relsubset'|)&\text{definition of $\relsubset'$}\\
&=\metric(\rlx,\relsubset')
\end{align*}

If $\relsubset_{-1} > k$ and $k\not\in \relsubset$ (as depicted by the green nodes Figure \ref{fig:rangesets}), then $\relsubset'=\{k\}\cup\{j\in \relsubset: j>k\}$.  Let $j=|\relsubset|-|\relsubset'|$,
\begin{align*}
\metric(\rlx,\relsubset)
&\geq\sum_{i=j+1}^{|\relsubset|}\exposure(\RPx_{\relsubset_i})\normalization(i-j,|\relsubset|-j)&\text{top-heaviness}\\
&=\exposure(\RPx_{\relsubset_{j+1}})\normalization(1,|\relsubset|-j)+\sum_{i=j+2}^{|\relsubset|}\exposure(\RPx_{\relsubset_i})\normalization(i-j,|\relsubset|-j)\\
&=\exposure(\RPx_{\relsubset_{j+1}})\normalization(1,|\relsubset'|)+\sum_{i=2}^{|\relsubset'|}\exposure(\RPx_{\relsubset'_i})\normalization(i,|\relsubset'|)&\text{definition of $j$ and $\relsubset'$}\\
&>\exposure(\RPx_{\relsubset'_1})\normalization(1,|\relsubset'|)+\sum_{i=2}^{|\relsubset'|}\exposure(\RPx_{\relsubset'_i})\normalization(i,|\relsubset'|)&\text{$\RPx_{\relsubset_{j+1}} < \RPx_k$ and $\exposure(i) > \exposure(i'), \forall i < i'$}\\
&=\sum_{i=1}^{|\relsubset'|}\exposure(\RPx_{\relsubset'_i})\normalization(i,|\relsubset'|)\\
&=\metric(\rlx,\relsubset')
\end{align*}

\end{proof}

\begin{figure}
\begin{subfigure}[b]{\linewidth}
\centering
\includegraphics[width=\linewidth]{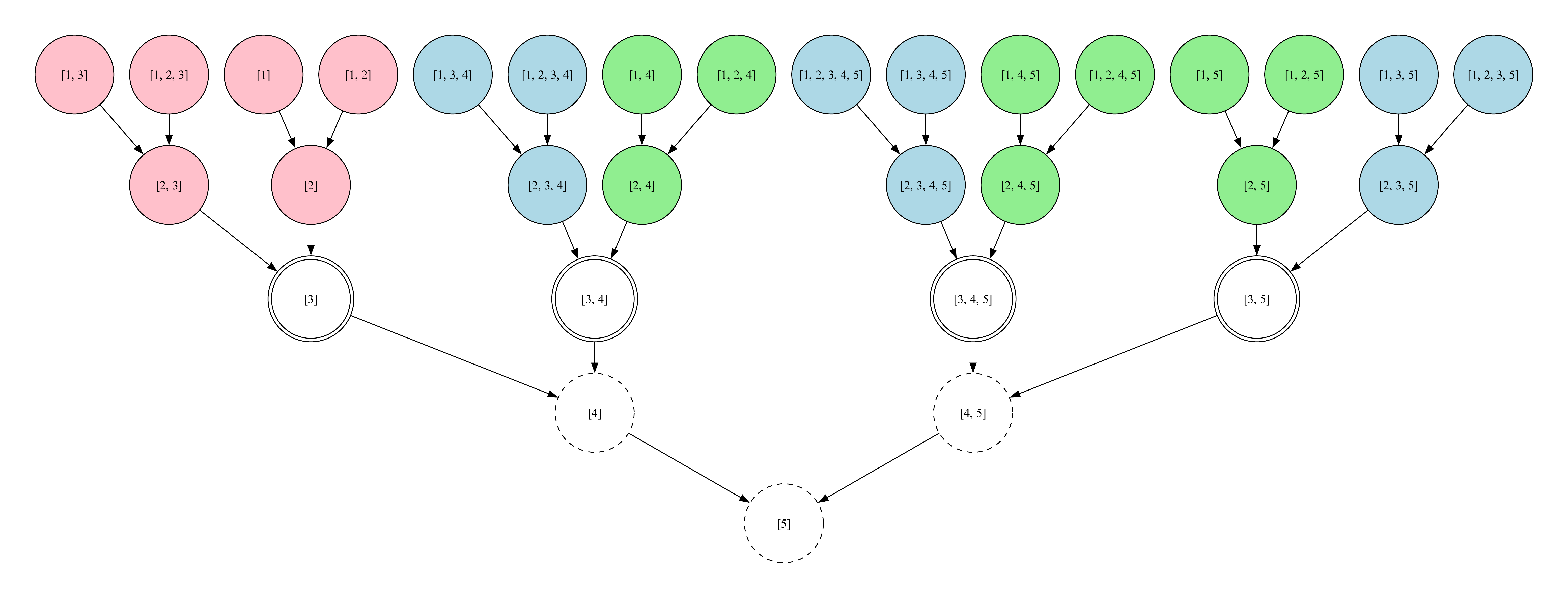}
\caption{Users. Partitioning $\relsubsets$ when $k=3$.   $\relsubsets$ is partitioned into  $\relsubsetsNewMin$  (double circle), $\relsubsetsGT{k}$ (dashed circle), and $\relsubsets-\relsubsetsGEQ{k}$ (colored circles, described in the proof of Lemma \ref{lem:min}). This figure best rendered in color.} \label{fig:leximin-sets}
\end{subfigure}
\vspace{1em}
\begin{subfigure}[b]{\linewidth}
\centering
\includegraphics[width=\linewidth]{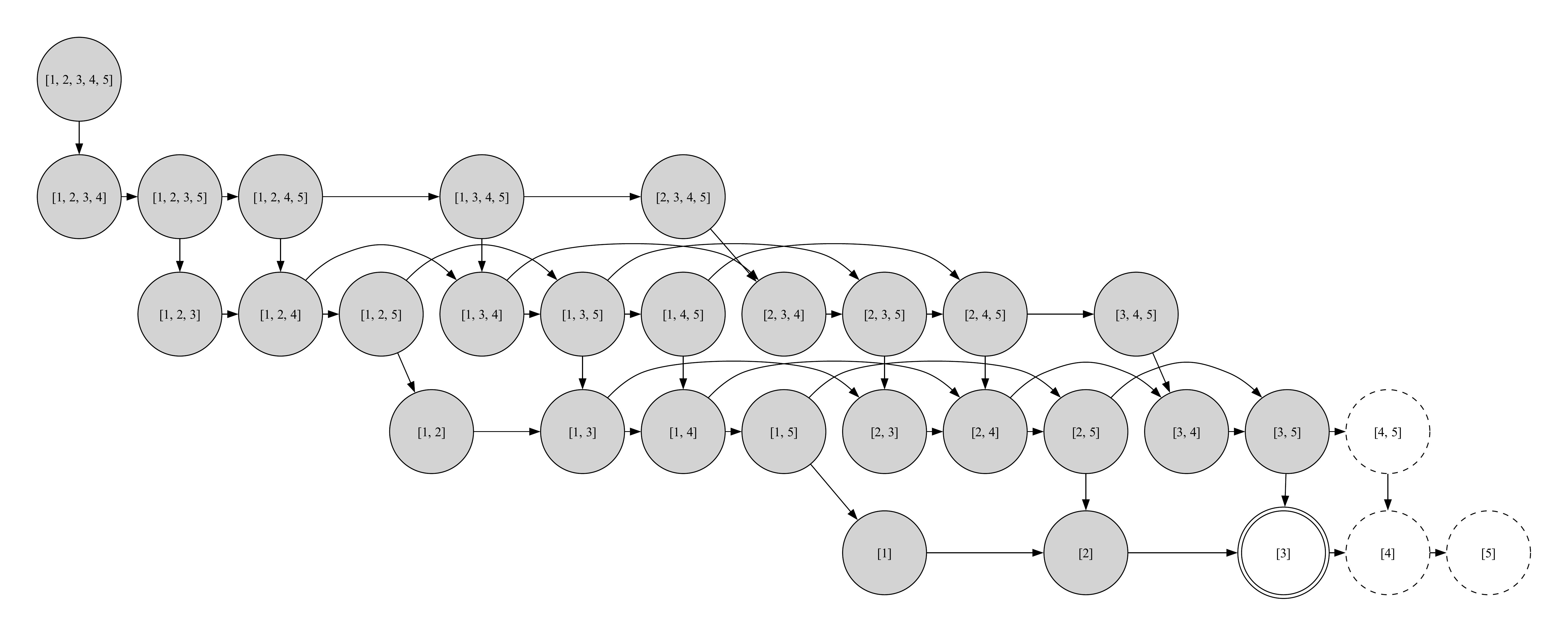}
\caption{Providers. Assuming $k=3$, then the provider with the minimum exposure is $\relsubset=[k]$.}
\end{subfigure}
\caption{Leximin}
\end{figure}

\begin{figure}
\centering
\vspace{3em}
\includegraphics[width=0.35\linewidth]{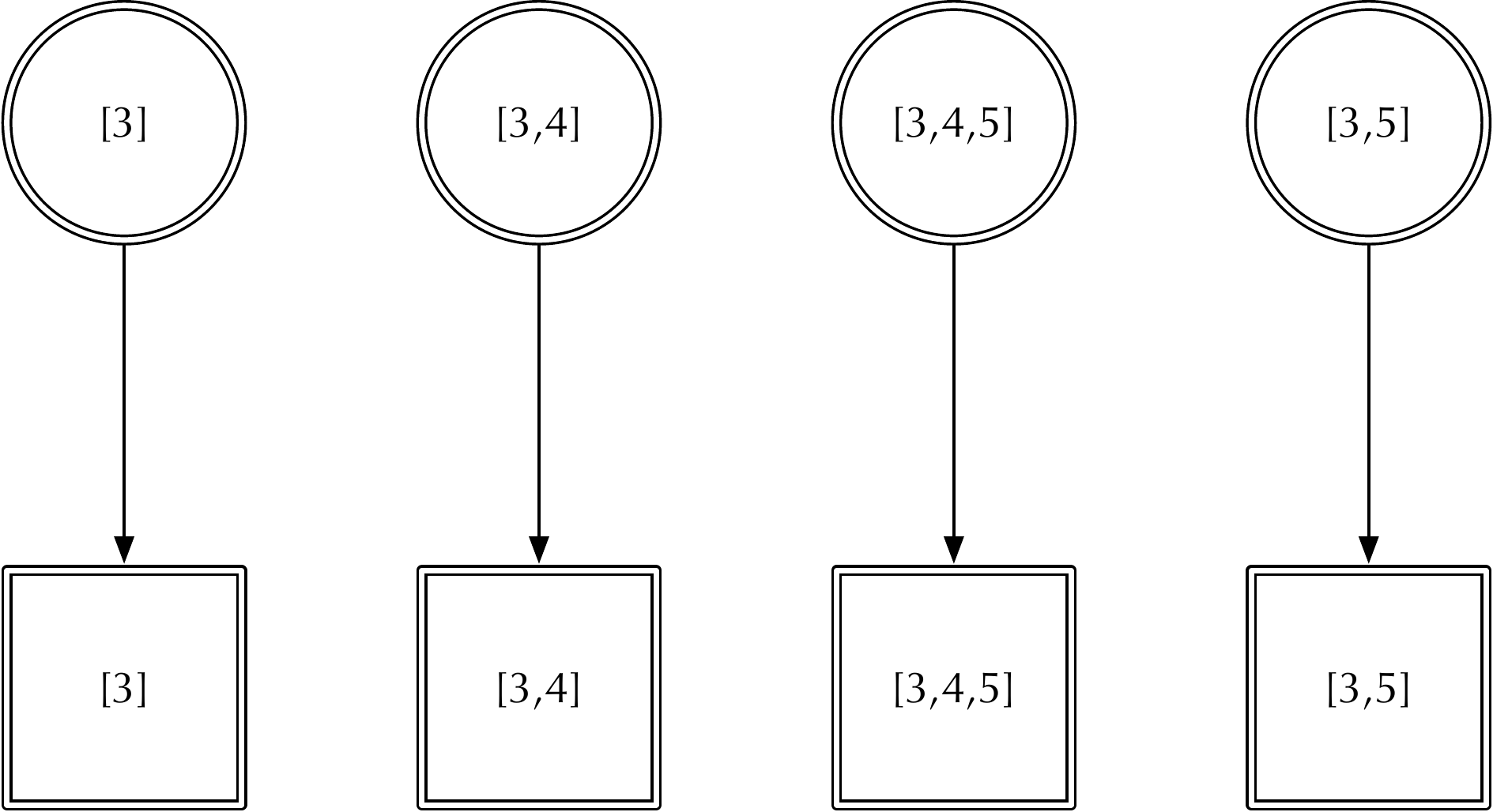}
\caption{Domination of $\relsubsetsNewMin$.  Let $\rlx$ be represented by circles and $\rly$ be represented by squares. } \label{fig:leximin-sets-pref}
\end{figure}

\begin{theorem}
\label{thm:leximin}
\begin{align*}
\GPx\leximinprefthm\GPy\rightarrow \UPx\leximinprefthm\UPy
\end{align*}
\end{theorem}
\begin{proof}
\setcounter{case}{0}
\begin{case}[$\GPx_{\numrel}>\GPy_{\numrel}$]
Theorem \ref{thm:worstcase} means that $\UPx_{\numusers} = \GPx_{\numrel}$ and $\UPy_{\numusers} = \GPy_{\numrel}$ and the implication is true.
\end{case}
\begin{case}[$\GPx_{\numrel}=\GPy_{\numrel}$]
Because $\GPx\leximinpref\GPy$, we know that there exists a $k$ such that $\GPx_k>\GPy_k$ and $\forall i\in[k+1\isep\numrel], \GPx_i=\GPy_i$ (as depicted in Figure \ref{fig:leximin-sets-pref}).  We know from Lemma \ref{lem:min} that the lowest-ranked, non-tied users are in $\relsubsetsNewMin$.

We want to show that, for every element $\relsubset\in\relsubsetsNewMin$, $\rlx$ is preferred to $\rly$.
\begin{align*}
\metric(\rlx,\relsubset)&=\sum_{i=1}^{|\relsubset|}\exposure(\RPx_{\relsubset_i})\normalization(i,|\relsubset|)\\
&=\exposure(\RPx_{k})\normalization(1,|\relsubset|)+\sum_{i=2}^{|\relsubset|}\exposure(\RPx_{\relsubset_i})\normalization(i,|\relsubset|)\\
&=\exposure(\RPx_{k})\normalization(1,|\relsubset|)+\sum_{i=2}^{|\relsubset|}\exposure(\RPy_{\relsubset_i})\normalization(i,|\relsubset|)&\text{$\forall i>k, \RPx_i=\RPy_i$}\\
&>\exposure(\RPy_{k})\normalization(1,|\relsubset|)+\sum_{i=2}^{|\relsubset|}\exposure(\RPy_{\relsubset_i})\normalization(i,|\relsubset|)&\text{$\RPx_k<\RPy_k$ and $\exposure(i) > \exposure(i'), \forall i < i'$}\\
&=\sum_{i=1}^{|\relsubset|}\exposure(\RPy_{\relsubset_i})\normalization(i,|\relsubset|)\\
&=\metric(\rly,\relsubset)
\end{align*}

\end{case}
\end{proof}

\begin{theorem}
\label{thm:leximineq}
\begin{align*}
\GPx\leximineqthm\GPy\rightarrow \UPx\leximineqthm\UPy
\end{align*}
\end{theorem}
\begin{proof}
The only way that $\GPx\leximineq\GPy$ is when $\forall i\in[1\isep\numrel], \GPx_i=\GPy_i$.  If this is the case, then $\forall \user\in\users, \metric(\rlx,\user)=\metric(\rly,\user)$ and, therefore, $\UPx\leximineq\UPy$.
\end{proof}

\begin{theorem}
\label{thm:leximin:provider}
\begin{align*}
\GPx\leximinprefthm\GPy\rightarrow \PPx\leximinprefthm\PPy
\end{align*}
where, for $\provider\in\providers$, $\PPx_{\provider}=\pmetric(\rlx,\relset,\provider)$.
\end{theorem}
\begin{proof}
\setcounter{case}{0}
\begin{case}[$\GPx_{\numrel}>\GPy_{\numrel}$]
We know from the proof of Theorem \ref{thm:worstcase:providers}, the lowest ranked items in $\PPx$ and $\PPy$ correspond to $\exposure(\RPx_{\numrel})=\GPx_{\numrel}$ and $\exposure(\RPy_{\numrel})=\GPy_{\numrel}$, respectively, and so the proof holds in this case.
\end{case}
\begin{case}[$\GPx_{\numrel}=\GPy_{\numrel}$]
Because $\GPx\leximinpref\GPy$, we know that there exists a $k$ such that $\GPx_k>\GPy_k$ and $\forall i\in[k+1,\numrel], \GPx_i=\GPy_i$.  The provider $\relsubset=[k]$ must be the worst-off non-tied element.  Assume that there exists a worse-off non-tied provider.  Because $\relsubset$ is a singleton and because of monotonically decreasing exposure, this worse off provider must be in $\relsubsetsGT{k}$.  However, we know that these providers are all tied and therefore we have a contradiction.
\end{case}
\end{proof}

\section{Optimal Rankings}
\label{app:proofs:opt}
Let $\optimalRankings$ be the set of permutations that rank all of the items in $\relset$ above $\docset-\relset$.  

\begin{theorem}
Given $\dranking\in\optimalRankings$,
\begin{align*}
\min_{\user\in\users}\expectation{\ranking\sim\optimalRankings}{\metric(\ranking,\user)}&\geq\min_{\user\in\users} \metric(\dranking,\user)
\end{align*}
\end{theorem}
\begin{proof}
Let $\wcuser=\argmin_{\user\in\users}\expectation{\ranking\sim\optimalRankings}{\metric(\ranking,\user)}$ be the worst-off user for a stochastic ranking and $\wcranking=\argmin_{\ranking\sim\optimalRankings}\metric(\ranking,\wcuser)$ the worst deterministic ranking in $\optimalRankings$ for $\wcuser$,
\begin{align*}
\expectation{\ranking\sim\optimalRankings}{\metric(\ranking,\wcuser)}&=\frac{1}{\numcollapsed}\sum_{\ranking\sim\optimalRankings}\metric(\ranking,\wcuser)\\
&\geq\metric(\wcranking,\wcuser)\\
&\geq\min_{\user\in\users} \metric(\wcranking,\user)\\
&=\min_{\user\in\users} \metric(\dranking,\user)
\end{align*}
where the final equality follows because  $\wcranking,\dranking\in\optimalRankings$ and, therefore, have isometric distributions of user utility.
\end{proof}

\section{Number of Ties}
\label{app:proofs:numties}

\begin{theorem}
\label{thm:ties:lexirecall}
\begin{align*}
\probthm(\rlx\lexirecalleqthm\rly)&=\frac{\numrel!(\numdocs-\numrel)!}{\numdocs!}
\end{align*}
\end{theorem}
\begin{proof}
If we sample a ranking uniformly from $\rankings$, the probability of any specific $\RPx$ is,
\begin{align*}
\prob(\RPx)&=\frac{\numcollapsed}{|\rankings|}\\
&={\numdocs\choose\numrel}^{-1}
\end{align*}
Let $\RPall$ be the set of all size $\numrel$ samples of unique integers from $[1\isep \numdocs]$.
\begin{align*}
\prob(\rlx\lexirecalleq\rly)&=\prob(\RPx=\RPy)\\
&=\sum_{\RPx,\RPy\in\RPall}\prob(\RPx)\prob(\RPy)\ident(\RPx=\RPy)\\
&=\sum_{\RPx\in\RPall}\prob(\RPx)^2\\
&={\numdocs\choose\numrel}\times\frac{1}{{\numdocs\choose\numrel}}\times\frac{1}{{\numdocs\choose\numrel}}\\
&=\frac{\numrel!(\numdocs-\numrel)!}{\numdocs!}
\end{align*}

\end{proof}
\begin{theorem}
\label{thm:ties:tse}
\begin{align*}
\probthm(\rlx=_{\tse}\rly)&={\numdocs\choose\numrel}^{-2}\sum_{i=\numrel}^\numdocs{i-1\choose\numrel-1}^2
\end{align*}
\end{theorem}
\begin{proof}
First, we will compute the probability of ranking $\rlx$ where the position of the last relevant item is $i$,
\begin{align*}
\prob(\RPx_\numrel=i)
&=\frac{{i-1\choose \numrel-1}}{{\numdocs\choose\numrel}}
\end{align*}
We can use this to compute the probability of a tie,
\begin{align*}
\prob(\rlx=_{\tse}\rly)&=\prob(\RPx_{\numrel}=\RPy_{\numrel})\\
&=\sum_{\RPx,\RPy\in\RPall}\prob(\RPx)\prob(\RPy)\ident(\RPx_\numrel=\RPy_\numrel)\\
&=\sum_{i=\numrel}^\numdocs\prob(\RPx_\numrel=i)^2\\
&=\sum_{i=\numrel}^\numdocs\frac{{i-1\choose \numrel-1}^2}{{\numdocs\choose\numrel}^2}
\end{align*}

\end{proof}

\begin{theorem}
\label{thm:ties:rk}
\begin{align*}
\probthm(\rlx=_{\recall_k}\rly)&={\numdocs\choose\numrel}^{-2}\sum_{i=0}^\numrel{{k\choose i}^2{\numdocs-k \choose \numrel-i}^2}
\end{align*}
\end{theorem}
\begin{proof}
Let $\relat(\RPx,k)=\sum_{i=1}^\numrel \ident(\RPx_i\leq k)$.
First, we will compute the probability of ranking $\rlx$ where $i$ items are ranked above position $k$,
\begin{align*}
\prob\left(\relat(\RPx,k)=i\right)&=\frac{{k\choose i}{\numdocs-k \choose \numrel-i}}{{\numdocs\choose\numrel}}
\end{align*}
We can use this to compute the probability of a tie,
\begin{align*}
\prob(\rlx=_{\text{TSE}}\rly)&=\prob(\relat(\RPx,k)=\relat(\RPy,k))\\
&=\sum_{\RPx,\RPy\in\RPall}\prob(\RPx)\prob(\RPy)\ident(\relat(\RPx,k)=\relat(\RPy,k))\\
&=\sum_{i=0}^\numrel\prob(\relat(\RPx,k)=i)^2\\
&=\sum_{i=0}^\numrel\frac{{k\choose i}^2{\numdocs-k \choose \numrel-i}^2}{{\numdocs\choose\numrel}^2}
\end{align*}

\end{proof}

\begin{theorem}
\label{thm:ties:rp}
\begin{align*}
\probthm(\rlx=_{\rprecision}\rly)&={\numdocs\choose\numrel}^{-2}\sum_{i=0}^\numrel{{\numrel\choose i}^2{\numdocs-\numrel \choose \numrel-i}^2}
\end{align*}
\end{theorem}
\begin{proof}
The proof follows that of Theorem \ref{thm:ties:rk}, substituting $\numrel$ for $k$.
\end{proof}

\end{document}